\documentclass[reqno,oneside,12pt]{amsart}

\usepackage{amssymb,graphicx,braket,amsmath,amsfonts,amsthm,bm,enumitem}
\usepackage[mathscr]{euscript}

\usepackage{hyperref}
\usepackage[utf8]{inputenc}
\usepackage{color}
\usepackage{amsaddr}
\usepackage[T1]{fontenc}
\usepackage[utf8]{inputenc}
\usepackage{tikz}
\usetikzlibrary{shapes,arrows,er,positioning}
\usetikzlibrary{shapes.multipart}
\usepackage{stmaryrd}
\usepackage{fancyhdr}
\usepackage{empheq}
\usepackage{xcolor}
\usepackage{array}
\usepackage{wasysym}
\usetikzlibrary{arrows}
\usetikzlibrary{arrows.meta}
\allowdisplaybreaks

\numberwithin{equation}{section}

\theoremstyle{plain}
\newtheorem{thm}{Theorem}[section]
\newtheorem{cor}[thm]{Corollary}
\newtheorem{lem}[thm]{Lemma}
\newtheorem{prop}[thm]{Proposition}
\newtheorem{conj}[thm]{Conjecture}

\theoremstyle{definition}
\newtheorem{defin}[thm]{Definition}

\newtheorem{rem}[thm]{Remark}

\DeclareMathOperator*{\Res}{Res}

\numberwithin{equation}{section}

\def\mainmatter{\def\baselinestretch{1}\normalfont \setlength{\parskip}{0.5em}}

\usepackage[margin=1in]{geometry}

\newcommand{\bC}{\mathbb{C}}
\newcommand{\bZ}{\mathbb{Z}}
\newcommand{\bP}{\mathbb{P}}

\newcommand{\cS}{\mathcal{S}}
\newcommand{\cP}{\mathcal{P}}

\newcommand{\cA}{\mathcal{A}}
\newcommand{\cB}{\mathcal{B}}
\newcommand{\cR}{\mathcal{R}}

\newcommand{\cH}{\mathcal{H}}

\newcommand{\sQ}{\mathscr{Q}}

\newcommand{\sA}{\mathscr{A}}
\newcommand{\sB}{\mathscr{B}}
\newcommand{\sI}{\mathscr{I}}

\makeatletter
\renewcommand{\section}{\@startsection
{section}%                   % the name
{1}%                         % the level
{\z@}%                       % the indent / 0mm
{-\baselineskip}%            % the before skip / -3.5ex \@plus -1ex \@minus -.2ex
{0.8\baselineskip}%          % the after skip / 2.3ex \@plus .2ex
{\centering\scshape\large}} % the style

\renewcommand{\subsection}{\@startsection
{subsection}%                   % the name
{2}%                         % the level
{\z@}%                       % the indent / 0mm
{-0.8\baselineskip}%            % the before skip / -3.5ex \@plus -1ex \@minus -.2ex
{0.5\baselineskip}%          % the after skip / 2.3ex \@plus .2ex
{\normalfont \bf \normalsize}} % the style

\renewcommand{\subsubsection}{\@startsection
{subsubsection}%                   % the name
{3}%                         % the level
{\z@}%                       % the indent / 0mm
{-0.8\baselineskip}%            % the before skip / -3.5ex \@plus -1ex \@minus -.2ex
{0.5\baselineskip}%          % the after skip / 2.3ex \@plus .2ex
{\normalfont \bf \normalsize}} % the style
\makeatother

\usepackage[backend=bibtex,style=alphabetic,maxbibnames=10,sorting=nyt,giveninits]{biblatex}
\renewbibmacro{in:}{}
\begin{document}
\title{Refined Topological Recursion Revisited \\
--- properties and conjectures}

\author{Kento Osuga}
\address{Graduate School of Mathematical Sciences, University of Tokyo,\\
3-8-1 Komaba, Meguro, Tokyo, 153-8914, Japan}
\email{osuga@ms.u-tokyo.ac.jp}

\begin{abstract}
For any (possibly singular) hyperelliptic curve, we give the definition of a hyperelliptic refined spectral curve and the hyperelliptic refined topological recursion, generalising the formulation for a special class of genus-zero curves by Kidwai and the author, and also improving the proposal by Chekhov and Eynard. Along the way, we uncover a fundamental geometric structure underlying the hyperelliptic refined topological recursion and investigate its properties --- parts of which remain conjectural due to computational difficulties. Moreover, we establish a new recursion valid in the so-called Nekrasov-Shatashivili limit and prove existence of the corresponding quantum curve.

\end{abstract}

\newpage\maketitle
%\vspace{-5mm}
\setcounter{tocdepth}{1}\tableofcontents\mainmatter

\newpage
\section{Introduction}
The purpose of the present paper is two-fold. One is to introduce the notion of the hyperelliptic refined topological recursion for any (possibly singular) hyperelliptic curve, generalising the work of Kidwai and the author which only considers a special class of genus-zero curves \cite{KO22}, and also clarifying subtleties that are not addressed in the Chekhov-Eynard proposal \cite{CE06}. In the process, we uncover a fundamental geometric structure underlying the hyperelliptic refined topological recursion which we call hyperelliptic refined loop equations. The other purpose is to define a new recursion valid in the so-called Nekrasov-Shatashivili limit and prove existence of the corresponding quantum curve. Some properties remain conjectural, and we perhaps need a new approach to prove them --- we will address why computational techniques utilised in the present paper, or in e.g. \cite{EO07,KO22}, are not easy to apply in the refined setting on a higher-genus curve.

\emph{Topological recursion} has become an important research subject in both mathematics and physics. Although it originated from the study of matrix models \cite{CE05}, the formalism has been mathematically developed with more generalities in \cite{EO07,EO08,BE12,BS15}, and therefore, topological recursion is now applicable to a variety of enumerative problems far beyond matrix models. The list of applications includes semi-simple cohomological field theory invariants \cite{NS11,DOSS12,Mil12}, and notably the remodeling theorem  \cite{BKMP07,BKMP08,EO12,FLZ16} for Gromov-Witten invariants of toric Calabi-Yau threefolds. Recently, an algebraic reformulation of topological recursion has been intensively studied with a new notion called \emph{Airy structures} \cite{KS17,ABCO17,BBCCN18,BCHORS19}, with the presence of a refinement parameter \cite{BBCC21,O21}. Topological recursion has also been utilised to construct so-called \emph{quantum curves} which are closely related to isomonodromy systems \cite{GS11,BE16,Iwa19}. There are many more influential works in topological recursion which we are not able to list due to length constraints, and we refer to the readers e.g. \cite{E14} and references therein for more achievements of topological recursion and related subjects.

Despite the fast and rich development of topological recursion, however, a \emph{refinement} of topological recursion has been a long-standing problem. The first attempt was carried out by Chekhov and Eynard \cite{CE06} in line with $\beta$-deformed matrix models, and they proposed the refined recursion formula. Shortly after, however, Chekhov reported in \cite{C10} a necessity of modifications to the Chekhov-Eynard formula because Mari\~{n}o and Pasquetti observed an inconsistent behaviour when the proposed refined formula was applied to a genus-one curve (see \cite[Section 1]{C10}). Since then, there have been several related attempts such as \cite{CEM10,BMS10,M12,MS15} in order to properly refine the formalism of topological recursion, but to author's best knowledge, none of those proposed recursions can be applied beyond $\beta$-deformed matrix models, or not even to $\beta$-deformed matrix models unless the associated spectral curve is very simple. 

Our aim is to break through this phase and to geometrically define a hyperelliptic refined spectral curve on which we define the \emph{hyperelliptic refined topological recursion}. For a special class of genus-zero curves, Kidwai and the author have initiated this direction and written down a refined recursion formula in \cite{KO22} (strongly inspired by the Chekhov-Eynard formula), with applications to refined quantum curves. The present paper extends their approach with more generalities and with more details. Importantly, we study a sequence of equations among multidifferentials which turn out to be a fundamental geometric structure underlying the refinement. Another achievement of the present paper is to establish a new recursion relevant in the Nekrasov-Shatashivili limit, to which the Chekhov-Eynard-Orantin topological recursion \cite{CE05,CEO06,EO07} has no access.

This paper is organised as follows. In Section \ref{sec:Def}, we give the definition of a hyperelliptic refined spectral curve, hyperelliptic refined loop equations, and the hyperelliptic refined topological recursion. We also compare our formula with the proposed one by Chekhov and Eynard. Then in Section \ref{sec:properties}, we investigate properties of the hyperelliptic refined topological recursion in detail, where some of the properties remain conjectural due to computational difficulties. In Section \ref{sec:NS}, we define a new recursion relevant to the Nekrasov-Shatashivili limit, and prove existence of the corresponding quantum curve. We comment on technical computations and interesting observations in Appendix \ref{sec:appendix}.

\subsection{Summary of Results}

Formalising a refinement of topological recursion itself is one of the achievements of the present paper. In the process, the most important step is to geometrically describe hyperelliptic refined loop equations (Definition \ref{def:GLE}) which are a fundamental structure underlying the hyperelliptic refined topological recursion. Although it is an essential geometric requirement, from now on we often drop ``hyperelliptic'' for brevity. Understanding the geometry of refined loop equations directly connects to properties of the refined topological recursion. Moreover, it also gives us a hint of how to refine the higher-ramified Bouchard-Eynard recursion \cite{BE12}. Since the whole story involves a few new notions, we leave all the details to Section \ref{sec:Def}. Nevertheless, let us give an intuitive summary of our main results, with a remark that theorems stated in this section are paraphrases of formal statements in the body. 

Roughly speaking, our proposal of a refinement of topological recursion goes as follows:
\begin{description}
\item[Input] We call initial data of the refined topological recursion a \emph{hyperelliptic refined spectral curve} (Definition \ref{def:RSC}) which we denote by $\cS_{{\bm\mu},{\bm\kappa}}$. It is a (compactified and normalised) hyperelliptic curve $(\Sigma,x,y)$ with a Torelli marking, i.e. a choice of canonical basis $\cA_i,\cB_i\in H_1(\Sigma,\bZ)$ for $i\in\{1,..,\tilde g\}$, together with new parameters ${\bm\mu},{\bm\kappa}$ that appear only in the refined setting.
\item[Equation] Given $\cS_{{\bm\mu},{\bm\kappa}}$, we define a sequence of equations among multidifferentials which depend on the refinement parameter $\sQ$, and we call them \emph{refined loop equations} (Definition \ref{def:GLE}). By assuming existence of a solution of refined loop equations, we will prove that its unique solution can be recursively constructed by a certain contour integral formula (\eqref{RTR1} in Theorem \ref{thm:RTR}).
\item[Recursion] We then change our perspective and take the derived recursion formula as the defining equation for multidifferentials $\omega_{g,n}$ for $2g,n\in\bZ_{\geq0}$. The formalism is recursive in $2g-2+n$, hence we call it the \emph{hyperelliptic refined topological recursion} (Definition \ref{def:RTR}). $\omega_{g,n}$ depends polynomially on the refinement parameter $\sQ$, up to degree $2g$, thus one may expand as:
\begin{equation}
    \omega_{g,n}=\sum_{k=0}^{2g}\sQ^k\omega_{g,n}^{(k)}.
\end{equation}
\end{description}

Our formalism is purely geometric without referring to $\beta$-deformed matrix models in physics. Moreover, we have new degrees of freedom (${\bm\mu},{\bm\kappa}$) in the definition of a refined spectral curve that do not appear in the matrix model setting. Existence of a solution of refined loop equations remains to be proven when $\Sigma\neq\bP^1$, and we hope to return to this problem in the near future. We mostly focus on hyperelliptic curves, but a possible extension to higher-ramified curves will be discussed in Section \ref{sec:comments}.

The refined formulae in \cite{KO22} or \cite{CE06} look different from the recursion formula in the present paper (Definition \ref{def:RTR}), but one can show that they are equivalent for the class of genus-zero curves \cite{KO22} considers. However, we remark that there are subtleties in the proposed formula by Chekhov and Eynard \cite{CE06} when $\Sigma\neq\bP^1$, due to contributions from nontrivial $H_1(\Sigma,\bZ)$, and their recursion formula does not solve refined loop equations. Comparisons between the two recursions will be carefully discussed in Section \ref{sec:CE}. We note that multidifferentials in the unrefined sector of our formalism, i.e. $\omega_{g,n}^{(0)}$, obeys the Chekhov-Eynard-Orantin recursion as expected (Proposition \ref{prop:Q=0}).

$\omega_{g,n}$ for $2g+n-2>0$ satisfy several interesting properties. Even though some of them are similar to unrefined analogues, their proofs are more involved and one encounters difficulties if one only repeats the same strategy as e.g. \cite{EO07} or \cite{KO22}. This may suggest a necessity of new techniques to prove all the statements. For now, let us list important ones:
\begin{thm}[{Proposition \ref{prop:1/2,2} and Theorem \ref{thm:existence}}]\label{thm:intro1}
\hfill

\begin{itemize}
     \item[{\bf (A):}] For any refined spectral curve, there exists a unique solution of the refined loop equation of type $(g,n)$ when $2g-2+n=1$, and its solution $\omega_{g,n}$  is constructed by the refined topological recursion. In particular, $\omega_{\frac12,2}$ is a symmetric bidiffierential. 
    \item[{\bf (B):}] When $\Sigma=\bP^1$, there exists a unique solution of refined loop equations for all $2g-2+n>0$, and its solution $\omega_{g,n}$ is constructed by the refined topological recursion
\end{itemize}
    \end{thm}

\begin{conj}[Conjecture \ref{conj:existence}]\label{conj:intro4}
{\rm Theorem \ref{thm:intro1} (B)} holds even when $\Sigma\neq\mathbb{P}^1$.
\end{conj}

The recursion formula for $\omega_{g,n}^{(0)}$ is self-closed, i.e., $\omega_{g,n}^{(0)}$ can be determined without knowing $\omega_{g,n}^{(k)}$ for any $k>0$, and this is not a surprise. Theorem \ref{thm:intro1} (A) shows that we have fixed the issue reported in \cite{C10}, providing evidence that the refined topological recursion formula of the present paper is valid. Recall that $\omega_{g,n}$ polynomially depends on $\sQ$ up to degree $2g$. Now a critical observation is that the recursion formula for $\omega_{g,n}^{(2g)}$ is also self-closed, i.e. no information about $\omega_{g,n}^{(k)}$ for $k<2g$ is necessary to define $\omega_{g,n}^{(2g)}$ recursively. We assign a special symbol to the top $\sQ$-degree term as $\varpi_{g,n}:=\omega_{g,n}^{(2g)}$ and call the self-closed recursion for $\varpi_{g,n}$ the \emph{$\sQ$-top recursion} (Definition \ref{def:NS}). It turns out that the $\sQ$-top recursion is relevant to the Nekrasov-Shatashivili limit -- in contrast, the unrefined sector is often referred to as the self-dual limit.

Moreover, another crucial observation is that the $\sQ$-top recursion is no longer a recursion in $2g-2+n$ but rather it is recursive in $g$ and $n$ separately. Therefore, it is not quite a `topological' recursion and in fact we intentionally removed `topological' from the name. In particular, the recursion for $\varpi_{g,1}$ does not involve $\varpi_{g,n+2}$. This phenomenon makes it possible for us to prove existence of the $\sQ$-top part of the refined loop equation of type $(g,1)$ for all $2g\in\bZ_{\geq0}$, hence Conjecture \ref{conj:intro4} will be upgraded to a theorem for $\varpi_{g,1}$ (it still remains conjectural for other $\varpi_{g,n+2}$). Furthermore, one can show that $\varpi_{g,1}$ satisfies the Riccati-type recursion in the context of the WKB analysis, implying that there exists an associated 2nd-order ordinary differential equation. Although technical details are only addressed in Section \ref{sec:NS}, let us state it as a theorem to emphasise:
\begin{thm}[Theorem \ref{thm:NS}]\label{thm:intro5}
Given $\cS_{\bm\mu,\bm\kappa}$, let $\varpi_{g,1}$ be differentials constructed by the $\sQ$-top recursion. For an appropriate choice of an open contour $\gamma\subset\Sigma$, there exists a formal 2nd-order ordinary differential equation of the following form (up to gauge transformation and overall rescaling):
\begin{equation}
\left(\epsilon_1^2\frac{d^2}{dx^2}+\bar Q^{\sQ\text{-{\rm top}}}(x;\epsilon_1)\right)\psi^{\sQ\text{-{\rm top}}}_{\gamma}(x;\epsilon_1)=0,
\end{equation}
where $\bar Q^{\sQ\text{-{\rm top}}}\in\bC(x)\llbracket\epsilon_1\rrbracket$, and $\psi^{\sQ\text{-{\rm top}}}_{\gamma}(x)$ is defined as a formal object in terms of a local inverse function of $x:\Sigma\to\bP^1$, by:
\begin{equation}
\psi^{\sQ\text{-{\rm top}}}_{\gamma}(x;\epsilon_1):=\exp\left(\sum_{g\in\frac12\bZ_{\geq0}}\epsilon_1^{2g-1}\int_{\gamma}^{x}\omega_{g,1}^{\sQ\text{-{\rm top}}}\right).
\end{equation}
\end{thm}

Theorem \ref{thm:intro5} means that the $\sQ$-top recursion can be utilised to quantise the refined spectral curve $\cS_{\bm\mu,\bm\kappa}$, and we call the resulting ordinary differential operator the \emph{$\sQ$-top quantum curve}. Let us remark that existence of a unrefined quantum curve was difficult to prove when $\Sigma\neq\bP^1$, and it has been only recently proven in \cite{Iwa19,EGF19,MO19} (for hyperelliptic curves) by incorporating a formal transseries structure. In contrast, quantisation via the $\sQ$-top recursion requires no such complicated structure. Importantly, we will provide a geometric interpretation of the $\sQ$-top quantisation procedure in Section \ref{sec:interpretation}.

As a final remark, quantisation with the full refined setting stands as an important, curious, and challenging open question.

\subsection{Several Future Directions}
\begin{itemize}
    \item As an application to enumerative geometry, it is known that monotone Hurwitz numbers are computed by the unrefined topological recursion \cite{DDM14}. Recently, $b$-monotone Hurwitz numbers were investigated in \cite{CD20,BCD21} in terms of constellation and generalised branch coverings. From a different perspective,  \cite{BBCC21} considered a recursive structure of Whittaker vectors in relation to $\beta$-deformed Virasoro algebra. The precise statement as well as the proof of the triangular relation between $b$-monotone Hurwitz numbers, Whittaker vectors and the refined topological recursion are work in progress joint with N. Chidambaram and M. Do\l \k{e}ga.
     \item  For a special class of genus-zero curves, Iwaki and collaborators showed that the free energy of the unrefined topological recursion is decomposed in terms of so-called \emph{BPS structures} \cite{IKT18-1,IKT18-2,IK20, IK21}. An interesting question is whether this story admits a refinement via the refined topological recursion, and whether \emph{refined BPS structures} can be determined. Kidwai and the author will report this result soon. 
     \item As briefly mentioned in the previous subsection, it is an open question how to construct a quantum curve via the refined topological recursion without taking any limit, i.e., full refinement of the results of \cite{Iwa19,EGF19,MO19}. Kidwai and the author \cite{KO22} constructed the refined quantum curve for a special class of genus-zero curves, but complexity of computations drastically increases for higher genus curves in the refined setting. This is due to a more complicated pole structure of $\omega_{g,n}$, and it requires a deeper understanding of refined loop equations to overcome the difficulties.
     \item An interesting and perhaps surprising observation is that the so-called variational formula (c.f. \cite[Section 5]{EO07}) sometimes does \emph{not} hold in the most general refined setting. For example, it can be shown that it still holds for hypergeometric curves that \cite{IKT18-2,IK20} consider. However, it does not hold for the zero-parameter Painlev\'{e}-I curve discussed in \cite{IS15}, unless one fine-tunes new parameters $\bm\mu$. Note that $\Sigma=\bP^1$ in this case, hence this phenomenon happens even for a genus-zero refined spectral curve. As one can see from these observations, one needs careful considerations to even state what the variational formula is in the refined setting. Since it is beyond the main scope of the present paper, the author will discuss it in a separate work \cite{O23}.
     \item Related to the above observation, one may wonder whether there is a relation between $F_g^{\sQ\text{-{\rm top}}}$ and $Q^{\sQ\text{-{\rm top}}}(x;\epsilon)$ in Theorem \ref{thm:intro5} as an analogous result of \cite{Iwa19,EGF19,MO19}. The author has already checked that they are in fact related at least for a certain class of curves (e.g. the curve considered in \cite{IS15}, after fine-tuning $\bm\mu$ to a specific value) which we will report in a separate work \cite{O23}. This aspect will connect the $\sQ$-top recursion in the present paper with the so-called twisted Nekrasov superpotential \cite{NS09}. 
    \item Computations of the unrefined topological recursion formula can be ``localised'' at contributions from ramification points. On the other hand, the global hyperelliptic involution operator plays an important role in the refined topological recursion (Definition \ref{def:RTR}), and the localisation is no longer possible. As a consequence, we cannot apply a clever approach in terms of Airy structures \cite{KS17,ABCO17,BBCCN18} to prove properties of $\omega_{g,n}$. This stands as a major open problem towards a complete refinement of topological recursion, but we hope that the present paper sheds light on future investigation.
\end{itemize}

\subsubsection*{Acknowledgement}

The author thanks Leonid Chekhov, Nitin Chidambaram, Maciej Do\l\k{e}ga, Lotte Hollands, Kohei Iwaki, Omar Kidwai, Masahide Manabe, Nicolas Orantin for discussions and correspondences. The author particularly appreciate beneficial discussions with Omar Kidwai. This work is supported by a JSPS Grant-in-Aid for JSPS Fellows, KAKENHI Grant Number: 22J00102. This work is also in part supported by JSPS KAKENHI 20K14323, 21H04994, and 23K12968.

\newpage

\section{Definitions}\label{sec:Def}
In this section we define geometrically, a hyperelliptic refined spectral curve, hyperelliptic refined loop equations, and the hyperelliptic refined topological recursion. Although these definitions are inspired by the work of $\beta$-deformed matrix models (e.g. \cite{CE06}), we note that they are purely geometric, without referring to such physical models. We also discuss the differences between our approach and that of \cite{CE06} in details.

\subsection{Refined Spectral Curve}\label{sec:RSC}
Let $(\Sigma,x,y)$ be a compactified and normalised hyperelliptic curve of genus $\tilde g$\footnote{Strictly speaking, $(\Sigma,x,y)$ is said to be a hyperelliptic curve if genus $\tilde g$ is greater than 1. However, we abuse the notation and call it hyperelliptic regardless of the genus.}. That is, $x,y$ are meromorphic functions on a compact Riemann surface $\Sigma$ of genus $\tilde g$ such that they satisfy an irreducible algebraic equation of the form 
\begin{equation}
    P(x,y)=a(x) y^2+b(x) y+c(x)=0,\label{P(x,y)=0}
\end{equation}
where $a(x),b(x),c(x)\in\bC[x]$.

Any hyperelliptic curve is equipped with a global involution $\sigma:\Sigma\to\Sigma$ such that for a generic $p\in\Sigma$,  $x(\sigma(p))=x(p)$ but $y(\sigma(p))\neq y(p)$. That is, $\sigma$ is the hyperbolic involution of the double cover $x:\Sigma\to\bP$. We denote by $\cR$ the set of $\sigma$-fixed points on $\Sigma$. $\cR$ can be equivalently defined as the set of simple zeroes and third-order poles of $dx$. We call $r\in\cR$ \emph{ineffective} if $\Delta y\cdot dx$\footnote{For any function $f$ on $\Sigma$, $\Delta f(p):=f(p)-f(\sigma(p))$. We will use the same notation for differentials too.} is singular at $r$, and \emph{effective} otherwise. We denote by $\cR^*\subset\cR$ the set of effective ramification points.

There are $2\tilde g$ independent basis in $H_1(\Sigma,\bZ)$, and we denote a canonical basis by $(\cA_i,\cB_i)$ where for $i,j\in\{1,...,\tilde g\}$ they satisfy:
\begin{equation}
    \cA_i\circ\cA_j=\cB_i\circ\cB_j=0,\quad \cA_i\circ\cB_j=\delta_{ij},\label{ABcycle}
\end{equation}
where $\circ$ is the symplectic pairing on $H_1(\Sigma,\bZ)$. Note that a choice of a canonical basis is not unique\footnote{For instance, both $(\cA_i,\cB_i)$ and $(\tilde \cA_i,\tilde \cB_i):=(\cB_i,-\cA_i))$ satisfy \eqref{ABcycle}.}, and its choice is often called a Torelli marking. As a part of the initial data in the refined setting, we associate a complex parameter $\kappa_i\in\bC$ to each pair $(\cA_i,\cB_i)$ where $i\in\{1,...,\tilde g\}$.

Let $\widetilde \cP^{(0,\infty)}$ be the set of zeroes and poles of $\Delta y\cdot dx$ that are not in $\cR$ respectively, and we consider $\widetilde\cP=\widetilde \cP^{(0)}\cup\widetilde \cP^{(\infty)}$. One can decompose it into $\widetilde \cP=\widetilde \cP_+\sqcup\sigma(\widetilde \cP_+)$ such that $\sigma(p)\in\sigma(\widetilde \cP_+)$ for every $p\in\widetilde \cP_+$. This decomposition is not unique. We associate a complex parameter $\mu_p\in\bC$ for each point $p\in\widetilde \cP_+$ which will be another degree of freedom for a refined spectral curve.

A hyperelliptic refined spectral curve is a Torelli marked hyperelliptic curve, together with a choice of non-unique objects above. It will be the initial data for the hyperelliptic refined topological recursion:

\begin{defin}\label{def:RSC}
A \emph{hyperelliptic refined spectral curve} $\cS_{{\bm \mu}, {\bm \kappa}}$ consists of the collection of the following data:
\begin{itemize}
    \item $(\Sigma,x,y)$: a (compactified, normalised) hyperelliptic curve,
    \item $(\cA_i,\cB_i,\kappa_i)$: a choice of a canonical basis and associated parameters for $i\in\{1,..,\tilde g\}$,
    \item $(\widetilde\cP_+,\mu_p)$: a choice of $\widetilde\cP_+\subset\widetilde\cP$ and associated complex parameters $\mu_p$ for all $p\in\widetilde\cP_+$.
\end{itemize}
\end{defin}
Note that $\widetilde\cP$ in \cite{KO22} denotes the set of zeroes and poles of $\Delta y$ that are not in $\cR$, instead of $\Delta y\cdot dx$. Their difference appears when $p\in\Sigma$ is simultaneously a double pole of $\Delta y\cdot dx$ and $dx$. Definition \ref{def:RSC} slightly generalises this aspect, and allows such points to be in $\widetilde\cP$.

\subsubsection{On the interpretation of $\mu_p$ and $\kappa_i$}

The parameters $\mu_p$ and $\kappa_i$ do not appear in the matrix model approach of Chekhov-Eynard \cite{CE06}, and they are purely geometric degrees of freedom. 

$\mu_p$ and $\kappa_i$ can be thought of as cousins in some cases. This is perhaps easier to address by looking at a concrete example. Consider a smooth elliptic curve whose defining algebraic equation is
\begin{equation}
    y^2-(x-a)(x-b)(x-c)=0.
\end{equation}
The branch points of this curve are at $x=a,b,c,\infty$. $\widetilde\cP$ is empty because the only pole of $ydx$ corresponds to the branch point $x=\infty$. Thus, there is no $\mu_p$ assigned for this curve. On the other hand, since this is a genus-one curve, there is one $\kappa$ associated with $\cA\in H_1(\Sigma,\bZ)$. 

Let us now consider the limit $c\to b$ which gives a singular curve. The branch points are only at $x=a,\infty$ and $x=b$ becomes a singular point. $ydx$ has a simple zero at two preimages of $x=b$, hence $\widetilde\cP$ is no longer empty. On the other hand, it becomes a genus-zero curve after normalisation, hence $H_1(\Sigma,\bZ)$ is trivial and there exists no $\kappa_i$. Therefore, one can interpret $\mu_p$ as a replacement of $\kappa$ when one considers the singular limit $c\to b$.

\subsubsection{Multidifferentials}

We will mainly focus on so-called multidifferentials which are defined on the product $(\Sigma)^{n+1}$ of the Riemann surfaces. For $i\in\{0,1,...,n\}$ let us denote by $\pi_i:(\Sigma)^{n+1}\to\Sigma$ the projection map to the $i$-th Riemann surface where we count from the \emph{zero}-th Riemann surface for our convenience. Then, an \emph{$(n+1)$-differential} is a meromorphic section of the line bundle
\begin{equation}
    K_{\Sigma}^{\boxtimes n+1}=\pi^*_0(K_\Sigma)\otimes\cdots\otimes\pi^*_n(K_\Sigma).
\end{equation}
It is often called a \emph{bidifferential} when $n=1$, and a \emph{multidifferential} when $n$ is not specified. In a local coordinate $z$, an $(n+1)$-multidifferential $\omega$ is locally expressed as
\begin{equation}
    \omega(p_0,..,p_n)=f(z_0,..,z_n)dz_0\otimes\cdots\otimes dz_n,
\end{equation}
where $z_i:=z(p_i)$ for $p_i\in\pi_i((\Sigma)^{n+1})$ and $f(z_0,..,z_n)$ is a meromorphic function of $z_0,..,z_n$. In particular, if $\omega\in H^0((\Sigma)^{n+1}, K_{\Sigma}^{\boxtimes n+1}(*))^{\mathfrak{S}_{n+1}}$, then $\omega$ is called a \emph{symmetric  multidifferential} where $*$ denotes a pole set of $\omega$. Finally, we say that a multidifferential $ \omega(p_0,..,p_n)$ is \emph{residue-free in $p_0$} if it has no residue when one evaluates it as a differential in $p_0$ while treating $p_i(\neq p_0)$ as parameters.

\subsubsection{Fundamental differentials}

There is a special bidifferential $B$ on $\Sigma\times\Sigma$. It is uniquely determined on any $\cS_{\bm\mu,\bm\kappa}$ by the following three properties:
\begin{description}
    \item[B1] $B$ is a meromorphic symmetric bidifferential,
    \item[B2] The only pole of $B$ is a double pole on the diagonal of biresidue 1,
    \item[B3] The $\cA_i$-period integral of $B$ vanishes for each $i=1,...,\tilde{g}$.
\end{description}
It is known that such a bidifferential exists, and it is called the \emph{fundamental bidifferential of the second kind}, or it is also referred to as the \emph{Bergman kernel}\footnote{Some literature take $B$ to be one of the defining data of a spectral curve, from which one determines the symplectic basis $\cA_i,\cB_i\in H_1(\Sigma)$ by respecting property $\textbf{B3}$.}. The $\cA_i$-dual basis $u_i\in H^1(\Sigma,\mathbb{Z})$ for $i=1,...,\tilde g$ is constructed from $B$ as
\begin{equation}
    u_i(p_0):=\frac{1}{2\pi i}\oint_{\cB_i}B(p_0,\cdot)\label{BBcycle}
\end{equation}
which satisfies the property
\begin{equation}
   \oint_{\cA_i}u_j=\delta_{ij}.
\end{equation}
For any hyperelliptic curve, $B$ satisfies
\begin{equation}
    B(p_0,\sigma(p_1))=B(\sigma(p_0),p_1),\quad B(p_0,p_1)+B(p_0,\sigma(p_1))=\frac{dx(p_0)\cdot dx(p_1)}{(x(p_0)-x(p_1))^2},\label{BB}
\end{equation}
where we often denote a tensor product by $\cdot$ for brevity.

Let us introduce another special differential. Given a hyperelliptic spectral curve $\cS_{\bm\mu,\bm\kappa}$, let $p\not\in\cR$ and we choose a specific representative $\sA_i$ of $\cA_i$ for each $i\in\{1,...,\tilde g\}$. Then the \emph{fundamental differential of the third kind normalised along $\sA_i$} denoted by $\eta^p_{\sA}$ is defined by the following property:
\begin{description}
    \item[$\bm\eta$1] $\eta^p_\sA$ is a meromorphic differential whose only pole is a residue $+1$ at $p$ and $-1$ at $\sigma(p)$
    \item[$\bm\eta$2] The contour integral of $\eta^p_\sA$ along $\sA_i$ vanishes for each $i=1,...,\tilde{g}$.
\end{description}
In particular, let $\mathfrak{F}$ be a fundamental domain of $\Sigma$ with $\partial\mathfrak{F}=\bigcup_{i=1}^{\tilde g}(\sA_i\cup\sB_i)$. Then,  $\eta^p_{\sA}$ can be written by an integral of $B$ as
\begin{equation}
    \eta^p_{\sA}(p_0)=\int_{\sigma(p)}^pB(p_0,\cdot),\label{projection}
\end{equation}
where the integration contour is taken within the fundamental domain. $\eta^p_\sA(p_0)$ is anti-invariant under $p_0\mapsto\sigma(p_0)$ due to \eqref{BB}. It is worth noting that $\oint_{p_0\in\cA_i}\eta_{\sA}^p(p_0)$ is not uniquely determined only by a choice of $\cA_i\in H_1(\Sigma,\mathbb{Z})$ because of the presence of the residues, hence we have to further choose each representative $\sA_i$. 

Computationally, $\eta^p_\sA$ plays an important role in the following sense. Let us consider a meromorphic differential $\lambda$ with a pole at $a\in\Sigma$. We then define a new differential $\tilde \lambda_a$ by
\begin{align}
    \tilde \lambda_a(p_0):=\Res_{p=a}\eta^p_\sA(p_0)\cdot \lambda(p).\label{sing}
\end{align}
When $a\neq p_0$, $\tilde \lambda_a(p_0)$ is anti-invariant under the involution $\sigma$ on $p_0$. In addition, $\tilde\lambda_a$ has the same poles structure as $\lambda$ at $p_0=a$, with the same sign at $p_0=a$ and the opposite sign at $p_0=\sigma(a)$. We note that the singular behaviour of $\lambda(p_0)$ at $p_0=\sigma(a)$ may differ from that of $\tilde \lambda(p_0)$ at $p_0=\sigma(a)$ in general. Furthermore, $\oint_{\sA_i}\tilde\lambda_a(p_0)=0$ for all $i=1,...,\tilde g$ thanks to the property ${\bm \eta}{\bf 2}$, regardless of the value of $\oint_{\sA_i}\lambda(p_0)$. In summary, $\eta^p_\sA$ in \eqref{sing} works as an operator to:
\begin{enumerate}
    \item extract the singular behaviour of the input differential $\lambda$ at a chosen point $a$,
    \item make the resulting differential $\tilde\lambda_a$ anti-invariant under the involution $\sigma$,
    \item normalise the input differential $\lambda$ along $\sA_i$ for all $i=1,...,\tilde{g}$.
\end{enumerate}
In particular, when $a\in\cR$ and when $\lambda$ is invariant under $\sigma$, it follows that $\tilde \lambda_a=0$ regardless of the pole order of $\lambda$.

Note that $\eta^p_\sA(p_0)$ works differently when $a=p_0$ or $a=\sigma(p_0)$ in \eqref{sing}. In these cases, we simply have:
\begin{equation}
    \Res_{p=p_0}\eta_\sA^p(p_0)\cdot \lambda(p)=-\lambda(p_0),\quad\Res_{p=\sigma(p_0)}\eta_\sA^p(p_0)\cdot \lambda(p)=\lambda(\sigma(p_0))\label{eta0}
\end{equation}

Before turning to refined loop equations, let us show two useful lemmas about the properties of $\eta^p_\sA$ which we use several times in the discussions below:
\begin{lem}\label{lem:reseta}
Let $\lambda^{I},\lambda^A$ be meromorphic differentials on $\Sigma$ which are invariant and anti-invariant under $\sigma$ respectively, and let $\mathfrak{p}_{+}\sqcup\sigma(\mathfrak{p}_+)$ be any subset of poles of $\lambda^I$ or $\lambda^A$ that are not in $\cR$ where for each $p\in\mathfrak{p}_{+}$, $\sigma(p)$ is necessarily in $\sigma(\mathfrak{p}_+)$. Then it follows
\begin{align}
 \left(\sum_{r\in \{p_0\}\cup\mathfrak{p}_+}+\sum_{r\in \{\sigma(p_0)\}\cup\sigma(\mathfrak{p}_+)}\right)\Res_{p=r}\eta_\sA^{p}(p_0)\cdot \lambda^I(p)&=0,\label{etaI}\\
    \left(\sum_{r\in \{p_0\}\cup\mathfrak{p}_+}-\sum_{r\in \{\sigma(p_0)\}\cup\sigma(\mathfrak{p}_+)}\right)\Res_{p=r}\eta^{p}_\sA(p_0)\cdot \lambda^A(p)&=0.\label{etaA}
\end{align}
\end{lem}
\begin{proof}
This is straightforward due to the anti-invariance of $\eta_\sA^p(p_0)$ under $\sigma$ and the property of $\eta_\sA^p$ in residues as explained above. 
\end{proof}

\begin{lem}[Riemann bilinear identity]\label{lem:RBI}
Let $\lambda$ be a meromorphic differential on $\Sigma$ and let $\cP_{\lambda}$ be the set of all poles of $\lambda$. Then, we have
\begin{align}
    \sum_{r\in\{p_0\}\cup\{\sigma(p_0)\}\cup\cP_{\lambda}}\Res_{p=r}\eta^p_\sA(p_0)\cdot \lambda(p)&= -\Delta \lambda(p_0)+\sum_{r\in\cP_{\lambda}}\Res_{p=r}\eta^p_\sA(p_0)\cdot \Delta\lambda(p)\nonumber\\
    &=-\sum_{i=1}^{\tilde g}u_i(p_0)\cdot \oint_{p\in\sA_i}\Delta \lambda(p).\label{RBI}
\end{align}
\end{lem}
Note that the first equality in \eqref{RBI} is due to \eqref{etaI}. In fact \eqref{etaI} holds even when the sum contains ramification points $\cR$. The second equality is actually what the Riemann bilinear identity means. That is, although $\{p_0\}\sqcup\{\sigma(p_0)\}\sqcup\cP_{\lambda}$ is the set of all poles of $\eta^p(p_0)\cdot \lambda(p)$, the sum of the residues does not vanish. This is because $\eta^p(p_0)$ as a function of $p$ is defined only within a fundamental domain, and we have to take account of the contributions from the boundary of the fundamental domain which are none other than the right-hand side.

\subsection{Refined Loop Equations}\label{sec:GLE}

Starting form a hyperelliptic refined spectral curve, our goal is to construct an infinite sequence of multidifferentials $\omega_{g,n+1}$ labelled by $2g,n\in\bZ_{\geq0}$ recursively in $2g-2+n$. We first define the unstable multidifferentials:
\begin{defin}\label{def:unstable}
   Given a hyperelliptic refined spectral curve $\cS_{\bm \mu,\bm\kappa}$, fix a \emph{refinement parameter} $\sQ\in\bC$ and choose a representative $\sA_i$ of $\cA_i$ for each $i\in\{1,...,\tilde g\}$. Then, the \emph{unstable multidifferentials} $\omega_{0,1},\omega_{0,2},\omega_{\frac12,1}$ are defined as follows:
   \begin{align}
    \omega_{0,1}(p_0):&=\frac12\Delta y(p_0)\cdot dx(p_0),\label{w01}\\
    \omega_{0,2}(p_0,p_1):&=-B(p_0,\sigma(p_1)),\label{w02}\\
    \omega_{\frac12,1}(p_0):&=\frac{\sQ}{2}\left(-\frac{d\Delta y(p_0)}{\Delta y(p_0)}+\sum_{p\in\widetilde{\cP}_+}\mu_p\cdot\eta^p_\sA(p_0)+\sum_{i=1}^{\tilde g}\kappa_i\cdot u_i(p_0)\right),\label{w1/2,1}
\end{align}
\end{defin}
Let us give a few comments. 
\begin{itemize}
    \item The standard definition of $\omega_{0,1}$ and $\omega_{0,2}$ in the Chekhov-Eynard-Orantin formalism is $\omega_{0,1}=y\cdot dx$ and $\omega_{0,2}=B$. For hyperelliptic curves, the difference between their definition and ours is merely conventional. We take the form \eqref{w01} and \eqref{w02} because the rest of the formula below becomes simpler. Note that in fact, our $\omega_{0,2}$ is a more natural form from the matrix model perspective (see e.g. \cite{E16,BO18}).
    \item When one compares two different choices of representatives $\sA_i$ and $\sA'_i$, their difference on $\omega_{\frac12,1}$ is always in $H^1(\Sigma,\bZ)$. This means that choosing different representatives $\cA'_i$ is equivalent to considering $\cS_{\bm \mu,\bm\kappa'}$ with an appropriate shift of $\kappa'_i$ by $\mu_p$'s. Thus, it does not cause any pathological effect.
    \item One may wonder why the choice of representatives $\sA_i$ is not a part of the data of $\cS_{\bm \mu,\bm\kappa}$. This is because we show or conjecture that multidifferentials $\omega_{g,n}$ for $2g-2+n>0$ are residue-free, implying that their $\cA,\cB$-integrals are uniquely. The choice of representatives is just a technical aspect and there is no conceptual importance.
\end{itemize}

Now we turn to the construction of other multidifferentials. Instead of directly defining them as \cite{KO22} does, in this section we rather impose several constraints, and discuss whether such constrained multidifferentials actually exist. This approach is similar to how \cite{BEO13,BS15,BBCCN18} \emph{derive} the topological recursion formula. That is, we would like to view the refined topological recursion formula as a way of solving a more fundamental structure which we call refined loop equations. We take this approach because it uncovers a more geometric structure underlying and also because it gives a hint about how to generalise the hyperelliptic framework into higher-ramified cases.

Let us introduce a convenient notation. We denote by $\Omega_{{\rm stab}}^{n+1}(\Sigma)$ the set of $(n+1)$-differentials $\omega$ such that for $i\in\{0,..,n\}$ the poles of $\omega(p_0,..,p_n)$ with respect to $p_i$ lie in $\cR^*$, $\sigma(\widetilde{\cP}_+)$, or at $p_i=\sigma(p_j)$ for all $j\neq i$. Furthermore, we denote by $\Omega_{{\rm stab}}^{n+1}(\Sigma)^{\mathfrak{S}_{n+1}}\subset\Omega_{{\rm stab}}^{n+1}(\Sigma)$ the set of symmetric $(n+1)$-differentials with the above pole structure.

We now define the hyperelliptic refined loop equations, and define stable multidifferentials as their solutions:

\begin{defin}\label{def:GLE}
Given a hyperelliptic refined spectral curve $\cS_{\bm \mu,\bm\kappa}$, define unstable multidifferentials. For $2g,n\in\bZ_{\geq0}$, a sequence of multidifferentials $\check \omega_{g,n+1}\in  \Omega_{{\rm stab}}^{n+1}(\Sigma)^{\mathfrak{S}_{n+1}}$ is said to be a solution of \emph{hyperelliptic refined loop equations} if $\check\omega_{0,1}=\omega_{0,1}$, $\check\omega_{0,2}=\omega_{0,2}$, $\check\omega_{\frac12,1}=\omega_{\frac12,1}$, and if for $2g-2+n\geq0$, $\check \omega_{g,n+1}$ is residue-free, normalised along the $\cA_i$-cycle for $i\in\{1,..,\tilde g\}$, and the differential (in $p_0$) $\check R_{g,n+1}(p_0;J)$ defined below is holomorphic at $\cR$ and anti-invariant under the involution $\sigma$: 
    \begin{align}\label{GLE}
       \check  R_{g,n+1}^{\sQ}(p_0;J):=\frac{\check Q_{g,n+1}^{\sQ}(p_0;J)}{2\omega_{0,1}(p_0)},
    \end{align}
    where $\{p_0\}\cup J:=(p_0,p_1,...,p_n)\in(\Sigma)^{n+1}$, and 
\begin{align}
  \check  Q_{g,n+1}^{\sQ}(p_0;J):=&\sum_{\substack{g_1+g_2=g\\J_1\sqcup J_2=J}}\check \omega_{g_1,n_1+1}(p_0,J_1)\cdot\check \omega_{g_2,n_2+1}(p_0,J_2)+\sum_{t\sqcup I=J}\frac{dx(p_0)\cdot dx(t)}{(x(p_0)-x(t))^2}\cdot\check \omega_{g,n}(p_0,I)\nonumber\\
    &+\check \omega_{g-1,n+2}(p_0,p_0,J)+\sQ \cdot dx \cdot d_0\frac{\check \omega_{g-\frac12,n+1}(p_0,J)}{dx(p_0)}.\label{Rec}
\end{align}
\end{defin}
\begin{defin}\label{def:stables}
   Given a hyperelliptic refined spectral curve $\cS_{\bm \mu,\bm\kappa}$, multidifferentials $\check \omega_{g,n+1}$ for $2g,n\in\bZ_{\geq0}$ with $2g-2+n\geq0$ are called \emph{stable multidifferentials} if they, together with unstable multidifferentials, are a solution of hyperelliptic refined loop equations.
\end{defin}

We often drop `hyperelliptic' and say a refined spectral curve or refined loop equations for brevity. Also, we say the refined loop equation of type $(g,n+1)$ when we would like to specify $g$ and $n+1$ of $\check R_{g,n+1}^{\sQ}$. 

Let us give a few more comments to clarify some notation and convention:
\begin{itemize}
    \item We keep the letter $\omega_{g,n+1}$ for the refined topological recursion (which we will define shortly) in order to symbolically distinguish from stable multidifferentials $\check \omega_{g,n+1}$ defined as a solution of refined loop equations. Since their definitions are different, showing $\omega_{g,n+1}=\check \omega_{g,n+1}$ becomes one of the tasks of the present paper.
    \item Refined loop equations are in fact a set of \emph{constraints} rather than equations because we do not specify e.g. pole of $\check R_{g,n+1}^{\sQ}$. In particular, we are not imposing any conditions on $\check R_{g,n+1}(p_0;J)$  with respect to variables in $J$. It is simply a tradition to call them equations.
    \item While the first condition on $\check R_{g,n+1}^{\sQ}$ is a local constraint at each $r\in\cR$, the second condition is a \emph{global} one because the involution $\sigma$ is globally defined. This is in fact a strong constraint, and if a solution of refined loop equations exists, one can uniquely determine $\check R_{g,n+1}^{\sQ}(p_0;J)$ as we will show shortly.
    \item In the setting of \cite{BEO13,BS15}, they soften the global involution condition on $\check R_{g,n+1}^{\sQ}$ to local ones at each $r\in\cR$, but instead they require so-called linear loop equations. In our setting, on the other hand, we can \emph{derive} linear loop equations when $\sQ=0$ thanks to the stronger global involution constraint. However, when $\sQ\neq0$, we will show shortly that there is no straightforward analogue of linear loop equations.
    \item One might wonder why $\check Q_{g,n+1}^{\sQ}$ in \eqref{Rec} is defined in this specific form. In fact, this is inspired by $\beta$-deformed Virasoro constraints, and also by observations given in \cite{CE06} in terms of matrix models. We also note that $\check Q_{g,n+1}^{\sQ}$ in the present paper looks slightly different from that of \cite{KO22}, due to the different definition of $\omega_{0,2}$. 
\end{itemize}

Since refined loop equations impose complicated constraints on $\check \omega_{g,n+1}$, there are three important questions to be answered, namely:
\begin{description}
    \item[Q1] Does there exist a solution of refined loop equations?
    \item[Q2] If so, is it unique?
    \item[Q3] If so, is there any formula to recursively construct $\check \omega_{g,n+1}$? 
\end{description}
$\textbf{Q1}$ is hard to prove. In fact, existence of a solution of loop equations is nontrivial to prove even in the unrefined setting. In the present paper we prove existence only when $\Sigma=\bP^1$ in Section \ref{sec:properties}, and leave the remaining cases to be a conjecture. On the other hand, if we assume existence, then the answer to \textbf{Q2} and \textbf{Q3} is simple; the solution is unique, and it can be constructed recursively by an explicit contour-integral formula.

\begin{defin}\label{def:contour}
    Let Let $J_0=\{p_0,..,p_n\}$ for $n\geq0$, and we assume $(J_0\cup\widetilde{\cP}_+)\cap(\cR\cup\sigma(J_0\cup\widetilde{\cP}_+))=\varnothing$. We define a contour $C_+\subset\Sigma$ as a connected and simply-connected closed contour encircling counter-clockwise all points of $J_0\cup\widetilde\cP_+$ but no points of $\cR\cup \sigma(J_0\cup\widetilde\cP_+)$. Such a contour can be always taken thanks to the assumption, and we drop the $n$-dependence from $C_+$ for brevity. Similarly, we define $C_-\subset\Sigma$ as a connected and simply-connected contour encircling counter-clockwise all points of $\cR\cup \sigma(J_0\cup\widetilde\cP_+)$ but no points of  $J_0\cup\widetilde\cP_+$. 
\end{defin}

\begin{thm}\label{thm:RTR}
    Given a hyperelliptic refined spectral curve $\cS_{\bm \mu,\bm\kappa}$ with a choice of representatives $\sA_i$ of $\cA_i\in H_1(\Sigma,\bZ)$ for each $i\in\{1,..,\tilde g\}$, assume that there exists a solution of hyperelliptic refined loop equations. Then, its unique solution is recursively constructed by one of the following formulae:
    \begin{description}
        \item[recursion 1]
        \begin{equation}
            \check \omega_{g,n+1}(p_0,J)=\frac{1}{2 \pi i}\left(\oint_{p\in C_+}-\oint_{p\in C_-}\right)\frac{\eta_\sA^p(p_0)}{4\omega_{0,1}(p)}\cdot \check {\rm Rec}_{g,n+1}^{\sQ}(p,J),\label{RTR1}
        \end{equation}
        \item[recursion 2]
        \begin{align}
            \check \omega_{g,n+1}(p_0,J)=&-\frac{\check {\rm Rec}_{g,n+1}^{\sQ}(p_0,J)}{2\omega_{0,1}(p_0)}+\hat R_{g,n+1}^{\sQ}(p_0,J),\label{RTR2}
        \end{align}
    \end{description}
    with 
        \begin{align}
        \check  {\rm Rec}_{g,n+1}^{\sQ}(p,J):&=\check Q_{g,n+1}^{\sQ}(p,J)-2\omega_{0,1}(p)\cdot\check \omega_{g,n+1}(p,J),\\
            \hat R_{g,n+1}^{\sQ}(p_0,J):&=\frac{1}{2 \pi i}\oint_{p\in \hat{C}_+}\frac{\eta_\sA^p(p_0)}{2\omega_{0,1}(p)}\cdot\check {\rm Rec}_{g,n+1}^{\sQ}(p,J)+\sum_{i=1}^{\tilde g}u_i(p_0)\cdot\oint_{p\in\sA_i}\frac{\check {\rm Rec}_{g,n+1}^{\sQ}(p,J)}{2\omega_{0,1}(p)}.\label{R_{g,n+1}}
        \end{align}
        where $\hat{C}_+$ contains the same points as $C_+$ except $p_0$, and we analytically continue $\check \omega_{g,n+1}(J_0)$ to $\sigma(\widetilde{\cP}_+)$ and ineffective ramification points after taking the contour integrals.
\end{thm}
\begin{proof}
Due to the pole structure of $\check \omega_{g,n+1}$, the Riemann bilinear identity (Lemma \ref{lem:RBI}) gives 
\begin{equation}
    \frac{1}{2 \pi i}\left(\oint_{p\in C_+}-\oint_{p\in C_-}\right)\eta^p_\sA(p_0)\cdot\check \omega_{g,n+1}(p,J)=-2\check \omega_{g,n+1}(p_0,J),
\end{equation}
where we also used the fact that $\check \omega_{g,n+1}(p_0,J)$ is normalised along the $\cA_i$-cycle. On the other hand, Lemma \ref{lem:reseta} implies that $\check R_{g,n+1}^{\sQ}(p,J)$ follows
    \begin{equation}
        \frac{1}{2 \pi i}\left(\oint_{p\in C_+}-\oint_{p\in C_-}\right)\eta^p_\sA(p_0)\cdot \check R_{g,n+1}^{\sQ}(p,J)=-\sum_{r\in\mathcal{R}}\Res_{p=r}\eta^p_\sA(p_0)\cdot \check R_{g,n+1}^{\sQ}(p,J)=0.
    \end{equation}
    Note that this holds precisely because we imposed two constraints on $\check R_{g,n+1}^{\sQ}(p,J)$, namely, holomorphicity in $p$ at $\cR$ and anti-invariance under $p\mapsto\sigma(p)$. These two equations ensure that for each $2g,n\in\bZ_{\geq0}$ with $2g-2+n\geq0$, the refined loop equation \eqref{GLE} of type $(g,n+1)$ gives the formula \eqref{RTR1} for $\check \omega_{g,n+1}$. Then, it is obvious by derivation that $\check \omega_{g,n+1}$ is uniquely and recursively constructed by the formula \eqref{RTR1}.

    If one applies the Riemann bilinear identity to the derived recursion formula \eqref{RTR1}, it is easy to see that we have:
    \begin{equation}
      \check   \omega_{g,n+1}(J_0)=\frac{1}{2 \pi i}\oint_{p\in {C}_+}\frac{\eta_\sA^p(p_0)}{2\omega_{0,1}(p)}\cdot\check {\rm Rec}_{g,n+1}^{\sQ}(p,J)+\sum_{i=1}^{\tilde g}u_i(p_0)\cdot\oint_{p\in\sA_i}\frac{\check {\rm Rec}_{g,n+1}^{\sQ}(p,J)}{2\omega_{0,1}(p)}.
    \end{equation}
    where we used the property:
    \begin{equation}
        \oint_{p\in\cA_i} (\check \omega_{g,n+1}(p,J)+ \check   \omega_{g,n+1}(\sigma(p),J))=\oint_{p\in\cA_i}\frac{\Delta\check {\rm Rec}_{g,n+1}^{\sQ}(p,J)}{2\omega_{0,1}(p)}=0
    \end{equation}
    Evaluating the contribution from $p_0$ inside $C_+$, we arrive at the second recursion formula.
\end{proof}

\begin{rem}\label{rem:cancel}
    Let us comment on the roles of $\hat R_{g,n+1}^{\sQ}(p_0,J)$. It is in fact the exact expression of $\check R_{g,n+1}^{\sQ}(p_0,J)$ in the refined loop equation of type $(g,n+1)$ (Definition \ref{def:GLE}). In particular, by reducing the contour integral along $\check C_+$ to the sum of residues inside $\check C_+$, it is easy to see that it satisfies the condition in Definition \ref{def:GLE} due to the property of $\eta_{\sA}^p(p_0)$. Furthermore, $\hat R_{g,n+1}^{\sQ}(p_0,J)$ cancels the pole of the first term in \eqref{RTR2} as $p_0\to p\in J\cup\widetilde\cP_+$, and also cancel the period integrals along $\sA_i$. Therefore, this explicitly shows that $\check \omega_{g,n+1}(p_0,J)$ has no pole as $p_0\to p\in J\cup\widetilde\cP_+$, being consistent with the requirement that $\check \omega_{g,n+1}\in\Omega_{\text{Stab}}^{n+1}(\Sigma)$.
\end{rem}

\subsection{Refined Topological Recursion}\label{sec:CE}

We derived the recursion formula \eqref{RTR1} as a way of solving refined loop equations. An advantage of this approach is to uncover a fundamental geometric structure underlying the recursion formula \eqref{RTR1}. In particular, although the pole structure of stable multidifferentials is very different from the unrefined setting, the imposed condition on refined loop equations itself (Definition \ref{def:GLE}) is similar to that on quadratic loop equations \cite{BEO13,BS15}. The only essential differences are whether the involution condition is imposed globally or locally (and also the absence of linear loop equations which we will discuss in the next section) which clarifies that the refinement of the present paper is a natural one.

On the other hand, the definition of stable multidifferentials (Definition \ref{def:stables}) is indirect in the sense that the recursion formula \eqref{RTR1} makes sense only after we assume existence of a solution of refined loop equations. Therefore, we now change our perspective, and aim for directly defining multidifferentials $\omega_{g,n+1}$ by taking the recursion formula \eqref{RTR1} as the defining equations. In doing so, one first has to recall the followings:
\begin{itemize}
    \item $\eta^p_\sA(p_0)$ is a well-defined function of $p$ only within a fundamental domain.
    \item The two closed contours $C_\pm$ in the recursion formula \eqref{RTR1} are connected and simply-connected in $\Sigma$.
    \item $p_0$ in the recursion formula \eqref{RTR1} plays a different role from any other variables in $J$.
\end{itemize}
These indicate that if one wishes to take the recursion formula \eqref{RTR1} as the starting point, then the definition of multidifferentials $\omega_{g,n+1}$ for $2g-2+n\geq0$ makes sense only within a fundamental domain $\mathfrak{F}\subset\Sigma$, and it is, \emph{a propri}, not guaranteed that the resulting multidifferentials $\omega_{g,n+1}$ are symmetric. In other words, it becomes our task to prove that $\omega_{g,n+1}$ satisfy all the expected properties. Having these remarks addressed, we now define the hyperelliptic refined topological recursion\footnote{Similar to a refined spectral curve and refined loop equations, we often drop `hyperelliptic' and call it the refined topological recursion for brevity.}.

\begin{defin}\label{def:RTR}
    Given a hyperelliptic refined spectral curve $\cS_{\bm\mu,\bm\kappa}$ together with a refinement parameter $\sQ$, consider a fundamental domain $\mathfrak{F}\subset\Sigma$ with $\partial\mathfrak{F}:=\bigcup_{i=1}^{\tilde g}(\sA_i\cup\mathscr{B}_i)$. Let $J_0:=(p_0,..,p_n)\in(\mathfrak{F})^{n+1}$ and assume that $(J_0\cup\widetilde{\cP}_+)\cap(\cR\cup\sigma(J_0\cup\widetilde{\cP}_+))=\varnothing$. Then, the \emph{hyperelliptic refined topological recursion} is a recursive definition of multidifferentials $\omega_{g,n+1}$ on $(\mathfrak{F})^{n+1}$ for $2g,n\in\bZ_{\geq0}$ such that $\omega_{0,1},\omega_{0,2},\omega_{\frac12,1}$ are defined as in Definition \ref{def:unstable} and for $2g-2+n\geq0$:
    \begin{equation}
            \omega_{g,n+1}(J_0):=\frac{1}{2 \pi i}\left(\oint_{p\in C_+}-\oint_{p\in C_-}\right)\frac{\eta_\sA^p(p_0)}{4\omega_{0,1}(p)}{\rm Rec}_{g,n+1}^{\sQ}(p,J),\label{RTR}
        \end{equation}
        where $C_\pm$ are defined as in Definition \ref{def:contour}, and we analytically continue $\omega_{g,n+1}(J_0)$ to the region $\cR\cup\sigma(J_0\cup\widetilde{\cP}_+)$ whenever it is finite.
\end{defin}

\begin{rem}\label{rem:notation}
    When we say that ``the refined topological recursion solves refined loop equations'', we mean that for all $g,n\in\bZ_{\geq0}$ $\omega_{g,n+1}$ constructed from the refined topological recursion is an element in $\Omega_{{\rm stab}}^{n+1}(\Sigma)^{\mathfrak{S}_{n+1}}$, it has no residue, and it is normalised along $\cA_i$-cycles. Equivalently, we may say ``the refined topological recursion constructs stable multidifferentials'', which means $\omega_{g,n+1}=\check \omega_{g,n+1}$ for all $2g,n\in\bZ_{\geq0}$. We will prove in the next section that the refined topological recursion in fact solves refined loop equations whenever $\Sigma=\bP^1$, and leave the other cases to be a conjecture.
\end{rem}

\subsubsection{Chekhov-Eynard formula}

The Chekhov-Eynard formula involves the contour $C_\mathcal{D}$ ``encircling clockwise all singular points (cuts)'' lying on the ``physical sheet/leaf'' \cite[Section 2]{CE06,C10}. See also \cite[Figure 1]{CE06}. Recall that the physical sheet denotes one of the two sheets on which $\omega_{g,n}$ has no poles, and a branch cut connects a pair of ramification points. Thus, the physical sheet viewed on $\Sigma$ can be thought of as the region inside $C_+$, and $C_\mathcal{D}$ in \cite{CE06} corresponds to $C_+$ (up to homotopy). Note that $C_-$ necessarily goes through the branch cuts when $\Sigma\neq\bP^1$ as it contains $\cR$. \cite{CE06} does not have a clear explanation about the treatment of $p\in J,\sigma(J)$ for multidifferentials, but it can be inferred that $C_+$ is the contour \cite{CE06} considers.

Then, the Chekhov--Eynard formula is written as:
\begin{equation}
    \omega^{\text{CE}}_{g,n+1}(p_0,J)=\frac{1}{2\pi i}\oint_{C_+}\frac{\eta_\sA^{p}(p_0)}{2\omega_{0,1}(p)}\cdot {\rm Rec}_{g,n+1}^{\sQ}(p,J),\label{CE}
\end{equation}
where ${\rm Rec}_{g,n+1}^{\sQ}(p,J)$ is the same as ours \eqref{Rec}. One can use the Riemann bilinear identity (Lemma \ref{lem:RBI}) to rewrite their recursion to the following form:
\begin{align}
    \omega^{\text{CE}}_{g,n+1}(p_0,J)=&\frac{1}{2\pi i}\left(\oint_{C_+}-\oint_{C_-}\right)\frac{\eta_\sA^{p}(p_0)}{4\,\omega_{0,1}(p)}\cdot {\rm Rec}_{g,n+1}^{\sQ}(p,J)\nonumber\\
    &-\sum_{i=1}^{\tilde g}u_i(p_0)\oint_{p\in\sA_i}\frac{{\rm Rec}_{g,n+1}^{\sQ}(p,J)}{2\omega_{0,1}(p)}.\label{CE2}
\end{align}
Therefore, the difference from our recursion formula is the last line, i.e. contributions from the $\sA_i$-contour integrals which remains nonzero when $\Sigma\neq\bP^1$. This difference appears exactly because the sum of residues of the integrand is not zero. As shown in Theorem \ref{thm:RTR}, a unique solution of refined loop equations is given by our refined recursion formula, not \eqref{CE}, hence $\omega^{\text{CE}}_{g,n}$ can not be stable multidifferentials (Definition \ref{def:stables}). It can even be shown that these contributions are generically nonzero \emph{even when we set} $\sQ=0$. This indicates that the Chekhov-Eynard formula \eqref{CE} should be improved.

\begin{rem}
   There is no notion of refined spectral curves in \cite{CE06}, hence their spectral curves correspond to \emph{specific} choices of $\widetilde\cP_+$ compatible with matrix models. Also, strictly speaking, the above subtlety exists even in their unrefined formula of \cite{CE05}, though their final formula in terms of residues at ramification points stands correct.
\end{rem}

To be clear, \cite[Eq. 2.12]{CE06}\footnote{The equation numbers we cite from \cite{CE06} are all from the JHEP version.} also discusses loop equations from the $\beta$-deformed matrix model perspective. It turns out that their loop equations are still consistent with ours if we set $\bm{\mu}$ to special values. Therefore, our and their approaches in fact have a similar starting point. However, the Chekhov-Eynard formula does not solve refined loop equations because of the contributions from the boundary terms. We also remark that knowing loop equations of matrix models does \emph{not} imply existence of \emph{meromorphic} $\omega_{g,n}$ on $\Sigma$ --- it only implies formal series in $x_i^{-k_i}$ on the base.
 
%Let us make this point clearer by following their argument.

% The subtlety in the derivation of their proposed formula appears in \cite[Eq. 3.15]{CE06} when they considered the relation between two contours $\cC_D$ and $\bar \cC_D$ where the bar is equivalent to the global involution $\sigma$ in our notation. They claim (the second equality in \cite[Eq. 3.15]{CE06}) that $\bar\cC_D$ equals to $-\cC_D$ plus contributions from the ramification points. However,  it is inevitable in their setting that $\bar\cC_D$ crosses the boundary of the original fundamental domain. As a consequence, the contributions from the boundary should be carefully taken into account which is missing in the discussion of \cite[Eq. (3.15)]{CE06}. Note that $H_1(\bP^1,\bZ)$ is trivial, and thus our recursion formula \eqref{RTR} is equivalent to the Chekhov-Eynard formula \eqref{CE} when we set $\bm\mu$ appropriately, as well as the formula by Kidwai and the author \cite{KO22} for a special class of genus-zero curves.

At last, we note that not all hyperelliptic curves can appear in matrix model loop equations. For example, we need a geometric definition when the hyperelliptic curve underlying is given by  $P(x,y)=y^2-x$ or $P(x,y)=y^2-(4x^3+2t x+u)$ for some parameter $t,u\in\bC$. In contrast, our refined recursion formula \eqref{RTR} is applicable to any hyperelliptic curve.

\newpage

\section{Properties}\label{sec:properties}
In this section, we investigate whether the refined topological recursion solves refined loop equations for a given a refined spectral curve $\cS_{\bm\mu,\bm\kappa}$, or whether the refined topological recursion constructs stable differentials (c.f. Remark \ref{rem:notation}).

In doing so, the following lemma will be useful throughout the section:
\begin{lem}\label{lem:wdecomp}
Let $\omega_{g,n+1}$ for $2g,n\in\bZ_{\geq0}$ with $2g-2+n\geq0$ be multidifferentials constructed by the refined topological recursion (Definition \ref{def:RTR}). Then, for each $g,n$ we have
\begin{align}
\omega_{g,n+1}(p_0,J)+\omega_{g,n+1}(\sigma(p_0),J)&=-\frac{\Delta_0{\rm Rec}_{g,n+1}(p_0,J)}{2\omega_{0,1}(p_0)},\\
\omega_{g,n+1}(p_0,J)-\omega_{g,n+1}(\sigma(p_0),J)&=\frac{1}{2\pi i}\left(\oint_{\hat{C}_+}-\oint_{\hat{C}_-}\right)\frac{\eta^p_\sA(p_0)}{2\omega_{0,1}(p)}{\rm Rec}_{g,n+1}(p,J)
\end{align}
where $\hat{C}_+$ contains the same points as $C_+$ except $p_0$, and $\hat{C}_-$ contains the same points as $C_-$ except $\sigma(p_0)$.
\end{lem}
\begin{proof}
Recall the property of the fundamental differential of the third kind $\eta_\sA^p$ (Section \ref{sec:RSC}). In particular, $\eta^p_\sA(p_0)$ behaves differently when it involves residues at $p=p_0$ and $p=\sigma(p_0)$. After evaluating the contribution of $p_0$ in $C_+$ and $\sigma(p_0)$ in $C_-$, we have
\begin{equation}
\omega_{g,n+1}(p_0,J)=-\frac{\Delta_0{\rm Rec}_{g,n+1}(p_0,J)}{4\omega_{0,1}(p_0)}+\frac{1}{2\pi i}\left(\oint_{\hat{C}_+}-\oint_{\hat{C}_-}\right)\frac{\eta^p_\sA(p_0)}{4\omega_{0,1}(p)}{\rm Rec}_{g,n+1}(p,J).
\end{equation}
It is obvious that the first term is invariant under the involution $\sigma$ in $p_0$. On the other hand, the second term is anti-invariant under the involution $\sigma$ due to Property $\bm{\eta2}$. 
\end{proof}
\begin{cor}
    If the refined topological recursion solves the refined loop equation of type $(g',n'+1)$ up to $2g'-2+n'=\chi$, then for $2g-2+n=\chi+1$, $\omega_{g,n+1}(p_0,J)$ is a differential on $\Sigma$ in terms of the first variable $p_0$, not only within the fundamental domain $\mathfrak{F}$.
\end{cor}

\subsection{Reduction to the unrefined topological recursion}
We show that the refined recursion formula \eqref{RTR} reduces to the unrefined topological recursion formula \cite{CE05,EO07} when we set $\sQ=0$. 

Let us first show:
\begin{lem}\label{lem:Qpoly}
    Let $\omega_{g,n+1}$ be a multidifferential constructed by the refined topological recursion for each $2g,n\in\bZ_{\geq0}$. Then, they polynomially depend on $\sQ$-dependence. Furthermore, when one expands $ \omega_{g,n+1}$ in $\sQ$ as
    \begin{equation}
        \omega_{g,n+1}=\sum_{k=0}^{2g}\sQ^k\omega_{g,n+1}^{(k)},
    \end{equation}
    where $\omega_{g,n+1}^{(k)}$ are independent of $\sQ$, it follows that for all $\ell\in\bZ_{\geq0}$:
    \begin{align}
        \forall g\in\bZ_{\geq0},\quad&\quad\omega_{g,n+1}^{(2\ell+1)}=0\\
        \forall g\in\bZ_{\geq0}+\frac12,\quad&\quad\omega_{g,n+1}^{(2\ell)}=0
    \end{align}
    .
\end{lem}
\begin{proof}
    This is easy to verify by counting the degree of $\sQ$ in the recursion formula \eqref{RTR}. 
\end{proof}
Since the $\sQ$-dependence appears only polynomially, one realises that the recursion for $\omega_{g,n+1}^{(0)}$ is self-closed, that is, the recursion for $\omega_{g,n+1}^{(0)}$ does not involve $\omega_{g,n+1}^{(k)}$ for $k>0$. We will show that $\omega_{g,n+1}^{(0)}$ obeys the unrefined topological recursion formula, which ensures that a solution of refined loop equations exists when $\sQ=0$, as expected.

We first derive what we call linear loop equations (c.f. \cite[Theorem 4.1]{EO07}):
\begin{lem}[Linear loop equation]\label{lem:LLE}
For $2g,n\in\bZ_{\geq0}$, assume that $\omega_{g,n+1}^{(0)}$ solves the refined loop equation of type $(g,n+1)$ up to $2g-2+n=\chi\in\bZ_{\geq-1}$ when $\sQ=0$. Then, for $2g-2+n=\chi+1$, we have:
\begin{equation}
    \omega_{g,n+1}^{(0)}(p_0,J)+\omega_{g,n+1}^{(0)}(\sigma(p_0),J)=0.\label{LLE}
\end{equation}
\end{lem}
\begin{proof}
When $\chi=-1$, there is no need to assume anything and one can directly see that ${\rm Rec}_{0,3}^{\sQ=0}(p,p_1,p_2)$ and ${\rm Rec}_{1,1}^{\sQ=0}(p)$ are in fact invariant under the involution $\sigma$, which guarantees \eqref{LLE} to hold for $\omega_{0,3}^{(0)}$ $\omega_{1,1}^{(0)}$ because $R_{0,3}$ and $R_{1,1}$ are anti-invariant. When $\chi\geq0$, we can easily show ${\rm Rec}_{g,n+1}^{\sQ=0}(p,J)$ also becomes invariant under the involution $\sigma$ due to the absence of the last term in \eqref{Rec}, and this proves the lemma. 
\end{proof}

It is important to remark that we are able to \emph{derive} linear loop equations (Lemma \ref{lem:LLE}) because the second condition in Definition \ref{def:GLE} is a global constraint on $\sigma$. As we will show shortly, an analogous equation for $\omega_{g,n+1}^{(k)}$ when $k>0$ shows that the right-hand side is not even holomorphic at ramification points. This clearly indicates that heuristic traditional local approaches are not applicable to the refined topological recursion.

\begin{prop}\label{prop:Q=0}
For $2g,n\in\bZ_{\geq0}$, assume that $\omega_{g,n+1}^{(0)}$ solves the refined loop equation of type $(g,n+1)$ up to $2g-2+n=\chi\in\bZ_{\geq-1}$ when $\sQ=0$. Then, when $2g-2+n=\chi+1$, $\omega_{g,n+1}^{(0)}$ satisfies the unrefined topological recursion formula.
\end{prop}

\begin{proof}
Since ${\rm Rec}_{g,n+1}^{\sQ=0}(p,J)$ becomes invariant under the involution $\sigma$, Lemma \ref{lem:reseta} implies that the contour integrals along $C_+$ and $C_-$ in the refined recursion formula \eqref{RTR} reduces to the residue computations only at $\cR$. After adjusting the sign by the linear loop equation (Lemma \ref{lem:LLE}), one finds that the remainder in the recursion formula is none other than the unrefined topological recursion formula of Chekhov-Eynard-Orantin \cite{CE05,CEO06,EO07} for a hyperelliptic curve.
\end{proof}

\begin{cor}\label{cor:Q=0}
    The refined topological recursion solves refined loop equations when $\sQ=0$.
\end{cor}

\subsection{Existence of a solution of refined loop equations}\label{sec:existence}

Let us consider the lowest level, i.e. $2g-2+n=0$. Since $\omega_{0,3}=\omega_{0,3}^{(0)}$, Corollary \ref{cor:Q=0} implies existence of type $(0,3)$. We will consider $\omega_{\frac12,2}$ and $\omega_{1,1}$ separately.

\subsubsection{$\omega_{\frac12,2}$}
Let us first show an analogue of the linear loop equation (Lemma \ref{lem:LLE}) for $\omega_{\frac12,2}$:
\begin{lem}\label{lem:LLE1/2,2}
$\omega_{\frac12,2}$ constructed from the refined topological recursion satisfies:
    \begin{equation}
        \omega_{\frac12,2}(p_0,p_1)+\omega_{\frac12,2}(\sigma(p_0),p_1)=-\sQ\cdot d_0\frac{\Delta_0\omega_{0,2}(p_0,p_1)}{2\omega_{0,1}(p_0)}\label{LLE1/2,2}
    \end{equation}
\end{lem}
\begin{proof}
Lemma \ref{lem:wdecomp} immediately implies that, by looking at terms in ${\rm Rec}_{\frac12,2}^{\sQ}$, one finds:
\begin{align}
    &\omega_{\frac12,2}(p_0,p_1)+\omega_{\frac12,2}(\sigma(p_0),p_1)\nonumber\\
    &=-\frac{1}{2\omega_{0,1}(p_0)}\left(-\frac{d\Delta y(p_0)}{\Delta y(p_0)}\cdot\Delta_0\omega_{0,2}(p_0,p_1)+dx(p_0)\cdot d_0\left(\frac{\Delta_0\omega_{0,2}(p_0,p_1)}{dx(p_0)}\right)\right).
\end{align}
Recalling the definition of $\omega_{0,1}$ \eqref{w01}, we arrive at \eqref{LLE1/2,2}. 
\end{proof}
\begin{rem}
    This is a clear contrast from unrefined linear loop equations (Lemma \ref{lem:LLE}). This is an important observation because it suggests that a naive approach in terms of Airy structures (c.f. \cite{BBCCN18}) does not work, or it needs a significant generalisation to be applicable to the refined setting.
\end{rem}

\begin{prop}\label{prop:1/2,2}
    The refined topological recursion solves the refined loop equation of type $(\frac12,2)$.
\end{prop}
\begin{proof}
With the help of Lemma \ref{lem:reseta} and Lemma \ref{lem:wdecomp}, we can simplify the refined topological recursion for $\omega_{\frac12,2}$ as (c.f. \cite[Eq. A.21, Eq. A.22]{KO22}): 
\begin{align}
    \omega_{\frac12,2}(p_0,p_1)=&-\frac{\sQ}{4\pi i}\left(\oint_{C_+}-\oint_{C_-}\right)\frac{\Delta\omega_{0,2}(p,p_0)\cdot\Delta\omega_{0,2}(p,p_1)}{4\,\omega_{0,1}(p)}\nonumber\\
    &-\frac{\sQ}{2}\cdot\sum_{r\in\cR^*}\Res_{p=r}\frac{\eta_\sA^p(p_0)}{4\,\omega_{0,1}(p)}\cdot \Delta\omega_{0,2}(p,p_1)\cdot\left(\sum_{q\in\widetilde{\cP}}\mu_q\cdot\eta_\sA^q(p)+\sum_{i=1}^{\tilde g}\kappa_i\cdot u_i(p)\right).\label{w1/2,2}
\end{align}
The first line is manifestly symmetric in $p_0\leftrightarrow p_1$. The second line can be also shown to be symmetric by using the same argument as \cite[Lemma 2.22]{KO22}. In particular, the second line may have poles only at $\cR^*$. 

We now use the Riemann bilinear identity (Lemma \ref{lem:RBI}) to bring the refined recursion formula \eqref{RTR} to the form analogous to \eqref{RTR2}. Note that Lemma \ref{lem:LLE1/2,2} is necessary to verify this transformation, because it ensures
\begin{equation}
    \oint_{p\in\sA_i}\frac{{\rm Rec}^{\sQ}_{\frac12,2}(p,p_1)}{2\omega_{0,1}(p)}=\oint_{p\in\sA_i}\frac12\Delta_p\frac{{\rm Rec}^{\sQ}_{\frac12,2}(p,p_1)}{2\omega_{0,1}(p)}.
\end{equation}
Then as explained in Remark \ref{rem:cancel}, the poles of $\omega_{\frac12,2}(p_0,p_1)$ with respect to $p_0$ are located in $\cR,\sigma(\widetilde\cP_+)$, and at $p_0=\sigma(p_1)$. Also, it ensures that $\omega_{\frac12,2}$ is noramlised along $\sA_i$-cycles. In addition, it is easy to see that there is no pole at the poles of $\omega_{0,1}$ because $\omega_{0,1}$ appears only in the denominator in the recursion formula \eqref{RTR2}. Then, since the $\sigma$-invariant part of $\omega_{\frac12,2}$ is residue-free thanks to Lemma \ref{lem:LLE1/2,2}, it implies that $\omega_{\frac12,2}$ itself is residue-free. This shows that $\omega_{\frac12,2}$ is normalised along $\cA_i$, not only along $\sA_i$.
\end{proof}

\begin{rem}
Proposition \ref{prop:1/2,2} is crucial, because it shows that the recursion formula of the present paper resolves the issue reported in \cite[Section 1]{C10} --- Mari\~{n}o and Pasquetti obtained a non-symmetric $\omega_{\frac12,2}$ from the Chekhov-Eynard formula \eqref{CE}. 
\end{rem}

\subsubsection{$\omega_{1,1}$}
An analogue of the linear loop equation (Lemma \ref{lem:LLE}) for $\omega_{1,1}$ is given as follows, which shows again that the $\sigma$-invariant part of $\omega_{1,1}$ is singular at $\cR^*$ unlike the unrefined setting:
\begin{lem}\label{lem:LLE1,1}
    \begin{equation}
        \omega_{1,1}(p_0)+\omega_{1,1}(\sigma(p_0))=-\frac{\sQ}{2}d_0\frac{\Delta\omega_{\frac12,1}(p_0)}{\omega_{0,1}(p_0)}\label{LLE1,1}
    \end{equation}
\end{lem}
\begin{proof}
    The proof is simple. Lemma \ref{lem:wdecomp} for $\omega_{1,1}$ gives
    \begin{equation}
        \omega_{1,1}(p_0)+\omega_{1,1}(\sigma(p_0)=-\frac{1}{2\omega_{0,1}}\left(-\sQ\frac{d\Delta y(p_0)}{\Delta y(p_0)}\Delta\omega_{\frac12,1}(p_0)+dx(p_0)d_0\frac{\Delta\omega_{\frac12,1}(p_0)}{dx(p_0)}\right).
    \end{equation}
    Then, one finds the expression \eqref{LLE1,1} by a simple manipulation.
\end{proof} 
\begin{prop}\label{prop:1,1}
    The refined topological recursion solves the refined loop equation of type $(1,1)$.
\end{prop}
\begin{proof}
    We now repeat the same strategy as $\omega_{\frac12,2}$. This case is in fact simpler because we do not have to worry about the symmetry and all we have to consider is the pole structure.
    
    Thanks to Lemma \ref{lem:LLE1,1}, we can transform the refined topological recursion \eqref{RTR} to the form analogous to \eqref{RTR2}. Then, Remark \ref{rem:cancel} indicates that poles of the first term in \eqref{RTR2} as $p_0\to r\in\widetilde\cP_+$ will be cancelled by the residue at $p=r$ in the second term. Thus, $\omega_{1,1}$ can possibly have poles only at $\cR$ and at $\sigma(\widetilde\cP_+)$. Also, it has vanishing $\sA_i$-contour integrals.

    It is straightforward to see that $\omega_{1,1}$ has no poles at the pole of $\omega_{0,1}$ because $\omega_{0,1}$ always appears in the denominator, if the pole order is two or higher. However, one has to be careful at simple poles of $\omega_{0,1}$ because there is $\omega_{\frac12,1}^2$ in $\text{Rec}_{1,1}^{\sQ}$ in the first term in \eqref{RTR2} whereas there is only $\omega_{0,1}$ in the denominator, which in total may give residues. However, since the $\sigma$-invariant part of $\omega_{1,1}$ is residue-free due to Lemma \ref{lem:LLE1,1}, and since $\omega_{1,1}$ has no pole at $\widetilde\cP_+$, the full $\omega_{1,1}$ cannot have residues at any of points in $\sigma(\widetilde\cP_+)$ either. Alternatively, one can check the cancellation explicitly as shown in \cite[Lemma A.2]{KO22}.
\end{proof}

\begin{rem}
    It is worth emphasising that in the unrefined setting, existence of a solution of the loop equation of type $(g,1)$ is trivial. This is because imposing linear loop equations (c.f. Lemma \ref{lem:LLE}) immediately implies that $\omega_{g,1}$ is residue-free. As a consequence, existence of a solution of unrefined loop equations typically amounts to showing only the symmetry of multidifferentials (c.f. \cite{BBCCN18}). However, it is no longer trivial once we consider the refined topological recursion, and potential residues need to be carefully investigated as we did in Lemma \ref{lem:LLE1,1} and Proposition \ref{prop:1,1} 
\end{rem}

\subsubsection{Existence for higher levels}

Proving existence of a solution of refined loop euqtions with full generalities is a challenging task at the moment of writing. In the present paper, we prove it when $\Sigma=\bP^1$ --- the work of Kidwai and the author \cite{KO22} is only for a special class of genus-zero curves. The complexity of the proof for higher genus curves originated from the fact that $\eta^p_\sA(p_0)$ is \emph{not} a meromorphic function in $p$ on $\Sigma$, hence one has to be careful about the use of the Riemann bilinear identity. Also, recall that the refined topological recursion involves residues not only at $\cR$ but also at other points. This indicates that one cannot prove existence by a clever method in terms of Airy structures (c.f. \cite{KS17,ABCO17}).

\begin{thm}\label{thm:existence}
    Given a hyperelliptic refined spectral curve $\cS_{\bm\mu}$ of genus-zero, i.e. $\Sigma=\bP^1$, there exists a solution of hyperelliptic refined loop equations, and its unique solution is constructed by the hyperelliptic refined topological recursion. 
\end{thm}

\begin{proof}
    Since the proof is heavily computational, we only show a sketch here and leave details to Appendix \ref{sec:appendix}.

    We proceed by induction. We have already proven existence of $\omega_{g,n+1}$ when $2g-2+n=0$. For $\chi\in\bZ_{>0}$ we assume that 
    there exists a solution of the refined loop equation of type $(g',n'+1)$ for all for $2g',n'\in\bZ_{\geq0}$ with $2g'-2+n'\leq\chi$, and we consider $\omega_{g,n+1}$ with $2g-2+n=\chi+1$. It turns out that we need a slightly different treatment when $n=0$ (see Appendix \ref{sec:wg1}), hence we focus on the case when $n>0$. Equivalently, we shift $n$ by one and consider $\omega_{g,n+2}(p_0,q_0,J)$ for $n\geq0$ with $2g-2+n+1=\chi+1$.
    
    The first step is to show that $\omega_{g,n+1}(p_0,q_0,J)$ is a symmetric differential by applying the recursion formula \emph{twice}. Since $\Sigma=\bP^1$, $\eta^p(p_0)$ is a well-defined function in $p$ on $\Sigma$, and since the sum of residues vanishes, one can write the refined recursion topological recursion formula \eqref{RTR} as
    \begin{equation}
        \omega_{g,n+2}(p_0,q_0,J)=\frac{1}{2 \pi i}\left(\oint_{p\in C_+^p}\right)\frac{\eta^p(p_0)}{2\omega_{0,1}(p)}\left(\Delta\omega_{0,2}(p,q_0)\omega_{g,n+1}(p,J)+{\rm Rec}_{g,n+1}^{\sQ,**}(p,q_0,J)\right),\label{sym1}
    \end{equation}
    where we have separated the terms involving $\omega_{0,2}(p,q_0)$ and ${\rm Rec}_{g,n+1}^{\sQ,**}$ denotes the collection of remaining terms, in which $\omega_{g',n'+2}(p,q_0,J')$ with $2g'-2+n'<\chi$ appear. One can then apply the recursion formula (in the form \eqref{sym1}) to every $\omega_{g',n'+2}(p,q_0,J')$ in ${\rm Rec}_{g,n+1}^{\sQ,**}$ with respect to $q_0$, which gives:
    \begin{align}
        \omega_{g,n+2}(p_0,q_0,J)=&\frac{1}{2 \pi i}\left(\oint_{p\in C_+^p}\right)\frac{\eta^p(p_0)}{2\omega_{0,1}(p)}\Delta\omega_{0,2}(p,q_0)\omega_{g,n+1}(p,J)\nonumber\\
        &+\frac{1}{2 \pi i}\left(\oint_{p\in C_+^p}\right)\frac{1}{2 \pi i}\left(\oint_{q\in C_+^q}\right)\frac{\eta^p(p_0)}{2\omega_{0,1}(p)}\cdot\frac{\eta^q(q_0)}{2\omega_{0,1}(q)}\cdot{\rm Rec}_{g,n+2}^{\sQ,\text{twice}}(p,q,J),\label{twice}
    \end{align}
    where ${\rm Rec}_{g,n+2}^{\sQ,\text{twice}}$ is given in \eqref{Recpq}, and $C_+^q$ encircles not only points in $J\cup\widetilde{\cP}_+$, but also $q=p$. It is then easy to show (see Appendix \ref{sec:contours}) that
    \begin{equation}
        \oint_{p\in C_+^p}\oint_{q\in C_+^q}=\oint_{q\in C_+^q}\left(\oint_{p\in C_+^p}-2\pi i \Res_{p=q}\right),\label{contour0}
    \end{equation}
    where $C_+^p$ on the right-hand side \emph{does} encircle $p=q$. This is none other than \cite[Eq. (2.39)]{KO22}, after turning the contour integral from $C_+$ to $C_-$.

    Then, after several steps of nontrivial manipulation (see Appendix \ref{sec:sym}), we can show that
    \begin{align}
         &\omega_{g,n+2}(p_0,q_0,J)- \omega_{g,n+2}(p_0,q_0,J)\nonumber\\
         &=\frac{1}{2 \pi i}\left(\oint_{q\in C_+^q}\right)\frac{\eta^q(p_0)}{2\omega_{0,1}(q)}\cdot\Delta\omega_{0,2}(q,q_0)\cdot \hat{R}_{g,n+1}(q,J)\nonumber\\
         &=\frac{1}{2 \pi i}\left(\oint_{q\in C_+^q}-\oint_{q\in C_-^q}\right)\frac{\eta^q(p_0)}{4\omega_{0,1}(q)}\cdot\Delta\omega_{0,2}(q,q_0)\cdot \hat{R}_{g,n+1}(q,J),\label{sym0}
    \end{align}
    where the second equality holds because the sum of residues vanishes. Since the integrand in \eqref{sym0} is anti-invariant under $q\to \sigma(q)$, Lemma \ref{lem:reseta} implies that the contour integral reduces to the sum of residues at $\cR$. Furthermore, since $\hat{R}_{g,n+1}(q,J)$ is holomorphic as $q\to\cR$, the right-hand side of \eqref{sym0} vanishes.

    What remains to be seen is the pole structure. As discussed in Remark \ref{rem:cancel}, we can show from \eqref{sym1} that $\omega_{g,n+2}(p_0,q_0,J)$ has no poles in $p_0$ at $\{q_0\}\cup J\cup\widetilde{\cP}_+$. One can easily show from \eqref{sym1} that $\omega_{g,n+2}(p_0,q_0,J)$ has no residues in $q_0$, which implies it is residue-free because $\omega_{g,n+2}(p_0,q_0,J)$ is symmetric. See Appendix \ref{sec:residue} to show that there is no poles at the poles of $\omega_{0,1}$. Thus, we have shown that the refined topological recursion constructs $\omega_{g,n+2}(p_0,q_0,J)$ with all the expected properties. Equivalently, we have shown existence of the refined loop equation of type $(g,n+2)$ for $2g-2+n+1=\chi+1$.
\end{proof}

Although our proof works only when $\Sigma=\bP^1$, we conjecture that there exists a solution of refined loop equations for higher-genus curves.

\begin{conj}\label{conj:existence}
    {\rm Theorem} \ref{thm:existence} extends to any hyperelliptic refined spectral curve.
\end{conj}

\subsection{Dilaton equation and free energy}
Let us fix notation to consider what we call the dilaton equation. We denote by $\mathfrak{p}_{g,n+1}$ the set of all poles of $\omega_{g,n+1}(p,p_1,...,p_n)$ with respect to $p$, and by $C_{g,n+1}^{\mathfrak{p}}$ a connected, simply-connected closed contour containing all points in $\mathfrak{p}_{g,n+1}$ but not containing any poles of $\omega_{0,1}$ (thus we assume that $p_i\not\in\widetilde{\cP}$ for all $i\in\{1,..,n\}$). In particular, $C_{g,n+1}^{\mathfrak{p}}$ is contained in $C_-$. We denote by $\phi$ any primitive of $\omega_{0,1}$. 

For example, $C_{0,2}^{\mathfrak{p}}$ contains $\sigma(p_1)$ hence we have
\begin{equation}
        \frac{1}{2\pi i}\oint_{C^{\mathfrak{p}}_{0,2}}\phi(p)\cdot\omega_{0,2}(p,p_1)=\Res_{p=\sigma(p_1)}\phi(p)\cdot\omega_{0,2}(p,p_1)=\omega_{0,1}(p_1).\label{dilaton0,2}
    \end{equation}
    Note that this holds exactly because our definition of $\omega_{0,2}$ is slightly different from the standard one. If we took $\omega_{0,2}:=B$, then the right-hand side of \eqref{dilaton0,2} would have the opposite sign.

Next, $C_{0,3}^{\mathfrak{p}}$ contains effective ramification points $\cR^*$. Then, it has already been shown in \cite{EO07} that 
    \begin{equation}
        \frac{1}{2\pi i}\oint_{C^{\mathfrak{p}}_{0,3}}\phi(p)\cdot\omega_{0,3}(p,p_1,p_2)=0.\label{dilaton0,3}
    \end{equation}
Note that $C_{0,3}^{\mathfrak{p}}$ does \emph{not} encircle ineffective ramification points which is important to derive \eqref{dilaton0,3}. It does not matter which primitive we choose in the above examples because $\omega_{0,2}$ and $\omega_{0,3}$ are residue free.

For $\omega_{\frac12,2}$, we also find the following:

\begin{prop}\label{prop:dilaton12,2}
$\omega_{\frac12,2}$ constructed from the refined topological recursion satisfies:
    \begin{equation}
        \frac{1}{2\pi i}\oint_{C^{\mathfrak{p}}_{\frac12,2}}\phi(p)\cdot\omega_{\frac12,2}(p,p_1)=0.
    \end{equation}
\end{prop}
\begin{proof}
    It can be shown that the second line in \eqref{w1/2,2} vanishes after taking the contour integral along $C_{\frac12,2}^{\mathfrak{p}}$ exactly by the same argument as the proof of \eqref{dilaton0,3} in \cite{EO07}. This is because the second line in \eqref{w1/2,2} can have poles only at $\cR^*$ of at most order two, which is exactly the same behaviour as $\omega_{0,3}$.

    We now evaluate the the first line of \eqref{w1/2,2}. In order to have a consistent notation with Lemma \ref{prop:dilaton12,2}, let us choose $q$ to be the variable of integration in \eqref{w1/2,2} and set $p_0=p$. We separate the contributions of $p,p_1,\sigma(p),\sigma(p_1)$ in the first line of \eqref{w1/2,2} as follows:
    \begin{align}
       \text{first line of \eqref{w1/2,2}}=&-\frac{\sQ}{4}\left(d_p\frac{\Delta\omega_{0,2}(p,p_1)}{\omega_{0,1}(p)}+d_1\frac{\Delta\omega_{0,2}(p,p_1)}{\omega_{0,1}(p_1)}\right)\nonumber\\
       &-\frac{\sQ}{4\pi i}\left(\oint_{q\in\check C_+}-\oint_{q\in\check C_-}\right)\frac{\Delta\omega_{0,2}(q,p)\Delta\omega_{0,2}(q,p_1)}{4\omega_{0,1}(q)}.\label{dilaton12,2}
    \end{align}
    where $\check C_\pm$ denotes the resulting contours after removing contributions of the four points. It is easy to see that the first line in \eqref{dilaton12,2} vanishes after multiplying $\phi(p)$ and applying the contour integral along $C^{\mathfrak{p}}_{\frac12,2}$.

    We show that the second line in \eqref{dilaton12,2} vanishes. In principle, $\check C_\pm$ contains the pole of $\omega_{0,1}$, but the integrand obviously is regular there, hence we remove those points from $\check C_\pm$ as well. Then, the domain encircled by $\check C_+$ has no overlap with that by $C^{\mathfrak{p}}_{\frac12,2}$, thus we can exchange the order of the contour integrals as (c.f. Appendix \ref{sec:contours})
    \begin{equation}
        \oint_{p\in C^{\mathfrak{p}}_{\frac12,2}}\oint_{q\in\check C_+}=\oint_{q\in\check C_+}\oint_{p\in C^{\mathfrak{p}}_{\frac12,2}}.
    \end{equation}
    Applying the contour integral after multiplying $\phi(p)$, this contribution vanishes --- note that $C^{\mathfrak{p}}_{\frac12,2}$ contains $\sigma(p_1)$ but not $\sigma(q)$.
    
    On the other hand, both $\check C_-$ and $C^{\mathfrak{p}}_{\frac12,2}$ contains ramification points, and one has to take care of residues at $p=q$ and $p=\sigma(q)$ when one exchanges the order of integrals. However, since we have
    \begin{align}
        \Res_{p=q}\phi(p)\frac{\Delta\omega_{0,2}(q,p)\Delta\omega_{0,2}(q,p_1)}{\omega_{0,1}(q)}&=\Delta\omega_{0,2}(q,p_1),\nonumber\\
        \Res_{p=\sigma(q)}\phi(p)\frac{\Delta\omega_{0,2}(q,p)\Delta\omega_{0,2}(q,p_1)}{\omega_{0,1}(q)}&=\Delta\omega_{0,2}(q,p_1),
    \end{align}
and since $\Delta\omega_{0,2}$ has no residues, we conclude that there is no contribution from the $\check C_-$ integral either.
\end{proof}

It turns out that \eqref{dilaton0,2}, \eqref{dilaton0,3}, and \eqref{dilaton12,2} are instances of a more general set of relations between $\omega_{g,n+1}$ and $\omega_{g,n+2}$. We call them the \emph{dilaton equations}, mostly without specifying what $g,n$ are.

\begin{prop}\label{prop:dilaton}
    Given a hyperelliptic refined spectral curve $\cS_{\bm\mu}$ of genus-zero, let $\omega_{g,n+1}$ be multidifferentials constructed by the refined topological recursion on $\cS_{\bm\mu}$. Then, for all $2g,n\in\bZ_{\geq0}$, we have
    \begin{equation}
        (2-2g-n-1)\cdot\omega_{g,n+1}(p_0,J)=\frac{1}{2\pi i}\oint_{C^{\mathfrak{p}}_{g,n+2}}\phi(p)\cdot\omega_{g,n+2}(p,p_0,J).\label{dilaton}
    \end{equation}
\end{prop}
\begin{proof}
    We have already seen that it holds for $(g,n)=(0,0), (0,1),(\frac12,0)$. We proceed by induction and apply the contour integral along $C^{\mathfrak{p}}_{g,n+2}$ to the formula analogous to \eqref{RTR2} for $\omega_{g,n+2}(p_0,J,p)$, that is, we treat $p$ as the last variable since $\omega_{g,n+2}$ is symmetric differential due to Theorem \ref{thm:existence}. Thanks to the induction ansatz, the first term in \eqref{RTR2} simply gives
    \begin{align}
        &\frac{1}{2\pi i}\oint_{C^{\mathfrak{p}}_{g,n+2}}\phi(p)\cdot\left(\frac{{\rm Rec}_{g,n+2}^{\sQ}(p_0,J,p)}{2\omega_{0,1}(p_0)}\right)\nonumber\\
        &=(2-2g-n)\left(-\frac{{\rm Rec}_{g,n+1}^{\sQ}(p_0,J)}{2\omega_{0,1}(p_0)}\right)-\frac12\omega_{g,n+1}(p_0,J),\label{dilatong,n0}
    \end{align}
    where the second term is the contribution from $\omega_{0,2}$ and the rest gives the first term.

    Let us next evaluate the first term of $\hat R^{\sQ}_{g,n+2}$. Similar to how we proved Lemma \ref{prop:dilaton12,2}, we choose $q$ to be the variable of integration and we separate the contribution of $p$ from others which gives:
    \begin{align}
        \hat R^{\sQ}_{g,n+2}(p_0,J,p)=d_p\left(\frac{\eta^p(p_0)}{2\omega_{0,1}(p)}\omega_{g,n+1}(p,J)\right)+\frac{1}{2\pi i}\oint_{q\in\check C_+}\frac{\eta^q(p_0)}{2\omega_{0,1}(q)}{\rm Rec}_{g,n+2}^{\sQ}(q,J,p),\label{dilatong,n1}
    \end{align}
    where $\check C_+$ is the resulting contour after removing the point $p$. The first term in \eqref{dilatong,n1} gives
    \begin{equation}
       \frac{1}{2\pi i}\oint_{C^{\mathfrak{p}}_{g,n+2}}\phi(p)\cdot d_p\left(\frac{\eta^p(p_0)}{2\omega_{0,1}(p)}\omega_{g,n+1}(p,J)\right)=-\frac12\omega_{g,n+1}(p_0,J),\label{dilatong,n3}
    \end{equation}
    where we used the Riemann bilinear identity (Lemma \ref{lem:RBI}) after integrating by parts. 
    On the other hand, since $\check C_+$ contains $J\cup\widetilde\cP_+$ whereas $C^{\mathfrak{p}}_{g,n+2}$ encircles the $\sigma$-conjugate of those points, we have (c.f. Appendix \ref{sec:appendix})
    \begin{equation}
        \oint_{p\in C^{\mathfrak{p}}_{g,n+2}}\oint_{q\in\check C_+}=\oint_{q\in\check C_+}\left(\oint_{p\in C^{\mathfrak{p}}_{g,n+2}}+2\pi i\Res_{p=\sigma(q)}\right).\label{contour2}
    \end{equation}
    We can now apply a similar computation to \eqref{dilatong,n0}, and obtain:
    \begin{align}
         &\frac{1}{2\pi i}\oint_{p\in C^{\mathfrak{p}}_{g,n+2}}\frac{1}{2\pi i}\oint_{q\in\check C_+}\phi(p)\cdot\frac{\eta^q(p_0)}{2\omega_{0,1}(q)}{\rm Rec}_{g,n+2}^{\sQ}(q,J,p)\nonumber\\
         &=(2-2g-n)\frac{1}{2\pi i}\oint_{q\in\check C_+}\frac{\eta^q(p_0)}{2\omega_{0,1}(q)}{\rm Rec}_{g,n+1}^{\sQ}(q,J)\nonumber\\
         &=(2-2g-n) \hat R^{\sQ}_{g,n+1}(p_0,J)\label{dilatong,n2}
    \end{align}
    where an analogous term to the second term in \eqref{dilatong,n0} disappears because such a term has no poles in $\check C_+$. 

    Combining \eqref{dilatong,n0}, \eqref{dilatong,n3}, and \eqref{dilatong,n2}, we arrive at the proposition.
\end{proof}

Notice that computations shown so far in the proof of Proposition \ref{prop:dilaton} do not depend on properties of $\bP^1$. We simply have not evaluated the second term in $\hat R^{\sQ}_{g,n+2}$ coming from $H_1(\Sigma)$ which is absent when $\Sigma=\bP^1$. We expect that this term also becomes consistent with the induction, hence we make the following conjecture:

\begin{conj}
    {\rm Proposition \ref{prop:dilaton}} extends to any hyperelliptic refined spectral curve.
\end{conj}

\begin{rem}
    The remaining task is to show the following equality:
    \begin{align}
        &\frac{1}{2\pi i}\oint_{p\in C^{\mathfrak{p}}_{g,n+2}}\phi(p)\cdot\sum_{i=1}^{\tilde g}u_i(p_0)\cdot\oint_{q\in\sA_i}\frac{{\rm Rec}_{g,n+2}^{\sQ}(q,J,p)}{2\omega_{0,1}(q)}\nonumber\\
        &=(2-2g-n)\sum_{i=1}^{\tilde g}u_i(p_0)\cdot\oint_{q\in\sA_i}\frac{{\rm Rec}_{g,n+1}^{\sQ}(q,J)}{2\omega_{0,1}(q)}.
    \end{align}
    Since the contour $C^{\mathfrak{p}}_{g,n+2}$ can be taken without crossing $\sA_i$-contour (decompose it into smaller contours if necessary), $C^{\mathfrak{p}}_{g,n+2}$ would not contain $\sigma(q)$ after exchanging the contours, unlike \eqref{contour2}. Therefore, what remains to be shown is:
    \begin{equation}
        \oint_{q\in\sA_i}\Res_{p=\sigma(q)}\phi(p)\frac{{\rm Rec}_{g,n+2}^{\sQ}(q,J,p)}{2\omega_{0,1}(q)}=0.
    \end{equation}
    It is easy to see that the contribution of $\omega_{0,2}(p,q)$ from ${\rm Rec}_{g,n+1}^{\sQ}$ vanishes, but the other terms seem nontrivial to prove.  
\end{rem}

\subsubsection{Free energy}

Recall Proposition \ref{prop:dilaton}. Notice that the right-hand side of \eqref{dilaton} makes sense even when $n=-1$ as long as $g>1$, and the left-hand side gives ``$\omega_{g,0}$''. In fact, we rather take it as the defining equation of what we call the free energy:
\begin{defin}\label{def:Fg}
    Given a hyperelliptic refined spectral curve $\cS_{\bm\mu,\bm\kappa}$, let $\omega_{g,n+1}$ be multidifferentials constructed by the refined topological recursion (Definition \ref{def:RTR}). Then, the \emph{free energy of genus} $g$ for $g\in\bZ_{>1}$ is denoted by $F_g$ and defined by
\begin{equation}
    F_g:=\omega_{g,0}:=\frac{1}{2-2g}\frac{1}{2\pi i}\oint_{C_{g,1}^{\mathfrak{p}}}\phi(p)\cdot\omega_{g,1}\label{Fg}
\end{equation}
\end{defin}

Note that as a consequence of Lemma \ref{lem:Qpoly}, the free energy $F_g$ also admits the following expansion:
    \begin{equation}
        F_g=\sum_{k=0}^{2g}\sQ^k F_g^{(k)},
    \end{equation}
    where $F_g^{(k)}$ are independent of $\sQ$, and for all $\ell\in\bZ_{\geq0}$,
    \begin{align}
        \forall g\in\bZ_{\geq0},\quad&\quad F_g^{(2\ell+1)}=0\\
        \forall g\in\bZ_{\geq0}+\frac12,\quad&\quad F_g^{(2\ell)}=0
    \end{align}

One may wonder how to define $F_0$, $F_\frac12$, and $F_1$. Since there is no effect of the refinement on $F_0$, we can take the definition given in \cite[Eq. (4.13)]{EO07}. On the other hand, the definitions of $F_\frac12$ and $F_1$ are in fact unclear, and we leave them to the future work. One possible way to define them is via the so-called variational formula which the author will report in a sequel paper \cite{O23}.

\newpage
\section{$\mathscr{Q}$-top recursion and the Nekrasov-Shatashivili limit}\label{sec:NS}

In this section, we establish a new recursion which we call the $\sQ$-top recursion. It is relevant to the so-called Nekrasov-Shatashivili limit (e.g. \cite{NS09}).

\subsection{Construction}
Since $\omega_{g,n+1}$ constructed by the refined topological recursion only polynomially depends on the refinement parameter $\sQ$ (Lemma \ref{lem:Qpoly}), there is no issue with expanding both sides of the recursion formula \eqref{RTR} in $\sQ$. For each $2g,n\in\bZ_{\geq0}$, we will focus on the $\sQ$-top term $\omega_{g,n+1}^{(2g)}$, which we denote by $\varpi_{g,n+1}:=\omega_{g,n+1}^{(2g)}$ to emphasise. It turns out that $\varpi_{g,n+1}$ can be recursively determined among themselves, without information about $\omega_{g,n+1}^{(k)}$ for $k<2g$. Note that, similar to the previous sections, we will use the letter $\check \varpi_{g,n+1}$ for the $\sQ$-top part of loop equations and $\varpi_{g,n+1}$ for the $\sQ$-top recursion in order to notationally distinguish them. Then, we have:

\begin{prop}
     Given a hyperelliptic refined spectral curve $\cS_{\bm\mu,\bm\kappa}$, assume existence of a solution of the $\sQ$-top part of hyperelliptic refined loop equations. Then, for $2g,n\in\bZ_{\geq0}$ with $2g-2+n$, $\check \varpi_{g,n+1}$ is uniquely constructed by one of the following formulae
     \begin{description}
        \item[recursion 1]
        \begin{equation}
           \check \varpi_{g,n+1}(J_0)=\frac{1}{2 \pi i}\left(\oint_{p\in C_+}-\oint_{p\in C_-}\right)\frac{\eta_\sA^p(p_0)}{4\omega_{0,1}(p)}\check {\rm Rec}_{g,n+1}^{\sQ\text{-{\rm top}}}(p,J)\label{NSR1}
        \end{equation}
        \item[recursion 2]
        \begin{align}
          \check   \varpi_{g,n+1}(J_0)=&-\frac{\check {\rm Rec}_{g,n+1}^{\sQ\text{-{\rm top}}}(p_0,J)}{2\omega_{0,1}(p_0)}+\hat R_{g,n+1}^{\sQ\text{-{\rm top}}}(p_0,J)\label{NSR2}
        \end{align}
    \end{description}
    where $C_\pm$ are the same as those in Theorem \ref{thm:RTR} and
    \begin{align}
    \check  {\rm Rec}_{g,n+1}^{\sQ\text{-{\rm top}}}(p,J):=&\sum^*_{\substack{g_1+g_2=g\\J_1\sqcup J_2=J}}\check \varpi_{g_1,n_1+1}(p,J_1)\cdot\check \varpi_{g_2,n_2+1}(p,J_2)+\sum_{t\sqcup I=J}\frac{dx(p_0)\cdot dx(t)}{(x(p_0)-x(t))^2}\cdot\check \varpi_{g,n}(p,I)\nonumber\\
     &+ dx \cdot d_p\frac{\check \varpi_{g-\frac12,n+1}(p,J)}{dx(p)}\label{RecNS}
\end{align}
and
        \begin{equation}
            \hat R_{g,n+1}^{\sQ\text{-{\rm top}}}(p_0,J)=\frac{1}{2 \pi i}\oint_{p\in \hat{C}_+}\frac{\eta_\sA^p(p_0)}{2\omega_{0,1}(p)}\cdot\check {\rm Rec}_{g,n+1}^{\sQ\text{-{\rm top}}}(p,J)+\sum_{i=1}^{\tilde g}u_i(p_0)\cdot\oint_{p\in\sA_i}\frac{\check {\rm Rec}_{g,n+1}^{\sQ\text{-{\rm top}}}(p,J)}{2\omega_{0,1}(p)}.\label{RNS{g,n+1}}
        \end{equation}
\end{prop}
\begin{proof}
    It is sufficient to show that the $\sQ$-top part of the right-hand side of the full recursion formula \eqref{RTR1} and \eqref{RTR2} respectively reduces to the above formulae \eqref{NSR1} and \eqref{NSR2}. We also note that a major difference between $\check {\rm Rec}_{g,n+1}^{\sQ}$ \eqref{Rec} and $\check {\rm Rec}_{g,n+1}^{\sQ\text{-{\rm top}}}$ \eqref{RecNS} is the $\check \omega_{g-1,n+2}$ term.

    For $g=0$ and $g=\frac12$, it is obvious by definition that
    \begin{equation}
        \check {\rm Rec}_{0,n+1}^{\sQ}= \check {\rm Rec}_{0,n+1}^{(0)}=\check {\rm Rec}_{0,n+1}^{\sQ\text{-{\rm top}}},\quad\check {\rm Rec}_{\frac12,n+1}^{\sQ}= \check {\rm Rec}_{0,n+1}^{(1)}=\check {\rm Rec}_{\frac12,n+1}^{\sQ\text{-{\rm top}}},
    \end{equation}
    where $\check {\rm Rec}_{g,n+1}^{(k)}$ denotes the $k$-th order term of the $\sQ$-expansion of $\check {\rm Rec}_{g,n+1}^{\sQ}$. Thus, as a consequence, the recursion formulae \eqref{RTR1} and \eqref{RTR2} respectively coincide with \eqref{NSR1} and \eqref{NSR2}. In particular, we find:
    \begin{equation}
      \check   \omega_{0,n+1}=\check \omega_{0,n+1}^{(0)}= \check \varpi_{0,n+1},\quad \check \omega_{\frac12,n+1}=\check \omega_{\frac12,n+1}^{(1)}= \check \varpi_{\frac12,n+1}.
    \end{equation}
    
    For general $g\in\bZ_{\geq0}$, the $\check \omega_{g-1,n+2}$ term in $ \check {\rm Rec}_{g,n+1}^{\sQ}$ depends on $\sQ$ only up to $\sQ^{2g-2}$. On the other hand, the $\check \omega_{g-\frac12,n+1}$ term contributes to $\check  {\rm Rec}_{g,n+1}^{(2g)}$ due to the explicit $\sQ$-factor in front. Thus, we have:
    \begin{align}
       \check  {\rm Rec}_{g,n+1}^{(2g)}(p,J)=& \sum^*_{\substack{g_1+g_2=g\\J_1\sqcup J_2=J}}\check \omega^{(2g_1)}_{g_1,n_1+1}(p,J_1)\cdot\check \omega^{(2g_2)}_{g_2,n_2+1}(p,J_2)+\sum_{t\sqcup I=J}\frac{dx(p_0)\cdot dx(t)}{(x(p_0)-x(t))^2}\cdot\check \omega^{(2g)}_{g,n}(p,I)\nonumber\\
        &+ dx \cdot d_p\frac{\check \omega_{g-\frac12,n+1}^{(2g-1)}(p,J)}{dx(p)}.\label{NS1}
    \end{align}
    Notice that $\check  {\rm Rec}_{g,n+1}^{(2g)}$ only involves $\check \omega_{h,m+1}^{(2h)}$, the $\sQ$-top part of each differential $\check \omega_{h,m+1}$ for $2h+m<2g+n$. Thus, one can write $ \check {\rm Rec}_{g,n+1}^{(2g)}$ only in terms of $\check \varpi_{h,m+1}$. Denoting by $ \check {\rm Rec}_{g,n+1}^{\sQ\text{-{\rm top}}}:= {\rm Rec}_{g,n+1}^{(2g)}$, we arrive at the formulae.
\end{proof}

We now define the $\sQ$-top recursion.

\begin{defin}\label{def:NS}
    Given a hyperelliptic refined spectral curve $\cS_{\bm\mu,\bm\kappa}$ together with a refinement parameter $\sQ$, consider a fundamental domain $\mathfrak{F}\subset\Sigma$ with $\partial\mathfrak{F}:=\bigcup_{i=1}^{\tilde g}(\sA_i\cup\mathscr{B}_i)$. Let $J_0:=(p_0,..,p_n)\in(\mathfrak{F})^{n+1}$ and assume that $(J_0\cup\widetilde{\cP}_+)\cap(\cR\cup\sigma(J_0\cup\widetilde{\cP}_+))=\varnothing$. Then, the \emph{$\sQ$-top recursion} is a recursive definition of multidifferentials $\varpi_{g,n+1}$ on $(\mathfrak{F})^{n+1}$ for $2g,n\in\bZ_{\geq0}$ such that $\varpi_{0,1}:=\omega_{0,1}$, $\varpi_{0,2}:=\omega_{0,2}$, $\varpi_{\frac12,1}:=\omega_{\frac12,1}$, and for $2g-2+n\geq0$:
     \begin{equation}
            \varpi_{g,n+1}(J_0):=\frac{1}{2 \pi i}\left(\oint_{p\in C_+}-\oint_{p\in C_-}\right)\frac{\eta_\sA^p(p_0)}{4\varpi_{0,1}(p)}{\rm Rec}_{g,n+1}^{\sQ\text{-{\rm top}}}(p,J),\label{NS}
        \end{equation}
        where $C_\pm$ are the same as those in Theorem \ref{thm:RTR}, ${\rm Rec}_{g,n+1}^{\sQ\text{-{\rm top}}}$ is analogously defined as \eqref{RecNS}, and we analytically continue $\varpi_{g,n+1}(J_0)$ to the region $\cR\cup\sigma(J_0\cup\widetilde{\cP}_+)$ whenever it exists. Furthermore, for $2g\in\bZ_{>2}$, the \emph{$\sQ$-top genus-$g$ free energy} is defined as
        \begin{equation}
             F_g^{\sQ\text{-{\rm top}}}:=\varpi_{g,0}:=\frac{1}{2-2g}\frac{1}{2\pi i}\oint_{C_{g,1}^{\mathfrak{p}}}\phi(p)\cdot\varpi_{g,1}
        \end{equation}
        
\end{defin}

\begin{rem}\label{rem:NS1}
    The absence of $\omega_{g-1,n+2}$ in ${\rm Rec}_{g,n+1}^{\sQ\text{-{\rm top}}}$ leads us to a stronger result than one may naively think. Namely, the $\sQ$-top recursion is recursive separately in $g$ and $n$. In particular, for fixed $n\geq0$, one can recursively define $\varpi_{g,n+1}$ for all $g$. Therefore, it is not quite `topological recursion' and we in fact removed `topological' from the name.
\end{rem}

\subsubsection{Nekrasov-Shatashivili Limit}\label{sec:justification}

We justify why the $\sQ$-top recursion (Definition \ref{def:NS}) is related to the so-called Nekrasov-Shatashivili limit in physics literature \cite{NS09}.

Let $\hbar$ be a formal parameter which typically indicates the order of quantum corrections. In physics, gauge theories in the so-called \emph{general $\Omega$-background} come with two quantum parameters $\epsilon_1,\epsilon_2$ which are related to $\hbar$ and our refinement parameter $\sQ$ by:
\begin{equation}
    \epsilon_1\cdot\epsilon_2=-\hbar^2,\quad\epsilon_1+\epsilon_2=\hbar \cdot \sQ,
\end{equation}
or equivalently, if one writes $\sQ=\beta^{\frac12}-\beta^{-\frac12}$,
\begin{equation}
    \epsilon_1=\hbar \cdot\beta^{\frac12},\quad\epsilon_2=-\hbar\cdot\beta^{-\frac12}.
\end{equation}
In this language, the \emph{self-dual limit} and the \emph{Nekrasov-Shatashivili limit} respectively denotes the following double-scaling limit:
\begin{description}
    \item[self-dual limit] $\epsilon_1+\epsilon_2=0$ while keeping $\epsilon_1\cdot\epsilon_2$ fixed.
    \item[Nekrasov-Shatashivili limit] $\epsilon_2=0$ while keeping $\epsilon_1$ fixed (or vice versa).
\end{description}

Intuitively, the self-dual limit corresponds to $\sQ=0$ while $\sQ\to\infty$ in the Nekrasov-Shatashivili limit. Then, since $\omega_{g,n}$ in the refined topological recursion framework polynomially depends on $\sQ$, the Nekrasov-Shatashivili limit does not make sense, at least naively. However, motivated from the matrix model expectation \cite{CE06}, let us considers a formal series $\omega_n$ in $\hbar$ as
\begin{equation}
    \omega_n:=\sum_{g\in\frac12\bZ_{\geq0}}\hbar^{2g}\cdot\omega_{g,n}.
\end{equation}
Then, order by order in $\hbar$, Lemma \ref{lem:Qpoly} implies that the Nekrasov-Shatashivili limit does not blow up, and furthermore, only the $\sQ$-top term of $\omega_{g,n}$ survives. In summary, we have:
\begin{equation}
    \lim_{\epsilon_2=0}^{\epsilon_1={\rm fixed}}\omega_n=\sum_{g\geq0}\epsilon_1^{2g}\cdot\omega_{g,n}^{(2g)}=\sum_{g\geq0}\epsilon_1^{2g}\cdot\varpi_{g,n}
\end{equation}
Therefore, $\varpi_{g,n}$ appears in the asymptotic expansion of $\omega_n$ in the Nekrasov-Shatashivili limit.

\subsection{Existence and Properties}

Since $\varpi_{g,n}$ is a part of the full $\omega_{g,n}$, anything proven for $\omega_{g,n}$ directly applies to $\varpi_{g,n}$. In particular, when $\Sigma=\bP^1$, Proposition \ref{thm:existence} implies that there exists a solution of the $\sQ$-top part of refined loop equations, and $\varpi_{g,n}$ constructed by the $\sQ$-top recursion (Definition \ref{def:NS}) is a symmetric residue-free multidifferential with the desired pole structure. A special property of the $\sQ$-top recursion is that we can prove existence of $\varpi_{g,1}$ for any hyperelliptic curve, not only when $\Sigma=\bP^1$. This is mainly because the recursion for $\varpi_{g,1}$ does not involve $\varpi_{g,n+2}$ as mentioned in Remark \ref{rem:NS1}.

\subsubsection{Linear loop equation}

Lemma \ref{lem:LLE1,1} implies
\begin{equation}
    \varpi_{1,1}(p_0)+\varpi_{1,1}(\sigma(p_0))=-d\frac{\Delta\varpi_{\frac12,1}(p_0)}{2\varpi_{0,1}(p_0)}.\label{NS1,1}
\end{equation}
As a generalisation of \eqref{NS1,1}, we have the following:
\begin{lem}\label{lem:LLENS1}
For $2g\in\bZ_{\geq2}$, we have
\begin{equation}
        \varpi_{g,1}(p_0)+\varpi_{g,1}(\sigma(p_0))=d\left(\sum_{k=1}^{2g-1}\sum_{g_1,...,g_k\in\frac12\bZ_{>0}}^{g_1+\cdots+g_k=g-\frac12}\frac{1}{k}\prod_{i=1}^k\left(-\frac{\Delta\varpi_{g_k,1}(p_0)}{2\varpi_{0,1}(p_0)}\right)\right).
\end{equation}
\end{lem}
\begin{proof}
    We have already seen that it holds when $g=1$. Let us also explicitly consider the case for $g=\frac32$ to see how the proof goes. For notation brevity, let us define the operator $\sI$ by $\sI f(p):=f(p)+f(\sigma(p))$, and $\sI$ makes sense no matter if $f$ is a function or a differential.
    Lemma \ref{lem:wdecomp} implies that
    \begin{align}
        \sI \varpi_{\frac32,1}(p)&=-\frac{1}{2\omega_{0,1}(p_0)}\left(2\varpi_{\frac12,1}(p)\cdot\varpi_{1,1}(p)-2\varpi_{\frac12,1}(\sigma(p))\cdot\varpi_{1,1}(\sigma(p))+dx \cdot d_p\frac{\Delta\varpi_{1,1}(p_0)}{dx(p_0)}\right)\nonumber\\
        &=-\frac{1}{2\omega_{0,1}(p_0)}\left(\Delta\varpi_{\frac12,1}(p)\cdot\sI\varpi_{1,1}(p)+\Delta\varpi_{\frac12,1}(p)\cdot\sI\varpi_{1,1}(p)+dx \cdot d_p\frac{\Delta\varpi_{1,1}(p_0)}{dx(p_0)}\right).
    \end{align}
    Recall that
    \begin{equation}
         \sI \varpi_{\frac12,1}(p)=-\frac{d_0\Delta y(p_0)}{\Delta y(p_0)}.
    \end{equation}
    Thus, by using \eqref{NS1,1} we find
    \begin{align}
        \sI \varpi_{\frac32,1}(p)=&-d\frac{\Delta\varpi_{1,1}(p_0)}{2\omega_{0,1}(p_0)}-\frac{\Delta\varpi_{\frac12,1}(p_0)}{2\omega_{0,1}(p_0)}\cdot d\left(-\frac{\Delta\varpi_{\frac12,1}(p_0)}{2\omega_{0,1}(p_0)}\right)\nonumber\\
        =&d\left(-\frac{\Delta\varpi_{1,1}(p_0)}{2\omega_{0,1}(p_0)}+\frac12\left(-\frac{\Delta\varpi_{\frac12,1}(p_0)}{2\omega_{0,1}(p_0)}\right)^2\right).
    \end{align}

    Now we proceed by induction on $2g\in\bZ_{\geq2}$. Assume that Lemma holds up to $\varpi_{g-\frac12,1}$. Then, for $\varpi_{g,1}$,  Lemma \ref{lem:wdecomp} gives:
\begin{align}
        &-2\omega_{0,1}(p_0)\cdot \sI\omega_{g,1}(p_0)\nonumber\\
        &=\sum^*_{g_1+g_2=g}\Big(\varpi_{g_1,1}(p_0)\cdot\varpi_{g_2,1}(p_0)-\varpi_{g_1,1}(\sigma(p_0))\cdot\varpi_{g_2,1}(\sigma(p_0))\Big)+dx(p_0) \cdot d_0\frac{\Delta\varpi_{g-\frac12,1}(p_0)}{dx(p_0)}\nonumber\\
        &=\frac12\sum^*_{g_1+g_2=g}\Big(\sI\varpi_{g_1,1}(p_0)\cdot\Delta\varpi_{g_2,1}(p_0)+\Delta\varpi_{g_1,1}(p_0)\cdot\sI\varpi_{g_2,1}(p)\Big)+dx(p_0) \cdot d_0\frac{\Delta\varpi_{g-\frac12,1}(p_0)}{dx(p_0)}.
    \end{align}
    The contribution from $\sI \varpi_{\frac12,1}$ will be combined with the last term which gives
    \begin{equation}
        d\left(-\frac{\Delta\varpi_{g-\frac12,1}(p_0)}{2\omega_{0,1}(p_0)}\right).
    \end{equation}
    The proof for the remaining terms is purely computational without requiring any clever trick, hence we omit it.
\end{proof}

\subsubsection{Existence}

Since we do not need to prove the symmetry for $\omega_{g,1}$, we can prove existence of a solution of the $\sQ$-top part of refined loop equations with Lemma \ref{lem:LLENS1}.

\begin{prop}\label{prop:NSg1}
    Given a hyperelliptic refined spectral curve $\cS_{\bm\mu,\bm\kappa}$ together with a refinement parameter $\sQ$ and a choice of representatives $\sA_i$ of $\cA_i$, there exists a solution of the $\sQ$-top part of refined loop equation of type $(g,1)$ for all $2g\in\bZ_{>0}$.
\end{prop}
\begin{proof}
    Lemma \ref{lem:LLENS1} implies that $\sI\omega_{g,1}$ has no residue, and in particular, we have
\begin{equation}
    \oint_{\sA_i}\sI\varpi_{g,1}=0.
\end{equation}
Furthermore, the second equation in Lemma \ref{lem:wdecomp} implies that $\oint_{\sA_i}\Delta\varpi_{g,1}=0$ due to the property of $\eta_\sA^p$. Therefore, we have
\begin{equation}
    \oint_{\sA_i}\varpi_{g,1}=0.
\end{equation}

Now by using the Riemann bilinear identity, we can bring the recursion formula \eqref{NS} we started with to another recursion formula \eqref{NSR2}. Then, when $p_0\to p\in \widetilde\cP_+$, one finds that the pole coming from the first term (the ${\rm Rec}_{g,1}^{\sQ\text{-{\rm top}}}$ term) is precisely cancelled by the corresponding contribution from $\hat R_{g,1}^{\sQ\text{-{\rm top}}}$ (recall Remark \ref{rem:cancel}). This implies that the pole of $\varpi_{g,1}$ may only lie in  $\cR^*,\sigma(\widetilde\cP_+^{(0)})$. Finally, since $\sI\varpi_{g,1}$ is residue-free and since it has no poles at $\widetilde\cP_+^{(0)}$, $\varpi_{g,1}$ itself cannot have residue anywhere. This also ensures that the integral of $\varpi_{g,1}$ vanishes along any representative in $\cA_i$. 
\end{proof}

\subsubsection{Existence for general $n$}
By repeating the same strategy as in the proof of Lemma \ref{lem:LLENS1}, one can extend the statement to the following form:
\begin{equation}
        \sI_0\varpi_{g,n+1}(p_0,J)=d_0\left(\sum_{k=1}^{2g-1}\sum_{g_1,...,g_k\in\frac12\bZ_{>0}}^{g_1+\cdots+g_k=g-\frac12}\sum_{J_1\sqcup \cdots\sqcup J_k=J}\frac{1}{k}\prod_{i=1}^k\left(-\frac{\Delta_0\varpi_{g_k,n_k+1}(p_0,J_k)}{2\omega_{0,1}(p_0)}\right)\right),\label{LLENSn}
\end{equation}
where $\sI_0$ takes the invariant part with respect to $p_0$. Therefore, the $\sI_0\varpi_{g,n+1}$ is residue-free, and we can similarly apply the arguments of the proof of Proposition \ref{prop:NSg1} to $\varpi_{g,n+1}$. Therefore, the proof of existence boils down to the proof of the symmetry of $\omega_{g,n+1}$ (recall Remark \ref{rem:cancel}).

\subsection{Quantum curves}
Let $p\in\Sigma$ be a point away from $\cR$, We introduce the following derivative notation acting on a function $f$ on $\Sigma$:
\begin{equation}
    \frac{d}{dx(p)}\cdot f(p):=\frac{1}{dx(p)}d f(p).
\end{equation}
This is none other than the contraction of $df$ by a (-1)-form $dx^{-1}$. Then, in the full refined setting, \emph{quantisation of a hyperelliptic curve via refined topological recursion} studies a second-order differential equation of the form:
\begin{equation}
   \left( \left(\hbar\frac{d}{dx(p)}\right)^2-\sum_{k\geq0}\hbar^k\cdot  Q_k(p)\right)\cdot\psi_\gamma(p;\hbar)=0,\label{QC1}
\end{equation}
where $\psi_\gamma(p;\hbar)$ is called the (\emph{perturbatuve}) \emph{wave-function}, given by
\begin{equation}
\psi_\gamma(p;\hbar):=\exp\left(\sum_{g\in\frac12\bZ_{\geq0}}\sum_{n\in\bZ_{\geq1}}\frac{\hbar^{2g-2+n}}{n!\beta^{\frac{n}{2}}}\int_\gamma^{p}\cdots\int_\gamma^{p}\omega_{g,n}\right),\label{wavefull}
\end{equation}
with an appropriate choice of $\gamma$. For example, \cite{KO22} has presented explicit forms of differential operators for a special class of genus-zero curves.

Since $\omega_{g,n}$ are meromorphic differentials on $\Sigma$, order by order in $\hbar$, $Q_k(p)$ is in general a meromorphic function on $\Sigma$. However, typically quantum curve program in the context of topological recursion considers the following stronger question: is it possible to construct an ordinary differential operator defined on the base $\bP^1$ with prescribed pole structures, with appropriate modifications to the wave-function if necessary? If such a differential operator on the base $\bP^1$ exists, it is often called the \emph{quantum curve}.

In the unrefined setting, it is already proved in \cite{Iwa19,EGF19,MO19} that when $\Sigma\neq\bP^1$, one has to in fact consider a transseries-valued wave-function. As a consequence, the resulting quantum curve is not a power series in $\hbar$ but rather it is a transseries in $\hbar$. Although it is difficult to investigate in the full refined framework at the moment of writing, in this section we will prove existence of the $\sQ$-top part of the quantum curve ($\sQ$-top quantum curve for brevity). It is worth emphasising in advance that the $\sQ$-top quantum curve indeed becomes a power series in $\hbar$, unlike transseries in the self-dual limit.

\subsubsection{Existence}

In order to prove existence of the $\sQ$-top quantum curve, let us define a sequence of meromorphic functions $S^{\sQ\text{-{\rm top}}}_k$ on $\Sigma$ for $k\in\bZ_{\geq-1}$ by
\begin{equation}
    S^{\sQ\text{-{\rm top}}}_k(p):=\frac{\varpi_{\frac{k+1}{2},1}(p_0)}{dx(p_0)}.
\end{equation}
We conventionally set $S^{\sQ\text{-{\rm top}}}_{k}=0$ for all $k<-2$. Similarly, we define $Q_0^{\sQ\text{-{\rm top}}}$ by
\begin{equation}
    Q_0^{\sQ\text{-{\rm top}}}(p_0):=Q_0(p_0):=\frac{b(x(p_0))^2-4a(x(p_0))\cdot c(x(p_0))}{4a(x(p_0))^2}
\end{equation}
where $a(x),b(x),c(x)$ appear in the defining equation of the underlying hyperelliptic curve \eqref{P(x,y)=0}. Next, we define $Q_1^{\sQ\text{-{\rm top}}}$ by
\begin{equation}
    Q_1^{\sQ\text{-{\rm top}}}(p_0):=\frac{\varpi_{0,1}(p_0)}{dx(p)\cdot dx(p_0)}\cdot\left(\sum_{p\in\widetilde\cP_+}\mu_p\eta^p_\sA(p_0)+\sum_{i=1}^{\tilde g}\kappa_i\cdot u_i(p_0)\right).
\end{equation}
Finally, we define a sequence of meromorphic functions by $Q^{\sQ\text{-{\rm top}}}_k(p)$ for $k\in\bZ_{\geq2}$ by
\begin{equation}
    Q_k^{\sQ\text{-{\rm top}}}(p_0):=\frac{2\varpi_{0,1}(p_0)\cdot \hat R^{\sQ\text{-{\rm top}}}_{\frac{k}{2},1}(p_0)}{dx(p_0)\cdot dx(p_0)}.
\end{equation}

In this notation, we can rewrite the recursion formula for $\varpi_{g,1}$ \eqref{NSR2} to that for $S^{\sQ\text{-{\rm top}}}_k$:
\begin{lem}\label{lem:WKB}
    Given a hyperelliptic refined spectral curve, the sequence of meromorphic functions $S^{\sQ\text{-{\rm top}}}_k$ satisfies the following relation for all $k\in\bZ_{\geq0}$:
    \begin{align}
        \sum_{i+j=k-2}S^{\sQ\text{-{\rm top}}}_i(p_0)\cdot S^{\sQ\text{-{\rm top}}}_j(p_0)+\frac{d}{d x(p_0)}\cdot S^{\sQ\text{-{\rm top}}}_{k-2}(p_0)=Q^{\sQ\text{-{\rm top}}}_k(p_0),
    \end{align}
    where $Q^{\sQ\text{-{\rm top}}}_k(p_0)$ is invariant under the involution $p_0\mapsto\sigma(p_0)$.
\end{lem}
\begin{proof}
For $k=0$, recall that by definition $x,y$ satisfies an algebraic equation \eqref{P(x,y)=0}. By shifting $y$ by $-b(x)/2$ and by dividing the whole equation by $a(x)$, we find the equation for $S_{-1}^{\sQ\text{-{\rm top}}}$ as
\begin{equation}
    S_{-1}^{\sQ\text{-{\rm top}}}(p_0)^2 = (\Delta y(p_0))^2=\left(y-\frac{b(x(p))}{2}\right)^2= Q_0^{\sQ\text{-{\rm top}}}(p_0).
\end{equation}
Obviously, $Q_0^{\sQ\text{-{\rm top}}}(p_0)$ is invariant under the involution $\sigma$.

For $k=1$, one can manipulate the definition of $\varpi_{\frac12,1}$ \eqref{w1/2,1} as follows
\begin{equation}
    2S^{\sQ\text{-{\rm top}}}_{-1}(p_0)\cdot S_0^{\sQ\text{-{\rm top}}}(p_0)+\frac{d}{d x(p_0)}\cdot S^{\sQ\text{-{\rm top}}}_{-1}(p_0)=Q_1^{\sQ\text{-{\rm top}}}(p_0).
\end{equation}
$Q_1^{\sQ\text{-{\rm top}}}(p_0)$ is invariant under the involution $\sigma$ by construction (recall that $\eta_\sA^p(p_0)$ and $u_i(p_0)$ are both anti-invariant under the involution $p_0\mapsto\sigma(p_0)$). 

For general $k$, this can be shown by simple manipulation of the $\sQ$-top recursion formula \eqref{NS}. And as explained in \ref{rem:cancel}, $\hat R^{\sQ\text{-{\rm top}}}_{g,1}(p_0)$ is anti-invariant under the involution $\sigma$, implying the invariance of $Q^{\sQ\text{-{\rm top}}}_{k}(p_0)$. 
\end{proof}

Lemma \ref{lem:WKB} shows that $S^{\sQ\text{-{\rm top}}}_k$ satisfies the WKB-type recursion, implying that there exists a second-order differential equation. In order to formally make a statement, let us define the so-called wave function which is only briefly mentioned at the beginning of this section:
\begin{defin}\label{def:wave}
    Let $\epsilon_1$ be a formal parameter, and choose an open contour $\gamma\subset\Sigma$ such that it does not go through poles of $\varpi_{g,1}$ for all $g\in\bZ_{\geq0}$. Let $p,q$ be the end points of $\gamma$ and we set $p,q\not\in\cR^*\cup\sigma(\widetilde{\cP}^{(0)})^*$. Then up to overall constant normalisation, we define the \emph{$\sQ$-top wave function} as a formal function of $p$ by
\begin{equation}
    \psi_\gamma^{\sQ\text{-{\rm top}}}(p;\epsilon_1):=\exp\left(\sum_{g\in\frac12\bZ_{\geq0}}\epsilon_1^{2g-1}\int^{p}_\gamma \varpi_{g,1} \right).
\end{equation}
\end{defin}
For a special class of genus-zero curves, \cite[Section 4.2]{KO22} defines an equivalent wave function which is denote by $\varphi^{\text{NS}}$. Note that a different choice of contour $\gamma$ and a different choice of the other end point $q$ only affect by normalisation of $\psi_\gamma^{\sQ\text{-{\rm top}}}(p;\epsilon)$, hence we have dropped their dependence. Then, we finally arrive at the following statement:
\begin{thm}\label{thm:NS}
    Given a hyperelliptic refined spectral curve $\cS_{\bm\mu,\bm\kappa}$, there exists a second-order differential equation, in the following form:
    \begin{equation}
        \left(\left(\epsilon_1\frac{d}{dx(p)}\right)^2-\sum_{\ell\in\bZ_{\geq0}}\epsilon_1^\ell\cdot  \bar Q_\ell^{\sQ\text{-{\rm top}}}(x(p))\right) \psi^{\sQ\text{-{\rm top}}}_\gamma(p;\epsilon_1)=0,\label{QC}
    \end{equation}
    where $\bar Q_\ell^{\sQ\text{-{\rm top}}}$ is the lift of a rational function on $\bP^1$ to $\Sigma$ satisfying $\bar  Q_\ell^{\sQ\text{-{\rm top}}}(x(p))= Q_\ell^{\sQ\text{-{\rm top}}}(p)$.
    %satisfying $\hat Q_\ell^{\sQ\text{-{\rm top}}}(x(p))=Q_\ell^{\sQ\text{-{\rm top}}}(p)$.
\end{thm}
\begin{proof}
    For $\ell\in\bZ_{\geq0}$, Lemma \ref{lem:WKB} shows that for $ Q_\ell^{\sQ\text{-{\rm top}}}(p)$ is a meromorphic function in $p$ invariant under the involution $\sigma$ for all $\ell\in\bZ_{\geq0}$. Since $x:\Sigma\to\bP^1$ is a double cover, this implies that there exists a rational function $\bar  Q_\ell^{\sQ\text{-{\rm top}}}(x)$ on the base $\bP^1$ satisfying
    \begin{equation}
       \bar  Q_\ell^{\sQ\text{-{\rm top}}}(x(p))= Q_\ell^{\sQ\text{-{\rm top}}}(p),
    \end{equation}
    where on the left-hand side we are considering the lift of a rational function $\bar Q_\ell^{\sQ\text{-{\rm top}}}(x)$ and treating it as a function on $\Sigma$. Then, since the definition of the wave function is in the WKB ansatz form, one can show order by order in $\epsilon_1$ that $\psi^{\sQ\text{-{\rm top}}}(p;\epsilon_1)$ satisfies the above second-order differential equation.
\end{proof}

\begin{defin}\label{def:QCNS}
Given a hyperelliptic refined spectral curve, consider the $\sQ$-top wave function $\psi^{\sQ\text{-{\rm top}}}_\gamma(p;\epsilon_1)$ and the associated second-order ordinary differential operator on $\Sigma$. We denote by $\hat x:=x$ and $\Delta\hat y:=\epsilon_1\cdot \frac{d}{dx}$. Then, the  \emph{$\sQ$-top quantum curve} $\hat P^{\sQ\text{-{\rm top}}}(\Delta\hat y,\hat x)$ is the projection of the differential operator to the base $\bP^1$. That is, $\hat P^{\sQ\text{-{\rm top}}}(\Delta\hat y,\hat x)$ is given by
   \begin{equation}
   \hat P^{\sQ\text{-{\rm top}}}(\Delta\hat y,\hat x):= \left(\Delta\hat y^2-\sum_{\ell\in\bZ_{\geq0}}\epsilon_1^\ell\cdot \bar  Q_\ell^{\sQ\text{-{\rm top}}}(\hat x)\right).\label{QCNS}
    \end{equation}
\end{defin}

Some literature also impose the condition on the pole structure of $\bar  Q_\ell^{\sQ\text{-{\rm top}}}(\hat x)$ to call \eqref{QCNS} a quantum curve. Here we are relaxing this condition but we will report an observation regarding the pole structure and the deformation of a refined spectral curve in a sequel paper.

\subsubsection{Justification}

Let us start with the wave function \eqref{wavefull} in the full refined setting, and consider the Nekrasov-Shatashivili limit. Recall that $\epsilon_2=0$ corresponds to the limit $\beta\to\infty$, hence $\omega_{g,n}$ for $n>1$ disappear due to the $\beta^{\frac{n}{2}}$ factor in the denominator. Furthermore, since $\omega_{g,n}$ polynomially depends on $\sQ$, up to degree $2g$, it implies that that the nonzero contributions are only coming from the $\sQ$-top term in $\omega_{g,1}^{(2g)}$. Therefore, the $\sQ$-top wave function (Definition \ref{def:wave}) should arise in the Nekrasov-Shatashivili limit of the full wave function \eqref{wavefull}. Putting another way, this suggests that if there exists a quantum curve in the full refined setting, then it should recover the $\sQ$-top quantum curve of Theorem \ref{thm:NS} by taking the Nekrasov-Shatashivili limit.

\subsubsection{Comparisons with the self-dual limit}

As shown in \cite{KO22}, full-refined quantum curves for a special class of genus-zero curves look similar to those in the self-dual sector. We note that $\omega_{g,n}$ in the refined setting has a drastically different pole structure, yet the wave function $\psi(x)$ behave in a similar way. In particular, there are no significant differences between quantum curves in the self-dual limit and those in the Nekrasov-Shatashivili limit (compare \cite{IKT18-2} and \cite{KO22}).

However, one immediately realises that the above similarity between the two limits is just an accident for $\Sigma=\bP^1$. As shown in \cite{Iwa19,EGF19,MO19}, the unrefined wave function $\psi(x)$ does not satisfy a second-order ordinary differential equation, but rather it satisfies a partial differential equation with parameter derivatives. In order to remove such parameter derivatives, a clever trick is to take a discrete Fourier transform of $\psi(x)$ with an appropriate weight, and we denote the resulting transseries by $\Psi(x)$. It has been shown that $\Psi(x)$ formally satisfies a second-order ordinary differential equation, in cost of being a transseries. This means the differential operator itself is also a transseries-valued. We note that this result has been only recently proven, initiated by the work of Iwaki \cite{Iwa19}.

On the other hand, $\sQ$-top quantum curves do not require such a complicated transseries structure. Importantly, it directly gives an ordinary differential operator without derivatives with respect to parameters. We emphasise that Theorem \ref{thm:NS} holds for a hyperelliptic refined spectral curve of any genus. Quantisation in the full refined setting is a completely open problem, and there are many things to be investigated.

\subsubsection{Interpretation about quantisation procedure}\label{sec:interpretation}

Although there are motivations originated from matrix models, the definition of the wave function \eqref{wavefull} in the unrefined setting is somewhat \emph{ad hoc}. The WKB-type recursion for the wave function differs from the Chekhov-Eynard-Orantin recursion. In particular, it is conceptually hard to grasp any feature of the unrefined quantisation. In contrast, there is a clear geometric characteristic how we should interpret $\sQ$-top quantisation.

Let us use the same notation as Definition \ref{def:QCNS}. By an appropriate gauge transformation, any second-order differential equation can be brought into the following form:
\begin{equation}
    \hat P(\Delta\hat y,\hat x)\cdot\Psi:= \left(\Delta\hat y^2-\sum_{\ell\in\bZ_{\geq0}}\epsilon_1^\ell\cdot   Q_\ell(\hat x)\right)\cdot\Psi=0,
\end{equation}
where its \emph{classical limit} gives the degree-two equation in $\Delta y$
\begin{equation}
    P(\Delta y, x)=\Delta y^2-Q_0(x)=0.\label{classical limit}
\end{equation}
Therefore, the geometric feature of $Q_\ell(\hat x)$ characterises quantisation procedure. For example, the simplest, or perhaps the \emph{canonical quantisation} of the hyperelliptic curve \eqref{classical limit} is simply given by replacing $\Delta y$ with $\Delta\hat y$ but without introducing any quantum correction to $Q_0(x)$,
\begin{equation}
    \hat P^{\text{cano}}(\Delta\hat y,\hat x)\cdot\Psi^{\text{cano}}:= \left(\Delta\hat y^2-Q_0(\hat x)\right)\cdot\Psi^{\text{cano}}=0.
\end{equation}

In order to describe the geometric aspect of $\sQ$-top quantisation, let us introduce the notion of Voros coefficients. For a contour $\gamma\subset\Sigma$ (may or may not be closed), the \emph{Voros coefficient} $V_\gamma$ of $\Psi$ along $\gamma$ is defined by
    \begin{equation}
        V_\gamma:=\frac12\int_{\gamma}\Delta d\log\Psi.
    \end{equation}
One might think of Voros coefficients as the $\epsilon_1$-extension of abelian period integrals $\int_\gamma ydx$, hence they are occasionally called \emph{quantum periods}. In general, Voros coefficients are power series in $\epsilon_1$ (strictly speaking starting with $\epsilon_1^{-1}$). Thus, the $\epsilon_1$-expansion of the Voros coefficients $ V_\gamma$ are given by
\begin{equation}
\forall k\in\bZ_{\geq-1}\qquad V_{\gamma,k}=\frac12\int_{\gamma}\Delta S_k \cdot dx,
\end{equation}
where $d\log\Psi=:\sum_{k\geq-1}\epsilon_1^k\cdot S_k \cdot dx$ and $V_\gamma=:\sum_{k\geq-1}\epsilon_1^k\cdot V_{\gamma,k}$

In order to simplify the arguments, let us focus on the case where $Q_0(x)$ is a polynomial of degree $2\tilde g+1$ with non-zero discriminant
\begin{equation}
    Q_0(x)=\sum_{i=0}^{2\tilde g+1} C_i^{(0)} x^i.
\end{equation}
In this setting $\widetilde{\cP}=\varnothing$ and $x=\infty$ is a branch point. Then, the $\sQ$-top quantisation can be interpreted as the $\epsilon_1$-deformation of the $\tilde g$ lowest coefficients $C_0^{(0)},..,C^{(0)}_{\tilde g-1}$ such that if one writes $Q_k(x)=\sum_{i=0}^{\tilde g-1} C_i^{(k)} x^i$, then $C_i^{(1)}$ is a certain linear combination of $\kappa_i$ (depending on the choice of $\cA_i$-cycles), and $C_i^{(k)}$ for $k\geq2$ are determined by the $\cA_i$-cycle normalisation condition $V_{\cA_i,k-1}=0$. Putting another way, there is a unique quantisation procedure that meets the following two conditions --- uniqueness is easy to verify hence we omit the proof:
\begin{description}
    \item[$\bm\sQ$-top1] $Q_k(x)$ is a polynomial in $x$ of degree $\tilde g-1$ for all $k\geq1$
    \item[$\bm\sQ$-top2] Voros coefficients along the $\cA_i$-cycles truncate at the order $\epsilon_1^0$:
    \begin{equation}
        V_{\cA_i}=\frac{V_{\cA_i,-1}}{\epsilon_1}+\kappa_i.
    \end{equation}
\end{description}
Such a quantisation framework is none other than $\sQ$-top quantisation. A similar quantisation procedure is discussed in the context of conformal blocks in e.g. \cite{LN21,LN22}. When $\widetilde{\cP}^{(\infty)}\neq\varnothing$, one can still characterise the $\sQ$-top quantisation by appropriately adjusting the conditions {\bf $\bm\sQ$-top1} and {\bf $\bm\sQ$-top2}.

When $\widetilde{\cP}^{(0)}\neq\varnothing$, then $\omega_{g,1}$ would have poles at $\sigma(\widetilde{\cP}^{(0)})$, hence interpretation of the $\sQ$-top quantisation becomes less straightforward. We will report more about this case in relation to the so-called variational formula and the Nekrasov twisted superpotential in a sequel paper.

\newpage

\newpage
\appendix
\section{Computational Details}\label{sec:appendix}

In this appendix, we provide detailed computations omitted in Section \ref{sec:properties}. Discussions are also omitted in \cite{KO22}, hence they will be useful to understand computations of \cite{KO22}.

\subsection{Contours and residues}\label{sec:contours}
Given a fundamental domain $\mathfrak{F}\subset\Sigma$, consider simply-connected contours $C_1, C_2\subset\mathfrak{F}$ and we compare the difference between $\oint_{p\in C_1}\oint_{q\in C_2}\omega(p,q)$ and $\oint_{q\in C_2}\oint_{p\in C_1}\omega(p,q)$ for some meromorphic bidifferential $\omega(p,q)$. Thanks to the Cauchy's integral theorem, the contour integrals reduce down to a sum of residues of $\omega(p,q)$ contained inside $C_1,C_2$ respectively. Thus, in this section we investigate a relation between  $\Res_{p=a}\Res_{q=b}$ and $\Res_{q=b}\Res_{p=a}$ for some $a,b\in\mathfrak{F}$, where $C_1$ is a closed contour encircling $a$ and $C_2$ encircling $b$.

For when $a\neq b,\sigma(b)$, it is easy to verify that
\begin{equation}
    \Res_{p=a}\Res_{q=b}=\Res_{q=b}\Res_{p=a}.
\end{equation}
This is because one can always take $C_1$ small enough that it does not intersect $C_2$. However, when $a=b \;(\not\in\cR)$, we can have two scenarios: either $C_1$ is contained inside $C_2$, or vise versa. 
Therefore, we have
 \begin{align}
     \Res_{p=a}\Res_{q=a}&=\Res_{q=a}\Res_{p=a}+\Res_{q=a}\Res_{p=q},\\
     \Res_{q=a}\Res_{p=a}&=\Res_{p=a}\Res_{q=a}+\Res_{p=a}\Res_{q=p}.\label{res1}
 \end{align}
When one extends $\{a\}$ in \eqref{res1} to $J_0\cup\widetilde{P}_+$, one arrives at \eqref{contour0}, i.e.
\begin{equation}
   \oint_{q\in C_+^q}\left(\oint_{p\in C_+^p}-2\pi i\Res_{p=q}\right)= \oint_{p\in C_+^p}\oint_{q\in C_+^q},\label{resrec}
\end{equation}
where  $C_+^q$ on the right-hand side contains $q=p$ while $C_+^p$ on the left-hand side contains $p=q$ whose contribution will be canceled by the extra residue term so that it is consistent with \eqref{res1}. This relation \eqref{resrec} is used in the proof of the symmetry of $\omega_{g,n+1}$.

Similarly, when $a=\sigma(b)\not\in\cR$, depending on whether $C_1$ is inside $C_2$ or outside, we have:
\begin{align}
     \Res_{q=\sigma(a)}\Res_{p=a}&=\Res_{p=a}\Res_{q=\sigma(a)}+\Res_{p=a}\Res_{q=\sigma(p)},\\
     \Res_{p=a}\Res_{q=\sigma(a)}&=\Res_{q=\sigma(a)}\Res_{p=a}+\Res_{q=a}\Res_{p=\sigma(q)}.
 \end{align}
 And when $a=b=\cR$, the relation will be slightly different as follows:
\begin{align}
     \Res_{q=r}\Res_{p=r}&=\Res_{p=r}\Res_{q=r}+\Res_{p=a}\Res_{q=p}+\Res_{p=a}\Res_{q=\sigma(p)},\\
     \Res_{p=r}\Res_{q=r}&=\Res_{q=r}\Res_{p=r}+\Res_{q=r}\Res_{p=q}+\Res_{q=r}\Res_{p=\sigma(q)},
 \end{align}
where they agree with \cite[Eq. (A.29)]{EO07}.

\subsubsection{Difficulties in higher genus spectral curves}
The relations shown above hold regardless of the genus of a spectral curve. This is because the discussions so far only consider local situations where $C_1$ and $C_2$ can be decomposed into \emph{non-intersecting} smaller contours.

Note that when the integrand of the double contour integral is meromorphic, i.e. when $\Sigma=\bP^1$, \eqref{contour0} (or \eqref{resrec}) can be also written as
\begin{equation}
     \left(\oint_{p\in C_+^p}-\oint_{p\in C_-^p}\right)\left(\oint_{q\in C_+^q}-\oint_{q\in C_-^q}\right)= \left(\oint_{q\in C_+^q}-\oint_{q\in C_-^q}\right)\left(\oint_{p\in C_+^p}-\oint_{p\in C_-^p}-4\pi i\Res_{p=q}\right),\label{resrec2}
\end{equation}
where $C_{\pm}^p$ on the right-hand side contains $p=q$ and $p=\sigma(q)$ respectively. That is, after exchanging the order of two recursion contours, the only difference is the last term on the right-hand side, i.e. the residue at $p=q$. As one can see from the computations below, the symmetry of $\omega_{g,n+2}$ stands if \eqref{resrec2} holds for a higher genus curve, or if all other contributions after the exchange of the contours vanish.

We note that the derivation of \eqref{resrec2} is the only place we use the property $\Sigma=\bP^1$ in the proof of symmetry. More concretely, we used the property that the sum of residues of a meromorphic differential is zero. Since $\eta^p_{\sA}(p_0)$ is not a globally defined meromorphic function on $\Sigma$ in $p$ when $\Sigma\neq\bP^1$, the above argument cannot be applied. One may still wish to show \eqref{resrec2} by a local contour analysis, e.g. \eqref{res1}. However, recall that the refined topological recursion formula takes the \emph{difference} between the $C_+^p$-integral and the $C_-^p$-integral. Thus, when one applies the recursion formula twice in the proof of symmetry of $\omega_{g,n+2}(p_0,q_0,J)$, for instance, the term in the following form appear
\begin{equation}
    \Res_{p=r}\left(\Res_{q=p}-\Res_{q=\sigma(p)}\right).
\end{equation}
For this combination, the contour with respect to $q$ depends on the location of $p$ and it necessarily intersects with the contour for $p$, hence no local contour analysis seems applicable. As a different approach, one may use the Riemann bilinear identity (Lemma \ref{lem:RBI}) and take account for nontrivial $\cA_i$ and $\cB_i$ integrals. However, it becomes harder and more complex to investigate those contributions, and we hope to solve this issue in the near future.

\subsection{Symmetry}\label{sec:sym}

As explained in Section \ref{sec:properties}, 
 we obtain \eqref{twice} after applying the refined recursion formula twice to $\omega_{g,n+2}(p_0,q_0,J)$ where $\text{Rec}^{\sQ,\text{twice}}_{g,n+2}(p,q,J)$ is given in \eqref{Recpq}. Terms in the same colour in  \eqref{Recpq} are \emph{almost} symmetric in $p$ and $q$. There are only two non-symmetric terms in $\text{Rec}^{\sQ,\text{twice}}_{g,n+2}(p,q,J)$, one is in the second black term and the other is in the second purple term. A clever trick is that we can symmetrise it by adding the following two terms:
\begin{align}
    \text{Rec}^{\sQ,\text{sym}}_{g,n+2}(p,q,J):=&\text{Rec}^{\sQ,\text{twice}}_{g,n+2}(p,q,J)\nonumber\\
    &+\sum_{t\sqcup I=J}\frac{dx(q)\cdot dx(t)}{(x(q)-x(t))^2}\cdot\frac{dx(p)\cdot dx(u)}{(x(p)-x(u))^2}\cdot\omega_{g,n}(p,I)\nonumber\\
    &+\sQ\cdot dx(q)\cdot d_q\cdot\left(\frac{1}{dx(q)}\frac{dx(q)\cdot dx(p)}{(x(q)-x(p))^2}\cdot\omega_{g-1,n+1}(p,J)\right).\label{Recsym}
\end{align}
It is crucial that the added two terms are $\sigma$-invariant quadratic differential in $q$, and they become regular at $\cR^*$ after dividing by $\omega_{0,1}(q)$. Therefore, Lemma \ref{lem:reseta} implies that they give no contributions after taking the contour integrals in $q$. As a consequence, \eqref{twice} holds even if $\text{Rec}^{\sQ,\text{twice}}_{g,n+2}(p,q,J)$ is replaced with $\text{Rec}^{\sQ,\text{sym}}_{g,n+2}(p,q,J)$ where the latter is symmetric in $p\leftrightarrow q$.

If we substitute \eqref{Recsym} to \eqref{twice} and subtract it by $\omega_{g,n+2}(q_0,p_0,J)$, almost all terms trivially cancel because $\text{Rec}^{\sQ,\text{sym}}_{g,n+2}(p,q,J)$ is symmetric in $p\leftrightarrow q$. The nontrivial terms can be assembled as follows:
\begin{align}
    &\omega_{g,n+2}(p_0,q_0,J)-\omega_{g,n+2}(q_0,p_0,J)\nonumber\\
    &=\left(\frac{1}{2\pi i} \oint_{q\in C_+}-\frac{1}{2\pi i} \oint_{q\in C_-}\right)\cdot\Bigg(\frac{\eta^{q}(p_0)\cdot\Delta\omega_{0,2}(q,q_0)-\eta^{q}(q_0)\cdot\Delta\omega_{0,2}(q,p_0)}{4\,\omega_{0,1}(q)}\cdot\omega_{g,n+1}(q,J)\nonumber\\
    &\quad\quad -2\Res_{p=q}\frac{\eta^p(p_0)}{4\omega_{0,1}(p)}\cdot\frac{\eta^q(q_0)}{4\omega_{0,1}(q)}\cdot\text{Rec}^{\sQ,\text{sym}}_{g,n+2}(p,q,J)\Bigg),\label{p0q0diff}
\end{align}
where the last term is precisely the contribution coming from the last term in \eqref{resrec2}.

The remaining task is to simplify the last line in \eqref{p0q0diff}. This is doable mainly because $\omega_{g',n'+2}(p,q,J)$ are regular at $p=q$ for any $2g',n'\in\bZ_{\geq0}$ except $g'=n'=0$. Explicitly, we have:
\begin{align}
    &\Res_{p=q}\frac{\eta^p(p_0)}{4\omega_{0,1}(p)}\cdot\frac{\eta^q(q_0)}{4\omega_{0,1}(q)}\cdot\text{Rec}^{\sQ,\text{sym}}_{g,n+2}(p,q,J)\nonumber\\
    &=\Res_{p=q}\frac{\eta^p(p_0)}{4\omega_{0,1}(p)}\cdot\frac{\eta^q(q_0)}{4\omega_{0,1}(q)}\nonumber\\
    &\cdot\Bigg(B(p,q)\sum_{g_1+g_2=g}^{\Delta}\sum_{J_1\sqcup J_2=J}\bigg(\omega_{g_1,n_1+1}(p,J_1)\cdot\omega_{g_2,n_2+1}(q,J_2)+\omega_{g-1,n+2}(p,q,J)\bigg)\nonumber\\
    &+\sQ\omega_{g-\frac12,n+1}(p,J)\cdot dx(q)\cdot d_q\cdot\left(\frac{B(p,q)}{dx(q)}\right)+\sQ\omega_{g-\frac12,n+1}(q,J)\cdot dx(p)\cdot d_p\cdot\left(\frac{B(p,q)}{dx(p)}\right)\Bigg),\label{A1}
\end{align}
where $\Delta$ in the sum denotes that $\omega_{0,2}$ is replaced with $\Delta\omega_{0,2}$ so that the second term in \eqref{Rec} is included. Evaluating the residue of the second-last line in \eqref{A1}, we have
\begin{align}
    &\frac{\eta^q(q_0)}{4\omega_{0,1}(q)}\cdot d_p\left(\frac{\eta^p(p_0)}{4\omega_{0,1}(p)}\cdot\left(\omega_{g_1,n_1+1}(p,J_1)\cdot\omega_{g_2,n_2+1}(q,J_2)+\omega_{g-1,n+2}(p,q,J)\right)\right)\bigg|_{p=q}\nonumber\\
    &=\frac{\eta^q(q_0)}{4\omega_{0,1}(q)}\cdot d_q\left(\frac{\eta^q(p_0)}{4\omega_{0,1}(q)}\right)\cdot\left(\omega_{g_1,n_1+1}(q,J_1)\cdot\omega_{g_2,n_2+1}(q,J_2)+\omega_{g-1,n+2}(q,q,J)\right)\nonumber\\
    &\;\;\;\;+\frac{\eta^q(q_0)}{4\omega_{0,1}(q)}\cdot \frac{\eta^p(p_0)}{4\omega_{0,1}(p)}\cdot d_p\left(\omega_{g_1,n_1+1}(p,J_1)\cdot\omega_{g_2,n_2+1}(q,J_2)+\omega_{g-1,n+2}(p,q,J)\right)\bigg|_{p=q},\label{A2}
\end{align}
where we are abusing the notation of $d_p$ not acting on functions, but the computations hold. We now introduce another clever trick which is what \cite[Eq. A.21]{EO07} means by integrating \emph{half} by parts. Notice that
\begin{align}
    &\sum_{g_1+g_2=g}^{\Delta}\sum_{J_1\sqcup J_2=J}d_p\left(\omega_{g_1,n_1+1}(p,J_1)\cdot\omega_{g_2,n_2+1}(q,J_2)+\omega_{g-1,n+2}(p,q,J)\right)\bigg|_{p=q}\nonumber\\
    &=\sum_{g_1+g_2=g}^{\Delta}\sum_{J_1\sqcup J_2=J}d_q\left(\omega_{g_1,n_1+1}(p,J_1)\cdot\omega_{g_2,n_2+1}(q,J_2)+\omega_{g-1,n+2}(p,q,J)\right)\bigg|_{p=q}.
\end{align}
Then one can integrate the $d_q$ by parts (with the integration in $q$) of the last line in \eqref{A2}, and we have
\begin{align}
    \sum_{g_1+g_2=g}^{\Delta}\sum_{J_1\sqcup J_2=J}\eqref{A2}=&\sum_{g_1+g_2=g}^{\Delta}\sum_{J_1\sqcup J_2=J}\left(\frac{\eta^q(q_0)}{4\omega_{0,1}(q)}\cdot d_q\left(\frac{\eta^q(p_0)}{4\omega_{0,1}(q)}\right)- d_q\left(\frac{\eta^q(q_0)}{4\omega_{0,1}(q)}\right)\cdot\frac{\eta^q(p_0)}{4\omega_{0,1}(q)}\right)\nonumber\\
    &\cdot\left(\omega_{g_1,n_1+1}(q,J_1)\cdot\omega_{g_2,n_2+1}(q,J_2)+\omega_{g-1,n+2}(q,q,J)\right)\nonumber\\
    =&\sum_{g_1+g_2=g}^{\Delta}\sum_{J_1\sqcup J_2=J}\frac{\eta^{q}(p_0)\cdot\Delta\omega_{0,2}(q,q_0)-\eta^{q}(q_0)\cdot\Delta\omega_{0,2}(q,p_0)}{16\,\omega_{0,1}(q)^2}\nonumber\\
    &\cdot\left(-\omega_{g_1,n_1+1}(q,J_1)\cdot\omega_{g_2,n_2+1}(q,J_2)-\omega_{g-1,n+2}(q,q,J)\right),\label{Bterms}
\end{align}
up to total derivatives in $q$ which disappear when one considers the contour integrals in $q$.

On the other hand, the last line of \eqref{A1} involves trickier integral by parts. After evaluating the residue at $p=q$, the first term in the last line of \eqref{A1} gives
\begin{align}
    &\sQ\cdot\frac{\eta^q(q_0)}{4y(q)}d_q\left(\frac{1}{dx(q)}d_q\left(\frac{\eta^q(p_0)}{4\omega_{0,1}(q)}\cdot\omega_{g-\frac12,n+1}(q,J)\right)\right)\nonumber\\
    &=-\sQ\cdot d_q\left(\frac{\eta^q(q_0)}{4y(q)}\right)\cdot\frac{1}{dx(q)}d_q\left(\frac{\eta^q(p_0)}{4y(q)}\cdot\frac{\omega_{g-\frac12,n+1}(q,J)}{dx(q)}\right),\label{1stlast}
\end{align}
where we integrated the first $d_q$ by parts at the equality, thus \eqref{1stlast} holds up to total derivative terms. As for the second term in the last line of \eqref{A1}, here are the steps of simplification: we first integrate $d_p$ by parts which moves $d_p$, next evaluate the residue at $p=q$ which gives another $d_q$, and then integrate the $d_q$ by parts again. By following these steps, we find:
\begin{align}
   &\sQ\cdot\Res_{p=q}d_p\left(\frac{\eta^p(p_0)}{4y(p)}\right)\cdot\frac{\eta^q(q_0)}{4\omega_{0,1}(q)}\cdot\omega_{g-\frac12,n+1}(q,J)\cdot\frac{B(p,q)}{dx(p)}\nonumber\\
   &=-\sQ\cdot\frac{\eta^q(q_0)}{4y(q)}\cdot\frac{\omega_{g-\frac12,n+1}(q,J)}{dx(q)}\cdot d_q\left(\frac{1}{dx(q)}d_q\left(\frac{\eta^q(p_0)}{4y(q)}\right)\right)\nonumber\\
   &=\sQ\cdot d_q\left(\frac{\eta^q(q_0)}{4y(q)}\cdot\frac{\omega_{g-\frac12,n+1}(q,J)}{dx(q)}\right)\cdot \frac{1}{dx(q)}d_q\left(\frac{\eta^q(p_0)}{4y(q)}\right),\label{2ndlast}
\end{align}
where the last equality holds up to total derivative terms in $q$.

Several terms in \eqref{1stlast} and \eqref{2ndlast} cancel when we add them to each other, and we obtain
\begin{align}
    \eqref{1stlast} + \eqref{2ndlast}=\frac{\eta^{q}(p_0)\cdot\Delta\omega_{0,2}(q,q_0)-\eta^{q}(q_0)\cdot\Delta\omega_{0,2}(q,p_0)}{16\,\omega_{0,1}(q)^2}\cdot (-\sQ)\cdot d_q\left(\frac{\omega_{g-\frac12,n+1}(q,J)}{dx(q)}\right).\label{dBterms}
\end{align}
Combining \eqref{Bterms} and \eqref{dBterms}, we get 
\begin{align}
   \eqref{A1}=& \eqref{Bterms} + \eqref{dBterms}\nonumber\\
   =&-\frac{\eta^{q}(p_0)\cdot\Delta\omega_{0,2}(q,q_0)-\eta^{q}(q_0)\cdot\Delta\omega_{0,2}(q,p_0)}{16\,\omega_{0,1}(q)^2}\cdot \text{Rec}^{\sQ}_{g,n+1}(q,J).
\end{align}
Therefore, we finally arrive at
\begin{align}
    &\omega_{g,n+2}(p_0,q_0,J)-\omega_{g,n+2}(q_0,p_0,J)\nonumber\\
    &=\frac{1}{2\pi i} \oint_{C_{{\rm Rec}}^p}\cdot\frac{\eta^{p}(p_0)\cdot\Delta\omega_{0,2}(p,q_0)-\eta^{p}(q_0)\cdot\Delta\omega_{0,2}(p,p_0)}{4\,\omega_{0,1}(p)}\cdot \hat R_{g,n+1}(p,J).\label{symfinal}
\end{align}
This vanishes by applying Lemma \ref{lem:reseta} and by the fact that the integrand is regular at $\cR$

\subsubsection{Pole structure}\label{sec:residue}
It mostly follows from Remark \ref{rem:cancel}, hence we make a comment about ineffective ramification points. Suppose there exits a solution of refined loop equations of type $(g',n'+1)$ for all $2g'-2+n'\leq\chi$. Since $\omega_{g,n+2}(p_0,q_0,J)$ is singular as $p_0\to\sigma(q_0)$, one may concern that $\omega_{g,n+2}(p,p,J)$ in $\text{Rec}_{g,n+1}^{\sQ}(p,J)$ becomes singular at ineffective ramification points. However, we can easily show by induction that the following differential in $p_0$
\begin{equation}
    \frac{\omega_{g',k+m+1}(p_0,..,p_0,p_1,..,p_m)}{\omega_{0,1}(p_0)^k}
\end{equation}
is holomorphic at ineffective ramification points for all $2g'-2+k+m\leq\chi$. In particular, this implies that
\begin{equation}
    \frac{\omega_{g-1,n+2}(p_0,p_0,J)}{\omega_{0,1}(p_0)}
\end{equation}
in the recursion formula \eqref{RTR2} is regular at ineffective ramification points for $2g-2+n=\chi+1$. This ensures that $\omega_{g,n+1}(p_0,J)$ has no pole in $p_0$ at ineffective ramification points. Furthermore, since $\omega_{0,1}$ only appears in the denominator of the recursion formula \eqref{RTR2}, it is straightforward to show by induction that there is no pole at $\widetilde{\cP}^{(\infty)}$.

\subsubsection{Residues}\label{sec:wg1}
The recursion formula \eqref{RTR2} clearly shows that $\omega_{g,n+2}(p_0,q_0,J)$ has no residues with respect to $q_0$, due to the induction ansatz. Then, since we have already proven that $\omega_{g,n+2}(p_0,q_0,J)$ is symmetric up to $2g-2+n=\chi$, this implies $\omega_{g,n+2}(p_0,q_0,J)$ is residue free in any variable.

Note that the above argument does not work for $\omega_{g,1}$ with $2g-2=\chi+1$, hence we need a different approach. Ideally, we would like to show that there exists a meromorphic function $f_{g,1}$ such that
\begin{equation}
    \omega_{g,1}(p_0)+\omega_{g,1}(\sigma(p_0)=d f_{g,1}(p_0).\label{wg11}
\end{equation}
This is sufficient to show that there is no residue because $\omega_{g,1}(p_0)$ does not have poles at $\widetilde{\cP}_+$. At the moment of writing, we can show an analogue of \eqref{wg11} in the unrefined sector (Lemma \ref{lem:LLE}) and the $\sQ$-top part (Proposition \ref{prop:NSg1}). However, we do not know how to prove \eqref{wg11} with full generality from the refined topological recursion formula \eqref{RTR}, hence we proceed as follows (this is the same trick as \cite{KO22}):
\begin{enumerate}
    \item For $g$ with $2g-2=\chi+1$, we leave the residue-free property of $\omega_{g,1}$ unproven for the moment, and consider $\omega_{g,2}(p_0,q_0)$.
    \item Prove the symmetry $\omega_{g,2}(p_0,q_0)$ and show that it has the expected pole structure. In particular, $\omega_{g,2}$ is residue-free by the same argument as above.
    \item Note that this can be done without showing the residue-free property of $\omega_{g,1}$ -- the recursion formula \eqref{RTR2} contains $\omega_{g,1}(p_0)$ but not $\omega_{g,1}(q_0)$.
    \item Show the dilaton equation \eqref{dilaton} between $\omega_{g,1}(p_0)$ and $\omega_{g,2}(p_0,q_0)$. Again, we do not need the residue-free property of $\omega_{g,1}(p_0)$.
    \item Prove that $\omega_{g,1}(p_0)$ has no residue by using the dilaton equation (c.f. \cite[Proposition 2.27]{KO22}. This shows existence of a solutino of type $(g,1)$.
\end{enumerate}

Let us emphasise once again that the above steps (2)-(4) can be done without knowing the residues of $\omega_{g,1}$, hence there is no breakdown of the logic. On the other hand, this proof is conceptually intertwined, and in some sense, artificial. We believe that a deeper understanding of refined loop equations will lead us to a simpler, more natural proof.

\subsection{Comments on Higher-ramified generalisations}\label{sec:comments}

The purpose of this section is not to prove any statements, but rather to share the author's views towards refinements with the readers.

The global involution $\sigma$ plays an essential role everywhere in the present paper. As a consequence, one may find it difficult to generalise the refined recursion formula beyond hyperelliptic curves. Let us discuss a potential extension for a refinement of the higher-ramified recursion of Bouchard-Eynard \cite{BE12}. 

\subsubsection{Spectral curves and Unstable multidifferentials}
Suppose the underlying curve is given in the form:
\begin{equation}
    y^r-P(x)=0,\label{higherrami1}
\end{equation}
for $r\geq2$ and $P$ is a rational function of $x$. We further assume that $\Sigma=\bP^1$ for simplicity, In this case, there exists a global automorphism $\sigma:\Sigma\to\Sigma$ which keeps $x$ invariant while acts on $y$ by $\sigma:y\mapsto e^{2\pi i/r}y$. The set of $\sigma$-fixed points then coincides with the set of ramification points $\cR$. The definition of $\omega_{0,1}$ is canonical:
\begin{equation}
    \omega_{0,1}(p_0):=y(p_0)dx(p_0).
\end{equation}

Recall that we slightly modify the definition of $\omega_{0,2}$ from the standard choice in Definition \ref{def:unstable}. In particular, we define it in such a way that $\omega_{0,2}(p_0,p_1)$ is regular at $p_0=p_1$ but singular at $p_0=\sigma(p_1)$. In the higher-ramifed setting, we respect this property and propose to define $\omega_{0,2}$ by
\begin{equation}
    \omega_{0,2}(p_0,p_1):=B(p_0,p_1)-\frac{dx(p_0)\cdot dx(p_1)}{(x(p_0)-x(p_1))^2}.
\end{equation}
This definition indeed reduces to \eqref{w02} when $r=2$. Note that $B$ on a higher-ramified curve satisfies the following instead of \eqref{BB}:
\begin{equation}
    \sum_{k=0}^{r-1}B(\sigma^k(p_0),p_1)=\frac{dx(p_0)\cdot dx(p_1)}{(x(p_0)-x(p_1)^2}.
\end{equation}
In particular, $\omega_{0,2}(p_0,p_1)$ is regular at $p_0=p_1$, but has a double pole at $p_0=\sigma^k(p_1)$ for any $k\in\{1,..,r-1\}$. 

As for $\omega_{\frac12,1}$, let $\widetilde\cP$ be the set of zeroes and poles of $ydx$ that are not in $\cR$ which would be decomposed into $\widetilde\cP=\sqcup_{k=0}^{r-1}\sigma^k(\widetilde\cP_0)$ with a certain choice $\widetilde\cP_0\subset\widetilde\cP$. Then, it is natural to propose the definition of $\omega_{\frac12,1}$ as
\begin{equation}
    \omega_{\frac12,1}(p_0):=\frac{r-1}{2}\sQ\left(-\frac{dy(p_0)}{y(p_0)}+\sum_{p\in\widetilde{\cP}_0}\sum_{k=0}^{r-1}\mu_{\sigma^k(p)}\cdot\eta^{{\sigma^k(p)},b}(p_0)\right)
\end{equation}
where $b\in\Sigma$ is any base point, $\mu_p\in\bC$ is a parameter assigned for each $p\in\widetilde{P}$, and we require $\sum_{k=0}^{r-1}\mu_{\sigma^k(p)}=0
$ so that $\omega_{\frac12,1}$ is independent of the base point $b$. When $r=2$, in particular, the definition effectively reduces to the definition \eqref{w1/2,1} by identifying $\mu_p-\mu_{\sigma(p)}$ with $\mu_p$ in a hyperelliptic refined spectral curve $\cS_{\bm\mu,\bm\kappa}$. Hence imposing the condition $\sum_{k=0}^{r-1}\mu_{\sigma^k(p)}=0
$ would not contradict to our main discussions.

Having these mentioned, a plausible definition of a genus-zero refined spectral curve of degree $r$ would be the following:
\begin{defin}
    A \emph{genus-zero refined spectral curve of degree $r$}, which we denote by $\cS_{\bm\mu,\bm\kappa}^{(r)}$,  consists of:
    \begin{itemize}
        \item $\Sigma=\bP^1$
        \item $(x,y)$: two meromorphic functions satisfying an equation of the form \eqref{higherrami1}
        \item $\widetilde\cP_0\subset\widetilde\cP$: a choice of the decomposition $\widetilde\cP=\sqcup_{k=0}^{r-1}\sigma^k(\widetilde\cP_0)$ and the associated parameter $\mu_p\in\bC$ for all $p\in\widetilde\cP$ such that $\sum_{k=0}^{r-1}\mu_{\sigma^k(p)}=0
$.
    \end{itemize}
\end{defin}

\subsubsection{Stable multidifferentials and the Heisenberg vertex operator algebra $\cH(\mathfrak{sl}_r)$.}

As a reasonable generalisation beyond the hyperelliptic case, we aim for constructing stable differentials $\omega_{g,n+1}$ for $2g,n\in\bZ_{\geq0}$ with $2g-2+n\geq0$ with the following properties:
\begin{itemize}
    \item $\omega_{g,n+1}$ is a meromorphic, symmetric, residue-free multidifferential,
    \item $\omega_{g,n+1}$ is regular as $p_0\to p\in J_0\cup\widetilde\cP_0$,
    \item $\omega_{g,n+1}$ may be singular as $p_0\to p\in\cR^*\cup\bigcup_{i=1}^{r-1}\sigma^i(J\cup\widetilde\cP_0)$,
    \item $\omega_{g,n+1}$ solves the ``higher-ramified refined loop equation of type $(g,n+1)$''.
\end{itemize}
The second condition may be considered as a geometric realisation of the notion of ``physical sheet'' in the higher-ramified setting. Note that, for example when $r=3$, if $\omega_{g,2}(p_0,p_1)$ has a pole at $p_0=\sigma(p_1)$, then the symmetry condition implies that $\omega_{g,n+1}$ necessarily has a pole at $p_0=\sigma^2(p_1)$ too. Thus, the above pole structure is in fact a reasonable extension from the hyperelliptic case.

An immediate question is: what are higher-ramified refined loop equations? Although the author does not have a concretely clear understanding at the moment of writing, it is worth sharing a possible approach.

Recall that the hyperelliptic refined topological recursion in the present paper is inspired by a $\beta$-deformation of the Virasoro algebra, or more concretely $\cH(\mathfrak{sl}_2)$ -- the Heisenberg vertex operator algebra of $\mathfrak{sl}_2$. We note that it is \emph{not} $\cH(\mathfrak{gl}_2)$ which is indeed important and is related to the breakdown of the linear loop equation (Lemma \ref{lem:LLE}) in the refined setting. Therefore, it is only natural to expect that the higher-ramified recursion should be compatible with $\cH(\mathfrak{sl}_r)$, but not $\cH(\mathfrak{gl}_r)$, unlike the unrefined Bouchard-Eynard recursion \cite{BE12} which is dual to twisted modules of $\cH(\mathfrak{gl}_r)$ \cite{BBCCN18}.

We focus on the case $r=3$. Let $J_1,J_2,J_3$ be the three Heisenberg fields of $\cH(\mathfrak{gl}_3)$, then for $i\in\{1,2,3\}$ the $\sQ$-deformed $W_i$-fields of $\cH(\mathfrak{gl}_3)$ are respectively given by\footnote{See \cite[Section 4]{BBCC21} which we use almost the same notation, except $(\hbar,\alpha_0)$ in \cite{BBCC21} correspond to $(\hbar^2,\sQ)$ in the present paper.}
\begin{align}
    W_1=&:J_1+J_2+J_3:,\label{W1}\\
    W_2=&:J_1J_2+J_3J_1+J_2J_3+\hbar\sQ \partial(J_2+2J_3):,\\
    W_3=&:J_1J_2J_3+\hbar\sQ(J_1\partial J_3+\partial(J_2J_3))+\hbar^2\sQ^2\partial^2J_3:\label{W3},
\end{align}
where $:\cdots:$ denotes normal orderings. In order to reduce $\cH(\mathfrak{gl}_3)$ to $\cH(\mathfrak{sl}_3)$, let us consider the change of basis as follows:
\begin{align}
    H_0:=\frac{1}{\sqrt{3}}(J_1+J_2+J_3),\quad H_1:=\frac{1}{\sqrt{6}}(2J_1-J_2-J_3),\quad H_2:=\frac{1}{\sqrt{2}}(J_2-J_3).
\end{align}
Note that $H_i$ for $i\in\{0,1,2\}$ commute with each other and are normalised appropriately. Then one finds $W_2(\mathfrak{sl}_3)$ and $W_3(\mathfrak{sl}_3)$ by expressing \eqref{W1}-\eqref{W3} in terms of $H_0,H_1,H_2$ and setting $W_1=H_0=0$:
\begin{align}
    W_2(\mathfrak{sl}_3)&=-\frac12:H_1^2+H_2^2+\sQ(\sqrt{6}H_1+\sqrt{2}H_2):,\\
    W_3(\mathfrak{sl}_3)&=\frac{1}{3\sqrt{6}}:H_1^3-3H_2H_3^2-\sqrt{6}\hbar\sQ(\sqrt{3}H_1+3H_2)\partial H_2-3\hbar^2\sQ^2(\partial^2H_1+\sqrt{3}\partial^2H_2):
\end{align}

The last step is to decode information about refined loop equations from $W_2(\mathfrak{sl}_2)$ and $W_3(\mathfrak{sl}_2)$. This involves a careful consideration because terms like $\omega_{g-2,n+3}(p,p,\sigma(p),J)$ may appear which at first glance diverges according to the prescribed pole structure. We hope to return to further investigation on this approach in the future. It is perhaps simpler to consider a higher-ramified analogue of the $\sQ$-top recursion, because the recursion may not involve such terms (c.f. no $\omega_{g-1,n+2}$ in the $\sQ$-top recursion).

\subsection{New Degrees of Freedom}
Let us discuss interesting new degrees of of the free energy $F_g$ freedom that only appears in the refined setting.

Suppose we shift the primitive $\Phi(p)$ by some function $U(x(p))$ in the definition of the free energy \eqref{Fg}. In the unrefined setting, Lemma \ref{lem:LLE} (linear loop equation) implies that such a shift is irrelevant and $F_g$ receives no contributions from $U(x(p))$. As discussed in Section \ref{sec:existence}, however, there is no refined analogue of Lemma \ref{lem:LLE}, hence $U(x(p))$ will modify $F_g$. On the other hand, the definition of the free energy (Definition \ref{def:Fg}) is consistent with the dilaton equation \eqref{dilaton} (Proposition \ref{prop:dilaton}). Therefore, one should ask: can we find $U(x(p))$ that keeps the dilaton equation unchanged but modifies the free energy $F_g$? If such $U(x(p))$ exists, then it is only natural to define the modified free energy $F_g^{U}$ by
\begin{equation}
    F_g^{U}:=\frac{1}{2-2g}\left(\sum_{r\in\mathcal{R}^*}\underset{p=r}{{\rm Res}}+\sum_{r\in\widetilde\cP_-^0}\underset{p=r}{{\rm Res}}\right)\Big(\Phi(p)+U(x(p))\Big)\cdot\omega_{g,1}(p).\label{FgU}
\end{equation}
Although a full classification is to be investigated, let us give a concrete example to indicate that such $U$ surprisingly exists.

Let us consider the genus-zero degree-two curve given by
\begin{equation}
   P(x,y)= 4x^2y^2-(x^2+4\lambda_{\infty}x+4\lambda_{0}^2)
\end{equation}
which appears in e.g. \cite{MS15,IKT18-2}. Note that $\widetilde{\cP}^{(0)}$ is empty while there are four points in $\widetilde{\cP}^{(\infty)}$, and we consider the corresponding refined spectral curve $\cS_{\bm\mu}$. We explicitly checked that if we define\footnote{This form is inspired by a matrix model investigation e.g. \cite{MS15}.}
\begin{equation}
    U(x):=\alpha \log x,\quad\Phi^U(p):=\Phi(p)+U(x(p)),
\end{equation}
then for a general $\alpha\in\bC$, we have
\begin{equation}
\left(\sum_{r\in\cR}\Res_{p=r}+\Res_{\sigma(p_1)}\right)\Phi^U(p)\cdot\omega_{\frac12,2}(p,p_1)=0.
\end{equation}
Furthermore, repeating the same technique as the proof of Proposition \ref{prop:dilaton}, one can show that the dilaton equation holds even after replacing $\Phi$ with $\Phi^U$.

As a curious observation, $U$ \emph{does} change $F_g^U$. For example we have
\begin{equation}
    F_2^{U}-F_2=\alpha\cdot\mu_{0}\cdot\sQ^2\frac{1-2\mu_{0}^2\cdot\sQ^2}{48\lambda_{0}^3}
\end{equation}
where $\mu_0$ is associated with one of the zeroes of $x$. Note that the difference disappears when $\sQ=0$ as expected.

On the one hand, the canonical definition of $F_g$ (Definition \ref{def:Fg}) requires no additional information than the corresponding refined spectral curve $\cS_{\bm\mu}$, hence it is geometrically the most natural form. However, on the other hand, this feature is peculiar in the refined setting, and one may suspect that existence of such a nontrivial $U$ means something deeper. It is an open question when such $U$ exists and what the role of $F_g^U$ in other subjects in mathematics and physics.

\newpage
\begin{align}
   \text{Rec}^{\sQ,\text{twice}}_{g,n+2}(p,q,J):=&\sum_{g_1+g_2+g_3=g}\sum_{J_1\sqcup J_2\sqcup J_3=J}2\omega_{g_1,n+1}(p,J_1)\cdot2\omega_{g_2,n_2+1}(q,J_2)\cdot\omega_{g_3,n_3+2}(q,p,J_3)\nonumber\\
    &+\textcolor{red}{\sum_{g_1+g_2=g}\sum_{J_1\sqcup J_2=J}2\omega_{g_1,n+1}(p,J_1)\cdot\sum_{t\sqcup I_2= p_\sqcup  J_2}\frac{dx(q)\cdot dx(t)}{(x(q)-x(t))^2}\cdot\omega_{g_2,n_2+1}(q,I_2)}\nonumber\\
    &+\textcolor{cyan}{\sum_{g_1+g_2=g}\sum_{J_1\sqcup J_2=J}2\omega_{g_1,n+1}(p,J_1)\cdot\omega_{g_2-1,n_2+3}(q,q,p,J_2)}\nonumber\\
    &+\textcolor{blue}{\sum_{g_1+g_2=g}\sum_{J_1\sqcup J_2=J}2\omega_{g_1,n+1}(p,J_1)\cdot\sQ \cdot dx(q)\cdot d_q\cdot\left(\frac{\omega_{g_2,n_2+2}(q,p,J_2)}{dx(q)}\right)}\nonumber\\
    &+\textcolor{red}{\sum_{t\sqcup I=J}\frac{dx(p)\cdot dx(t)}{(x(p)-x(t))^2}\cdot\sum_{g_1+g_2=g}\sum_{I_1\sqcup I_2=I}2\omega_{g_2,n_2+1}(q,I_1)\cdot\omega_{g_3,n_3+2}(q,p,I_2)}\nonumber\\
    &+\sum_{t\sqcup I=J}\frac{dx(p)\cdot dx(t)}{(x(p)-x(t))^2}\cdot\sum_{u\sqcup \hat I= p\sqcup I}\frac{dx(q)\cdot dx(u)}{(x(q)-x(u))^2}\cdot\omega_{g,n}(q,\hat I)\nonumber\\
    &+\textcolor{orange}{\sum_{t\sqcup I=J}\frac{dx(p)\cdot dx(t)}{(x(p)-x(t))^2}\cdot\omega_{g-1,n+2}(q,q,p,I)}\nonumber\\
    &+\textcolor{purple}{\sum_{t\sqcup I=J}\frac{dx(p)\cdot dx(t)}{(x(p)-x(t))^2}\cdot\sQ\cdot dx(q)\cdot d_q\left(\frac{\omega_{g-\frac12,n+1}(q,p,I)}{dx(q)}\right)}\nonumber\\
    &+\sum_{g_1+g_2=g}\sum_{J_1\sqcup J_2=J}2\omega_{g_1,n+2}(q,p,J_1)\cdot\omega_{g_2-1,n_2+2}(q,p,J_2)\nonumber\\
    &+\textcolor{cyan}{\sum_{g_1+g_2=g}\sum_{J_1\sqcup J_2=J}2\omega_{g_1,n+1}(q,J_1)\cdot\omega_{g_2-1,n_2+3}(q,p,p,J_2)}\nonumber\\
    &+\textcolor{orange}{\sum_{t\sqcup I= p \sqcup p\sqcup  J}\frac{dx(q)\cdot dx(t)}{(x(q)-x(t))^2}\cdot\omega_{g-1,n+1}(q,I)}\nonumber\\
    &+\omega_{g-2,n+4}(q,q,p,p,J)\nonumber\\
    &+\textcolor{green}{\sQ \cdot dx(q)\cdot d_q\left(\frac{\omega_{g-\frac32,n+3}(q,p,p,J)}{dx(q)}\right)}\nonumber\\
    &+\textcolor{blue}{\sQ\cdot dx(p)\cdot d_p\cdot\left(\sum_{g_1+g_2=g}\sum_{J_1\sqcup J_2=J}\frac{2\omega_{g_1,n_1+1}(q,J_1)\cdot\omega_{g_2-\frac12,n_2+2}(q,p,J_2)}{dx(p)}\right)}\nonumber\\
    &+\textcolor{purple}{\sQ\cdot dx(p)\cdot d_p\cdot\left(\frac{1}{dx(p)}\sum_{t\sqcup I=p\sqcup J}\frac{dx(q)\cdot dx(t)}{(x(q)-x(t))^2}\cdot\omega_{g-1,n+1}(q,I)\right)}\nonumber\\
    &+\textcolor{green}{\sQ \cdot dx(p)\cdot d_p\cdot \left(\frac{\omega_{g-\frac32,n+3}(q,q,p,J)}{dx(p)}\right)}\nonumber\\
    &+\sQ^2\cdot dx(p)\cdot d_p \cdot dx(q)\cdot d_q \left(\frac{\omega_{g-1,n+2}(q,p,J)}{dx(p)}\right)\label{Recpq}
\end{align}

\newpage

\newpage
%\bibliographystyle{plainurl}
%\bibliography{ref}

\printbibliography

@article{Iwa19,
    author = "Iwaki, Kohei",
    title = "{2-Parameter $\tau $-Function for the First Painlev\'e Equation: Topological Recursion and Direct Monodromy Problem via Exact WKB Analysis}",
    eprint = "1902.06439",
    archivePrefix = "arXiv",
    primaryClass = "math-ph",
    journal = "Commun. Math. Phys.",
    volume = "377",
    number = "2",
    pages = "1047--1098",
    year = "2020"
}

@article{IS15,
    author = "Iwaki, Kohei and Saenz, Axel",
    title = "{Quantum Curve and the First Painleve Equation}",
    eprint = "1507.06557",
    archivePrefix = "arXiv",
    primaryClass = "math-ph",
    journal = "SIGMA",
    volume = "12",
    pages = "011",
    year = "2016"
}

@article{CE05,
    author = "Chekhov, L. and Eynard, B.",
    title = "{Hermitean matrix model free energy: Feynman graph technique for all genera}",
    eprint = "hep-th/0504116",
    archivePrefix = "arXiv",
    reportNumber = "SPHT-T05-046",
    journal = "JHEP",
    volume = "03",
    pages = "014",
    year = "2006"
}

@article{BBCCN18,
    author = {Borot, Ga\"etan and Bouchard, Vincent and Chidambaram, Nitin K. and Creutzig, Thomas and Noshchenko, Dmitry},
    title = "{Higher Airy structures, W algebras and topological recursion}",
    eprint = "1812.08738",
    archivePrefix = "arXiv",
    primaryClass = "math-ph",
    month = "12",
    year = "2018"
}

@article{IK21,
    author = "Iwaki, Kohei and Kidwai, Omar",
    title = "{Topological recursion and uncoupled BPS structures II: Voros symbols and the $\tau$-function}",
    eprint = "2108.06995",
    archivePrefix = "arXiv",
    primaryClass = "math-ph",
    month = "8",
    year = "2021"
}

@article{IK20,
    author = "Iwaki, Kohei and Kidwai, Omar",
    title = "{Topological recursion and uncoupled BPS structures I: BPS spectrum and free energies}",
    eprint = "2010.05596",
    archivePrefix = "arXiv",
    primaryClass = "math-ph",
    journal = "Adv. Math.",
    volume = "398",
    pages = "108191",
    year = "2022"
}

@article{O21,
    author = "Osuga, Kento",
    title = "{Super Topological Recursion and Gaiotto Vectors For Superconformal Blocks}",
    eprint = "2107.04588",
    archivePrefix = "arXiv",
    primaryClass = "math-ph",
    journal = "Lett. Math. Phys.",
    volume = "112",
    pages = "48",
    year = "2022"
}

@article{BCHORS19,
    author = "Bouchard, Vincent and Ciosmak, Pawe\l{} and Hadasz, Leszek and Osuga, Kento and Ruba, Blazej and Su\l{}kowski, Piotr",
    title = "{Super Quantum Airy Structures}",
    eprint = "1907.08913",
    archivePrefix = "arXiv",
    primaryClass = "math-ph",
    reportNumber = "CALT-2019-025",
    journal = "Commun. Math. Phys.",
    volume = "380",
    number = "1",
    pages = "449--522",
    year = "2020"
}

@article{BO18,
    author = "Bouchard, Vincent and Osuga, Kento",
    title = "{Supereigenvalue Models and Topological Recursion}",
    eprint = "1802.03536",
    archivePrefix = "arXiv",
    primaryClass = "hep-th",
    journal = "JHEP",
    volume = "04",
    pages = "138",
    year = "2018"
}

@article{KS17,
    author = "Kontsevich, Maxim and Soibelman, Yan",
    title = "{Airy structures and symplectic geometry of topological recursion}",
    eprint = "1701.09137",
    archivePrefix = "arXiv",
    primaryClass = "math.AG",
    month = "1",
    year = "2017"
}

@article{ABCO17,
    author = {Andersen, Jorgen Ellegaard and Borot, Ga\"etan and Chekhov, Leonid O. and Orantin, Nicolas},
    title = "{The ABCD of topological recursion}",
    eprint = "1703.03307",
    archivePrefix = "arXiv",
    primaryClass = "math-ph",
    month = "3",
    year = "2017"
}

@article{EO07,
    author = "Eynard, Bertrand and Orantin, Nicolas",
    title = "{Invariants of algebraic curves and topological expansion}",
    eprint = "math-ph/0702045",
    archivePrefix = "arXiv",
    journal = "Commun. Num. Theor. Phys.",
    volume = "1",
    pages = "347--452",
    year = "2007"
}

@article{CE06,
    author = "Chekhov, L. and Eynard, B.",
    title = "{Matrix eigenvalue model: Feynman graph technique for all genera}",
    eprint = "math-ph/0604014",
    archivePrefix = "arXiv",
    journal = "JHEP",
    volume = "12",
    pages = "026",
    year = "2006"
}

@article{CEO06,
    author = "Chekhov, Leonid and Eynard, Bertrand and Orantin, Nicolas",
    title = "{Free energy topological expansion for the 2-matrix model}",
    eprint = "math-ph/0603003",
    archivePrefix = "arXiv",
    journal = "JHEP",
    volume = "12",
    pages = "053",
    year = "2006"
}

@article{BEO13,
    author = {Borot, Ga\"etan and Eynard, Bertrand and Orantin, Nicolas},
    title = "{Abstract loop equations, topological recursion and new applications}",
    eprint = "1303.5808",
    archivePrefix = "arXiv",
    primaryClass = "math-ph",
    journal = "Commun. Num. Theor. Phys.",
    volume = "09",
    pages = "51--187",
    year = "2015"
}

@article{EO12,
    author = "Eynard, Bertrand and Orantin, Nicolas",
    title = "{Computation of Open Gromov\textendash{}Witten Invariants for Toric Calabi\textendash{}Yau 3-Folds by Topological Recursion, a Proof of the BKMP Conjecture}",
    eprint = "1205.1103",
    archivePrefix = "arXiv",
    primaryClass = "math-ph",
    reportNumber = "IPHT-T12-030",
    journal = "Commun. Math. Phys.",
    volume = "337",
    number = "2",
    pages = "483--567",
    year = "2015"
}

@article{EO08,
    author = "Eynard, Bertrand and Orantin, Nicolas",
    title = "{Algebraic methods in random matrices and enumerative geometry}",
    eprint = "0811.3531",
    archivePrefix = "arXiv",
    primaryClass = "math-ph",
    reportNumber = "IPHT-T08-189, CERN-PH-TH-2008-222139",
    month = "11",
    year = "2008"
}

@article{IKT18-2,
  title={Voros coefficients for the hypergeometric differential equations and Eynard--Orantin’s topological recursion: Part II: For confluent family of hypergeometric equations},
  author={Iwaki, Kohei and Koike, Tatsuya and Takei, Yumiko},
  journal={Journal of Integrable Systems},
  volume={4},
  number={1},
  pages={xyz004},
  year={2019},
  publisher={Oxford University Press},
  eprint = "1810.02946",
    archivePrefix = "arXiv",
    primaryClass = "math.CA",
}

@article{IKT18-1,
      title={Voros Coefficients for the Hypergeometric Differential Equations and Eynard-Orantin's Topological Recursion - Part I : For the Weber Equation}, 
      author={Kohei Iwaki and Tatsuya Koike and Yumiko Takei},
      year={2018},
      eprint={1805.10945},
      archivePrefix={arXiv},
      primaryClass={math.CA}
}

@article{KO22,
      title={Quantum curves from refined topological recursion: the genus 0 case}, 
      author={Omar Kidwai and Kento Osuga},
      year={2022},
      eprint={2204.12431},
      archivePrefix={arXiv},
      primaryClass={math.AG}
}

@article{BKMP07,
    author = "Bouchard, Vincent and Klemm, Albrecht and Mari\~{n}o, Marcos and Pasquetti, Sara",
    title = "{Remodeling the B-model}",
    eprint = "0709.1453",
    archivePrefix = "arXiv",
    primaryClass = "hep-th",
    reportNumber = "BONN-TH-2007-07, CERN-PH-TH-2007-153, NEIP-07-03",
    journal = "Commun. Math. Phys.",
    volume = "287",
    pages = "117--178",
    year = "2009"
}

@article{BE12,
    author = "Bouchard, Vincent and Eynard, Bertrand",
    title = "{Think globally, compute locally}",
    eprint = "1211.2302",
    archivePrefix = "arXiv",
    primaryClass = "math-ph",
    reportNumber = "IPHT-T12-114",
    journal = "JHEP",
    volume = "02",
    pages = "143",
    year = "2013"
}

@article{BE16,
  title={Reconstructing WKB from topological recursion},
  author={Bouchard, Vincent and Eynard, Bertrand},
  journal={Journal de l’{\'E}cole polytechnique—Math{\'e}matiques},
  eprint = "1606.04498",
    archivePrefix = "arXiv",
    primaryClass = "math-ph",
  volume={4},
  pages={845--908},
  year={2017}
}

@article{CEM10,
    author = "Chekhov, L. and Eynard, B. and Marchal, O.",
    title = "{Topological expansion of $\beta$-ensemble model and quantum algebraic geometry in the sectorwise approach}",
    eprint = "1009.6007",
    archivePrefix = "arXiv",
    primaryClass = "math-ph",
    reportNumber = "ITEP-TH-34-10",
    journal = "Theor. Math. Phys.",
    volume = "166",
    pages = "141--185",
    year = "2011"
}

@article{M12,
    author = "Marchal, Olivier",
    title = "{One-cut solution of the $\beta$ ensembles in the Zhukovsky variable}",
    eprint = "1105.0453",
    archivePrefix = "arXiv",
    primaryClass = "math-ph",
    journal = "J. Stat. Mech.",
    volume = "1201",
    pages = "P01011",
    year = "2012"
}

@article{BMS10,
    author = "Brini, Andrea and Mari\~{n}o, Marcos and Stevan, Sebastien",
    title = "{The Uses of the refined matrix model recursion}",
    eprint = "1010.1210",
    archivePrefix = "arXiv",
    primaryClass = "hep-th",
    journal = "J. Math. Phys.",
    volume = "52",
    pages = "052305",
    year = "2011"
}

@article{E14,
    author = "Eynard, B",
    title = "{A short overview of the ``Topological recursion''}",
    eprint = "1412.3286",
    archivePrefix = "arXiv",
    primaryClass = "math-ph",
    reportNumber = "IPHT-T14-033-CRM3335",
    month = "12",
    year = "2014"
}

@article{NS11,
  title={Gromov--Witten invariants of $\mathbb{P}^1$ and Eynard--Orantin invariants},
  author={Norbury, Paul and Scott, Nick},
  eprint = "1106.1337",
    archivePrefix = "arXiv",
    primaryClass = "math.AG",
  journal={Geometry \& Topology},
  volume={18},
  number={4},
  pages={1865--1910},
  year={2014},
  publisher={Mathematical Sciences Publishers}
}

@article{C10,
    author = "Chekhov, L.",
    title = "{Logarithmic potential $beta$-ensembles and Feynman graphs}",
    eprint = "1009.5940",
    archivePrefix = "arXiv",
    primaryClass = "math-ph",
    reportNumber = "ITEP-TH-33-10",
    journal = "Proc. Steklov Inst. Math.",
    volume = "272",
    number = "1",
    pages = "58--74",
    year = "2011"
}

@article{FLZ16,
  title={On the remodeling conjecture for toric Calabi-Yau 3-orbifolds},
  author={Fang, Bohan and Liu, Chiu-Chu and Zong, Zhengyu},
  eprint = "1604.07123",
    archivePrefix = "arXiv",
    primaryClass = "math.AG",
  journal={Journal of the American Mathematical Society},
  volume={33},
  number={1},
  pages={135--222},
  year={2020}
}

@article{BKMP08,
    author = "Bouchard, Vincent and Klemm, Albrecht and Mari\~{n}o, Marcos and Pasquetti, Sara",
    title = "{Topological open strings on orbifolds}",
    eprint = "0807.0597",
    archivePrefix = "arXiv",
    primaryClass = "hep-th",
    journal = "Commun. Math. Phys.",
    volume = "296",
    pages = "589--623",
    year = "2010"
}

@inproceedings{NS09,
    author = "Nekrasov, Nikita A. and Shatashvili, Samson L.",
    title = "{Quantization of Integrable Systems and Four Dimensional Gauge Theories}",
    booktitle = "{16th International Congress on Mathematical Physics}",
    eprint = "0908.4052",
    archivePrefix = "arXiv",
    primaryClass = "hep-th",
    reportNumber = "TCD-MATH-09-19, HMI-09-09, IHES-P-09-38",
    pages = "265--289",
    month = "8",
    year = "2009"
}

@article{DOSS12,
    author = "Dunin-Barkowski, P. and Orantin, N. and Shadrin, S. and Spitz, L.",
    title = "{Identification of the Givental formula with the spectral curve topological recursion procedure}",
    eprint = "1211.4021",
    archivePrefix = "arXiv",
    primaryClass = "math-ph",
    journal = "Commun. Math. Phys.",
    volume = "328",
    pages = "669--700",
    year = "2014"
}

@article{BS15,
    author = {Borot, Ga\"etan and Shadrin, Sergey},
    title = "{Blobbed topological recursion: properties and applications}",
    eprint = "1502.00981",
    archivePrefix = "arXiv",
    primaryClass = "math-ph",
    journal = "Math. Proc. Cambridge Phil. Soc.",
    volume = "162",
    number = "1",
    pages = "39--87",
    year = "2017"
}

@article{BBCC21,
    author = {Borot, Ga\"etan and Bouchard, Vincent and Chidambaram, Nitin Kumar and Creutzig, Thomas},
    title = "{Whittaker vectors for $\mathcal{W}$-algebras from topological recursion}",
    eprint = "2104.04516",
    archivePrefix = "arXiv",
    primaryClass = "math-ph",
    reportNumber = "MPIM-Bonn-2021",
    month = "4",
    year = "2021"
}

@article{E16,
    author = "Eynard, B.",
    title = "{Counting Surfaces}",
    journal = "Springer",
    series = "Progress in Mathematical Physics",
    volume = "70",
    year = "2016"
}

@article{MO19,
    author = "Marchal, Olivier and Orantin, Nicolas",
    title = "{Quantization of hyper-elliptic curves from isomonodromic systems and topological recursion}",
    eprint = "1911.07739",
    archivePrefix = "arXiv",
    primaryClass = "math-ph",
    journal = "J. Geom. Phys.",
    volume = "171",
    pages = "104407",
    year = "2022"
}

@article{EGF19,
  title={From topological recursion to wave functions and PDEs quantizing hyperelliptic curves},
  author={Eynard, Bertrand and Garcia-Failde, Elba},
  eprint = "1911.07795",
    archivePrefix = "arXiv",
    primaryClass = "math-ph",
      year={2019},
}

@article{Mil12,
    author = "Milanov, Todor",
    title = "{The Eynard\textendash{}Orantin recursion for the total ancestor potential}",
    eprint = "1211.5847",
    archivePrefix = "arXiv",
    primaryClass = "math.AG",
    journal = "Duke Math. J.",
    volume = "163",
    number = "9",
    pages = "1795--1824",
    year = "2014"
}

@article{CD20,
    author = "Chapuy, Guillaume and Do\l{}\k{e}ga, Maciej",
    title = "{Non-orientable branched coverings, b-Hurwitz numbers, and positivity for multiparametric Jack expansions}",
    eprint = "2004.07824",
    archivePrefix = "arXiv",
    primaryClass = "math.CO",
    journal = "Adv. Math.",
    volume = "409",
    pages = "108645",
    year = "2022"
}

@article{BCD21,
    author = "Bonzom, Valentin and Chapuy, Guillaume and Do\l{}\k{e}ga, Maciej",
    title = "{$b$-monotone Hurwitz numbers: Virasoro constraints, BKP hierarchy, and $O(N)$-BGW integral}",
    eprint = "2109.01499",
    archivePrefix = "arXiv",
    primaryClass = "math.CO",
    journal = {Int. Math. Res. Not.},
    year = {2022},
    month = {07},
}

@article{DDM14,
    author = "Do, Norman and Dyer, Alastair and Mathews, Daniel V.",
    title = "{Topological recursion and a quantum curve for monotone Hurwitz numbers}",
    eprint = "1408.3992",
    archivePrefix = "arXiv",
    primaryClass = "math.GT",
    journal = "J. Geom. Phys.",
    volume = "120",
    pages = "19--36",
    year = "2017"
}

@article{MS15,
    author = "Manabe, Masahide and Su\l{}kowski, Piotr",
    title = "{Quantum curves and conformal field theory}",
    eprint = "1512.05785",
    archivePrefix = "arXiv",
    primaryClass = "hep-th",
    journal = "Phys. Rev. D",
    volume = "95",
    number = "12",
    pages = "126003",
    year = "2017"
}

@article{LN22,
    author = "Lisovyy, O. and Naidiuk, A.",
    title = "{Perturbative connection formulas for Heun equations}",
    eprint = "2208.01604",
    archivePrefix = "arXiv",
    primaryClass = "math-ph",
    journal = "J. Phys. A",
    volume = "55",
    number = "43",
    pages = "434005",
    year = "2022"
}

@article{LN21,
    author = "Lisovyy, O. and Naidiuk, A.",
    title = "{Accessory parameters in confluent Heun equations and classical irregular conformal blocks}",
    eprint = "2101.05715",
    archivePrefix = "arXiv",
    primaryClass = "math-ph",
    journal = "Lett. Math. Phys.",
    volume = "111",
    number = "6",
    pages = "137",
    year = "2021"
}

@article{GS11,
    author = "Gukov, Sergei and Sulkowski, Piotr",
    title = "{A-polynomial, B-model, and Quantization}",
    eprint = "1108.0002",
    archivePrefix = "arXiv",
    primaryClass = "hep-th",
    journal = "JHEP",
    volume = "02",
    pages = "070",
    year = "2012"
}

@article{O23,
author = "Osuga, Kento",
title = "{Deformation and quantisation condition of the $\mathscr{Q}$-top recursion}",
    journal = "To appear",
year = "2023"
}

\end{document}